\documentclass[aps,pra,reprint,superscriptaddress,floatfix]{revtex4-1}
\usepackage{multirow}
\usepackage{braket}
\usepackage{bbm}
\usepackage{bm}
\usepackage{booktabs}
\usepackage{amsmath}
\usepackage{amsfonts}
\usepackage{amsthm}
\usepackage{mathtools}
\DeclarePairedDelimiter{\ceil}{\lceil}{\rceil}
\theoremstyle{definition}
\newtheorem{definition}{Definition}[section]
\newtheorem{theorem}{Theorem}[section]
\newtheorem{lemma}{Lemma}[section]
\usepackage{graphicx}
\usepackage{placeins}
\usepackage{xcolor}
\graphicspath{{./figures/}}
\usepackage[]{hyperref}
\hypersetup{
   pdftitle={Unitary Entanglement Construction in Hierarchical Networks},
    pdfauthor={gorshkovgroup},
    colorlinks=true,            
}
\usepackage[all]{hypcap}
\usepackage[normalem]{ulem}

\newcommand{\abs}[1]{\ensuremath{\left| #1 \right|}}

\newcommand{\vecalpha}{{\vec{\alpha}}}
\newcommand{\vecbeta}{{\vec{\beta}}}
\newcommand{\myvec}[1]{\vec{#1}}

\newcommand{\pr}[0]{^\prime}
\newcommand{\paren}[1]{\left(#1\right)}
\newcommand{\brac}[1]{\left[#1\right]}
\newcommand{\curly}[1]{\left\{#1\right\}}
\newcommand{\mc}[1]{\mathcal{#1}}
\newcommand{\suml}[3]{\sum\limits_{#1}^{#2}#3}

\newcommand{\whp}[0]{\mathbin{\Pi_\alpha}}
\newcommand{\vhp}[0]{\mathbin{\Pi_{\vec{\alpha}}}}
\newcommand{\uhp}[0]{\mathbin{\Pi}}
\newcommand{\tuhp}[0]{\mathbin{\Gamma}}
\newcommand{\twhp}[0]{\mathbin{\Gamma_\alpha}}
\newcommand{\tvhp}[0]{\mathbin{\Gamma_{\vec{\alpha}}}}
\newcommand{\id}[0]{\mathbbm{1}}
\newcommand{\mbb}[1]{\mathbbm{#1}}

\newcommand{\isep}{\mathrel{{.}\,{.}}\nobreak}

\usepackage[capitalise,compress]{cleveref}
\crefname{section}{Sec.}{Secs.}
\Crefname{section}{Section}{Sections}
\crefrangelabelformat{equation}{\textup{(#3#1#4)}--\textup{(#5#2#6)}}

\begin{document}
\title{Unitary Entanglement Construction in Hierarchical Networks}
\date{\today}
\author{Aniruddha Bapat}
\affiliation{Joint Center for Quantum Information and Computer Science, NIST/University of Maryland, College Park, MD 20742, USA}
\affiliation{Joint Quantum Institute, NIST/University of Maryland, College Park, MD 20742, USA}
\author{Zachary Eldredge}
\affiliation{Joint Center for Quantum Information and Computer Science, NIST/University of Maryland, College Park, MD 20742, USA}
\affiliation{Joint Quantum Institute, NIST/University of Maryland, College Park, MD 20742, USA}
\author{James R. Garrison}
\affiliation{Joint Center for Quantum Information and Computer Science, NIST/University of Maryland, College Park, MD 20742, USA}
\affiliation{Joint Quantum Institute, NIST/University of Maryland, College Park, MD 20742, USA}
\author{Abhinav Deshpande}
\affiliation{Joint Center for Quantum Information and Computer Science, NIST/University of Maryland, College Park, MD 20742, USA}
\affiliation{Joint Quantum Institute, NIST/University of Maryland, College Park, MD 20742, USA}
\author{Frederic T. Chong}
\affiliation{Department of Computer Science, University of Chicago, Chicago, USA}
\author{Alexey V. Gorshkov}
\affiliation{Joint Center for Quantum Information and Computer Science, NIST/University of Maryland, College Park, MD 20742, USA}
\affiliation{Joint Quantum Institute, NIST/University of Maryland, College Park, MD 20742, USA}

\begin{abstract}
	The construction of large-scale quantum computers will require modular architectures that allow physical resources to be localized in easy-to-manage packages. In this work, we examine the impact of different graph structures on the preparation of entangled states. We begin by explaining a formal framework, the hierarchical product, in which modular graphs can be easily constructed. This framework naturally leads us to suggest a class of graphs, which we dub hierarchies. We argue that such graphs have favorable properties for quantum information processing, such as a small diameter and small total edge weight, and use the concept of Pareto efficiency to identify promising quantum graph architectures. We present numerical and analytical results on the speed at which large entangled states can be created on nearest-neighbor grids and hierarchy graphs. We also present a scheme for performing circuit placement -- the translation from circuit diagrams to machine qubits --  on quantum systems whose connectivity is described by hierarchies.\end{abstract}

\maketitle

\section{Introduction}
As quantum computers grow from the small, few-qubit machines currently deployed to the large machines required to realize useful, fault-tolerant computations, it will become increasingly difficult for every physical qubit to be part of a single contiguous piece of hardware. Just as modern classical computers do not rely on a single unit of processing and memory, instead using various components  such as CPUs, GPUs, and RAM, we expect that a quantum computer will likewise use specialized modules to perform different functions. At a higher level, computers can be organized into clusters, data centers, and cloud services which allow for a distributed approach to computational tasks, another paradigm quantum computers will no doubt emulate. Already, there has been significant interest in how quantum algorithms for elementary operations such as arithmetic perform in distributed-memory situations \cite{VanMeter2005,Meter2008} and how to automate the design of quantum computer architectures \cite{Ahsan2015}. In addition, the construction of a fault-tolerant quantum computer naturally suggests a separation of physical qubits into groups corresponding to logical qubits, which makes modularity an attractive framework for building fault-tolerant computers \cite{Metodi2005}. Modular and scalable computing architectures have been explored for both ion trap \cite{Duan2010,Monroe2013} and superconducting platforms \cite{Devoret2013, Brecht2016, Kurpiers2017}. 

In this paper, we use tools from graph theory to discuss benefits and drawbacks of different potential architectures for a modular quantum computer. A graph-theoretic approach allows us to flexibly examine a wide range of possible arrangements quantitatively and allows for convenient numerical simulation using existing software packages designed for network analysis \cite{NetworkX}. We especially wish to focus on families of graphs that can scale with the desired number of qubits.
In general, we assume that connectivity, i.e., being able to quickly perform operations between nodes, is desirable in an architecture, but that building additional graph edges is in some way costly or difficult, and so will try to minimize the number of needed edges to achieve a highly communicative graph.

We will examine the performance of graphs in generating large entangled states such as the multi-qubit Greenberger-Horne-Zeilinger (GHZ) state (also known as a cat state). 
The GHZ state has perfect quantum correlations between different qubits; it thus can be used to perform high-precision metrology \cite{Bollinger1996,Eldredge2016}. In addition, the creation of a GHZ state can be used as part of a state-transfer protocol, which may be useful as part of large quantum computations \cite{Eldredge2017}.

An additional property of GHZ state preparation and state transfer which makes them a useful starting point is that, in nearest-neighbor connected systems, performing these tasks using unitary processes from an initial product state is limited by the Lieb-Robinson bound \cite{Bravyi2006a, Bentsen2018}. It takes a time proportional to the distance between two points to establish maximal quantum correlation between them. By examining these tasks on a range of different graphs, we hope to understand how the graph structure can affect the limitations on quantum processes caused by locality considerations. Prior work has characterized the difficulty of creating graph states \cite{Cuquet2012}, but preparation of such states is not limited by Lieb-Robinson considerations.

Our work in this paper should be contrasted with work on entanglement percolation \cite{Acin2007, Perseguers2013}. Entanglement percolation describes the process of using low-quality entanglement between adjacent nodes on a graph to create one unit of long-range, high-quality entanglement (e.g., a Bell pair). The use of entanglement percolation to prepare large cluster states on a lattice was considered in Ref.~\cite{Kieling2007}. The nature of entanglement growth in complex networks was considered in Refs.~\cite{Cuquet2009, Cuquet2011}, showing that so-called ``scale-free'' networks are particularly easy to produce large entangled states in. We are interested in the overall capability of different graph structures to perform large computations and in the use of  graph eigenvalue methods to understand the spread of quantum information \cite{YangWang2003}. GHZ state preparation and state transfer are just two possible benchmark tasks, and it is possible that other tasks would result in different evaluations of relative performance between graphs.

Our work should also be considered in the context of classical network theory, where much is known about complicated graph structures \cite{Watts1998, Barabasi1999, Albert2002}. It remains to be seen to what degree classical network theory can be easily exported to the quantum domain. Quantum effects such as the no-cloning theorem may limit our ability to distribute information, or conversely we can take advantage of teleportation to distribute quantum bandwidth in anticipation of it actually being needed. As further examples of how quantum and classical networks differ, it has been shown that entanglement swapping may be used to permit quantum networks to reshape themselves into interesting and useful topologies \cite{Perseguers2010}. It has also been shown that, in general, the optimal strategy for entanglement generation in quantum networks can be difficult to calculate because many aspects of classical control theory do not apply
\cite{DiFranco2012}.

The structure of this paper is as follows. In \cref{sec:hp}, we will introduce a binary operation on graphs known as the hierarchical product, describe how it can be used to produce families of graphs we call hierarchies, and discuss the properties of these hierarchies. In \cref{sec:comparison}, we will compare hierarchies to other families of graphs, examining how certain graph-theoretic quantities scale with the total number of included qubits. In \cref{sec:construction}, we will use analytic and numerical methods to examine how long is required to construct GHZ states spanning our graphs or to transfer states across them. Finally, in \cref{sec:placement}, we will show how the unique structure of hierarchies allows for simple heuristics to map qubits in an algorithm into physical locations in hardware.

\section{Hierarchical Products of Graphs}
\label{sec:hp}
\subsection{Background and Notation}
One of the defining features of modularity in a network is the presence of clusters of nodes that are well-connected. Qualitatively, a modular network can be partitioned into such node clusters, or \emph{modules}, that have a sparse interconnectivity. In quantum networking, it is believed that fully connected architectures will suffer greatly decreasing performance or increasing costs as the number of nodes becomes larger, and this motivates the search for alternative network designs. For instance, Ref.~\cite{Monroe2014} estimates that a single module of trapped-ion qubits will likely contain no more than $10-100$ ions, noting that the speed at which gates are possible becomes slower as the module is expanded. On the network scale, we might imagine a network of nodes over longer distances connected by quantum repeaters \cite{Briegel1998}. In such a network, establishing direct links between every possible pair of $N$ nodes would require $\Theta(N^2)$ sets of quantum repeaters, a prohibitive cost as $N$ becomes large.

The state of the art in quantum technologies, such as ion traps and superconducting qubits, is the ability to control a small number ($\sim 10-100$) of physical qubits using certain fixed sets of one- and two-qubit operations. Instead of increasing the size of these modules, one could instead build a network out of many small modules that are connected at a higher level in a sparse way, perhaps by optical communication links \cite{Monroe2014}.

Our first goal will be to describe modular architectures in the language of graph theory. This will then allow us to quantify and compare their connectivity properties against other network designs, notably the nearest-neighbor grid architecture.

An unweighted graph $G = (V,E)$ is conventionally specified by a set of vertices $V$, and a set of edges between the vertices $E$, where an edge between distinct vertices $i$ and $j$ will be denoted by the pair $(i,j)$. In this paper, we use the terms ``vertex'' and ``node'' synonymously. The \emph{order} of a graph is the total number of vertices in the graph, $|V|$. It will be useful for the purposes of this paper to work with \emph{weighted} graphs, where we specify a weight $w_{ij}\in\mbb{R}$ for each pair of vertices $\paren{i,j}\in V\times V$. Two vertices $i$ and $j$ are said to be \emph{disconnected} if $w_{ij}= 0$, and connected by an edge with weight $w_{ij}\ne 0$ otherwise. Thus, unweighted graphs may be thought of as graphs with unit weight on every edge. 

Finally, the graphs we consider here will be \emph{simple}, meaning: 
\begin{itemize}
\item The edges have no notion of direction. In other words, $w_{ij} = w_{ji}$ for all $i,j\in V$. 
\item There are no self-edges, i.e., $w_{ii}=0$ for all $i\in V$.
\item Any two vertices have at most one edge between them.
\end{itemize}
Henceforth, graphs will be simple and weighted, unless otherwise specified.

The information contained in a graph can be represented as a matrix known as the \emph{adjacency matrix}, whose rows and columns are labeled by the vertices in $V$ and whose entries hold edge weights. Thus, the adjacency matrix is an $n\times n$ matrix where $|V|=n$. The adjacency matrix $A_G$ (or simply $A$ for shorthand) for a graph $G$ is given by
\begin{equation}
A_{ij} = \begin{cases}
0, &\text{if } i=j,\\
w_{ij} ,& \text{if } i\ne j.\\
\end{cases}
\end{equation}
An important measure of local connectivity is given by the \emph{valency} $v_i$ of a node $i$, with $v_i = \suml{j=1}{n}{w_{ij}}$. For unweighted graphs, the valency of any node is simply the number of edges incident at that node, otherwise known as the \emph{degree} of the node. We will also define the graph diameter, $\delta(G)$, as the maximization of the shortest distance between two nodes on the graph over all pairs of nodes. 

Graphs may also be described by the \emph{Laplacian}. The algebraic Laplacian $L$ is given by
\begin{equation}
L_{ij} = \begin{cases}
v_i,&\text{if } i=j,\\
-w_{ij},& \text{if } i\ne j. \\
\end{cases}
\end{equation}
The algebraic Laplacian is closely related to the adjacency matrix, since we may write $L = \Delta - A$, where $\Delta = \text{diag}\paren{v_1,\ldots,v_n}$ is the diagonal matrix of vertex valencies. The eigenvalues of the algebraic Laplacian give us bounds on various graph properties, as discussed further in Sec.~\ref{sec:spectral}.

Finally, we remark that the algebraic Laplacian should not be confused with the normalized Laplacian $\mc{L} = \Delta^{ -\frac{1}{2}}L\Delta^{-\frac{1}{2}}$, which is frequently seen in the network theory literature. The algebraic properties discussed in the next section (such as associativity of the hierarchical product) apply to the adjacency matrix as well as the algebraic Laplacian, but not to the normalized Laplacian.

\subsection{Hierarchical Product}
\label{sec:hpdef}
Here, we will define the hierarchical product and illustrate it with simple examples. For a fuller exposition, see Ref.~\cite{Barriere2009}, where the hierarchical product of graphs was introduced. Note that,  in some contexts, the hierarchical product is also known as the rooted product \cite{Godsil1978}. 

Given a graph $G$, let $\id_G$ denote the identity matrix on $n=|V|$ vertices. We will denote by $D_G$ an $n\times n$ diagonal matrix with 1 as the first entry and zero everywhere else. Note that there is no natural notion of order to graph vertices, so the choice of ``first'' vertex must be specified explicitly. Graphs with such a specified first vertex are called \textit{rooted graphs} \cite{Harary1955}. We write these matrices as
\begin{align}
\id = \begin{pmatrix}1 & & & & \\ & 1 & & & \\ & & 1 & & \\ & & & \ddots & \\ & & & & 1\end{pmatrix},\ \ \  D = \begin{pmatrix}1 & & & & \\ & 0 & & & \\ & & 0 & & \\ & & & \ddots & \\ & & & & 0\end{pmatrix}.
\end{align}
\begin{definition}
\label{def:ghp}
Given graphs $G$ and $H$, the \emph{hierarchical product} $P = G\uhp H$ is the graph on vertices $V_P = V_G\times V_H$ and edges $E_P \subseteq V_P\times V_P$ specified by the adjacency matrix
\begin{equation}
	A_P = A_G \otimes D_H + \id_G \otimes A_H \label{eq:ghp_adj},
\end{equation}
      or, equivalently, by the algebraic Laplacian
\begin{equation}
	L_P = L_G \otimes D_H + \id_G \otimes L_H .
\end{equation}
We will often use the shorthand $A_P = A_G\uhp A_H$ and $L_P = L_G\uhp L_H$. 
\end{definition}

\begin{figure}[tb]
  \centering
  \includegraphics[width = .45\textwidth]{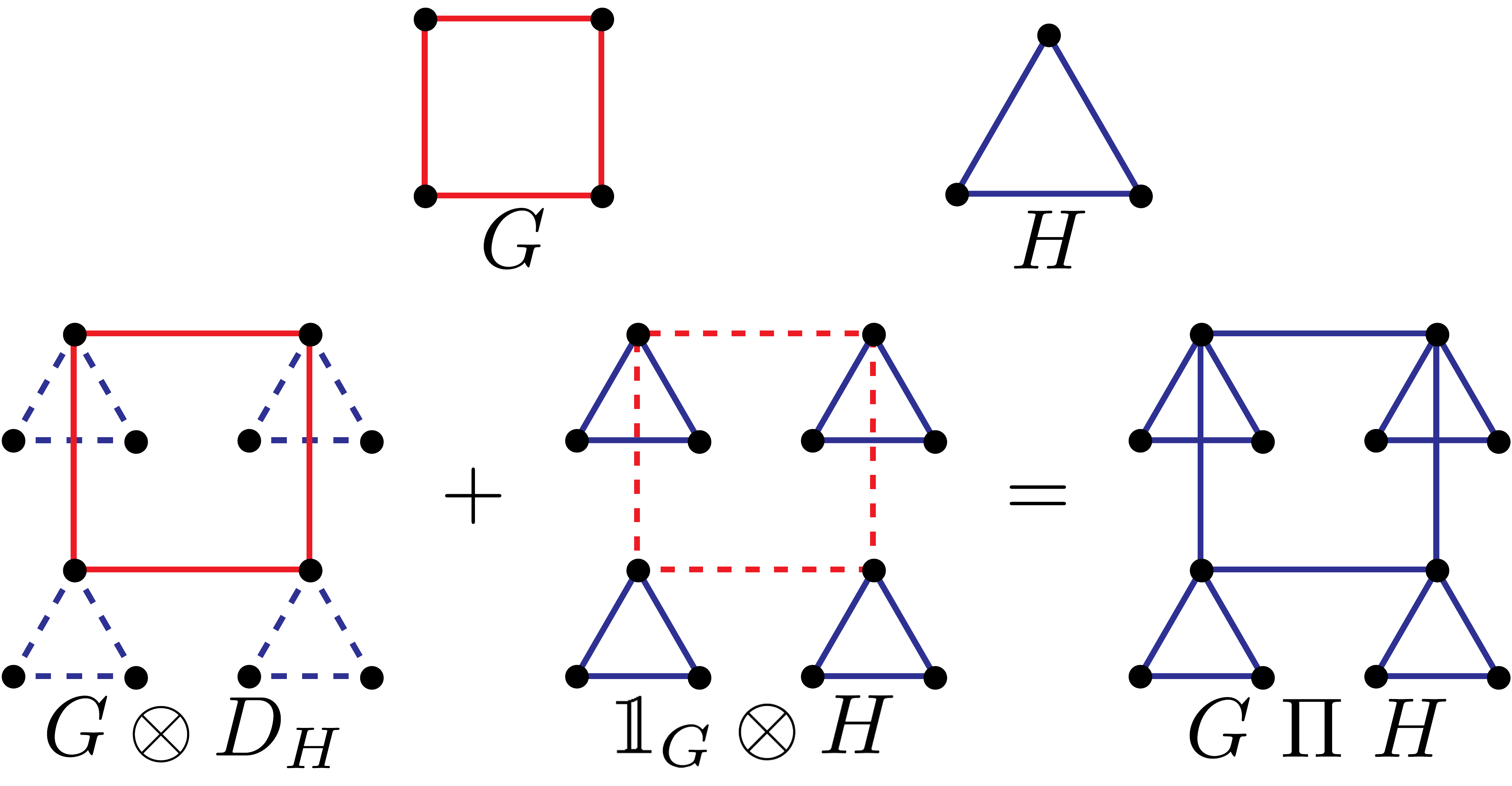}
  \caption{A simple example of the hierarchical product $G\uhp H$ between the cycle graphs $G= C_4$ and $H= C_3$.  The first term in \cref{eq:ghp_adj}, $A_G\otimes D_H$,  creates one copy of $G$ on the vertex set formed by the first vertices of each $H$ copy, while the second term $\id_G\otimes A_H$ creates the four copies of $H$.}
  \label{fig:hp}
\end{figure}

If $G$ and $H$ are graphs, then $G \uhp H$ may be thought of as one copy of $G$ with $\abs{G}$ copies of $H$, each attached to a different vertex of $G$ (see Fig.~\ref{fig:hp}).
Thus, $G\uhp H$ is a graph which has $\abs{G}$ \emph{modules} of $\abs{H}$ nodes each. The modules' internal connectivity is described by $H$, and the modules are connected to one another in a manner described by $G$. The hierarchical product formalism therefore naturally produces modular graphs. Its main advantage comes from the convenience of working with the algebra at the level of adjacency matrices and Laplacians, which in turn makes the computation of important properties of such graphs straightforward. 

We now present some properties of the hierarchical product which make it an attractive formalism for practical applications in quantum networking.

\subsubsection{Structural Properties} 
\label{sec:struc}
At the level of adjacency matrices, the hierarchical product is \emph{associative}. Let $A, B, C$ be three adjacency matrices. Then,
\begin{equation}
  \paren{A\uhp B}\uhp C = A\uhp\paren{B\uhp C}.
\end{equation}
For a proof, we refer the reader to Ref.~\cite{Barriere2009}.

Associativity implies that a product of multiple graphs does not depend on the order of evaluation. Therefore, we can unambiguously take the hierarchical product over many graphs to produce a graph of the form $G_k\uhp G_{k-1}\uhp\cdots\uhp G_1$. We will refer to such graphs as \emph{hierarchies}, and the $i$-th graph in the product $G_i$ as the $i$-th level of the hierarchy, enumerated from the bottom level upwards (symbolically, from right to left). In particular, if all $G_i$ are equal to some graph $G$, then we write 
\begin{align}
\label{eq:unwHPcomp}
G^{\uhp k} := \underbrace{G \uhp \cdots G\, \uhp}_{k-1\mathrm{\ times}} G.
\end{align} 
and refer to $G^{\uhp k}$ as a depth-$k$ (or $k$-level) hierarchy. 

Note that the hierarchical product \emph{does not} satisfy many properties which are commonly assumed for operations on matrices. In particular,
\begin{enumerate}
\item Bilinearity:  $\paren{A_1+A_2}\uhp B = A_1\otimes D_B + A_2\otimes D_B + \id_{(A_1+A_2)}\otimes B \neq A_1\uhp B + A_2\uhp B$. Similarly, $A\uhp\paren{B_1+B_2}\neq A\uhp B_1 + A\uhp B_2$.
\item Scalar multiplication: For any scalar $\alpha$, $\paren{\alpha A}\uhp B = \alpha A\otimes D_B + \id_A\otimes B \neq \alpha\paren{A\uhp B} \neq A\uhp \paren{\alpha B}$. Note however that scalar multiplication is distributive in the following way: $\alpha\paren{A\uhp B} = \paren{\alpha A}\uhp \paren{\alpha B}$.
\end{enumerate}

Hierarchical graphs are also instances of hyperbolic graphs. The Gromov-hyperbolicity \cite{Gromov2007}, which measures curvature and is small for a graph with large negative curvature, is only a constant for hierarchical graphs. Since the hyperbolicity in general is at most half the graph diameter, whereas in this case it is independent of the diameter, it is termed \emph{constantly  hyperbolic} in the parlance of Ref.~\cite{Chen2013}. Hyperbolic graphs are seen in several real-world complex networks \cite{Lohsoonthorn2003, Krioukov2010}, most notably the internet \cite{Montgolfier2011, Boguna2010}. Hyperbolic lattices have also been realized recently in superconducting circuits \cite{Kollar2018}.

Finally, hierarchies have low tree-, clique- and rank-widths, which are each measures of the decomposibility of a graph \cite{OUM200579}. These structural properties imply efficient algorithms for optimization problems expressible in monadic second-order (MSO) logic -- a class which, for arbitrary graphs, includes several NP-hard problems. This feature could potentially be used to solve circuit layout and optimization problems on modular architectures without resorting to heuristics. We refer the reader to Ref.~\cite{Courcelle1998} for details on these structural results. 

\subsubsection{Scalability} 

So far we have discussed hierarchies in which the edges in different levels of the hierarchy are equally weighted. However, one useful generalization would be to allow the weight of edges at each layer of the hierarchy to vary. The meaning of this weight could vary depending on the context. In some cases, weights can be used to quantify the costs of an edge (\textit{cost weight}). In others, we may wish to use weighted edges to quantify the power or performance of a network, interpreting edge weights as the strength of terms in a Hamiltonian or, inversely, the time required to communicate between nodes (\textit{time weight}). 

In this work, we prefer to remain agnostic to the meaning of the weights as much as is possible. When we calculate graph properties in Sec.~\ref{sec:comparison}, we will do so without reference to the meaning of the weights. 
In general, we will allow a graph to assign multiple kinds of weights to its edges, and each type of weight might scale differently. For now, we define a generalization of the hierarchical product which will allow us to construct hierarchies that incorporate different weights at different levels of the hierarchy.
\begin{definition}
\label{def:hp}
Given graphs $G$ and $H$, and $\alpha\in\mathbb{R_+}$, the \emph{$\alpha$-weighted hierarchical product} $P = G\whp H$ is a graph on vertices $V_P = V_G\times V_H$ and edges $E_P \subseteq V_P\times V_P$ specified by the adjacency matrix
\begin{equation}  
  A_P = \alpha A_G \otimes D_H + \id_G \otimes A_H,
\end{equation}
or, equivalently, by the algebraic Laplacian
\begin{equation}
  L_P = \alpha L_G \otimes D_H + \id_G \otimes L_H .
\end{equation}
We will often use the shorthand $A_P = A_G\whp A_H$, and $L_P = L_G\whp L_H$. 
\end{definition}

As before, we may construct a $k$-level, \emph{weighted} hierarchy out of $k$ base graphs $G_1, \ldots, G_k$, and $k$ weights $\alpha_i,\ldots, \alpha_k\equiv \vec{\alpha}$, so that the edges of the $i$-th level graph $G_i$ are weighted by the $i$-th component of $\vec{\alpha}$, $\alpha_i$. The adjacency matrix of such a hierarchy may be written as
\begin{equation}
\label{eq:hierarchysum}
A^{\vhp k} := \suml{i=1}{k}{\alpha_i \id_{[i+1 \isep k] }\otimes A_i \otimes D_{[1 \isep i-1]}},
\end{equation}
where the subscripts $[a\isep b]$ on $\id$ and $D$ are shorthand for the Kronecker product of matrices over all descending indices in the integer interval $[a\isep b]$. For instance, $D_{[1\isep i-1]} := D_{G_{i-1}}\otimes D_{G_{i-2}}\otimes \cdots \otimes D_{G_1}$. 

Defined as above, a weighted hierarchy $G^{\vhp k}$ is uniquely and efficiently specified by a real vector of weights $\vec{\alpha}\in\mbb{R_+}^k$ and an ordered tuple of graphs $\paren{G_1,\ldots, G_k}$. It will be the case that our analyses are unaffected by an overall scaling of the weight vector, so that one may identify $\vec{\alpha} \equiv c\vec{\alpha}$ for any real scalar $c$. As convention, we will always normalize by setting $\alpha_1=1$, which corresponds to assigning a unit-weight multiplicative factor to the lowest-level graphs in the hierarchy. 

We can construct the adjacency matrix of the graph $G^{\vhp k}$ by repeated application of the two-fold product (Def.~\ref{def:hp}) in some well-defined way, analogous to \cref{eq:unwHPcomp}. However, unlike before, the weighted product is non-associative, so we must first define an order of operations for manifold weighted products. Unless otherwise specified, we will always evaluate a manifold product from \emph{right to left}, which corresponds to building the hierarchies from the bottom up, and is required in order to ensure that this definition matches \cref{eq:hierarchysum}. For example, in the 3-fold product $A_3\uhp_{\alpha_3} A_2 \uhp_{\alpha_2} (\alpha_1A_1)$, we will first evaluate the product $A_2 \uhp_{\alpha_2} (\alpha_1 A_1)$, and then take the product of $A_3$, weighted by $\alpha_3$, with the resulting graph. The final result is
\begin{equation} \label{eq:explicit_kron}
  \alpha_3A_3\otimes D_2\otimes D_1 + \alpha_2\id_3\otimes A_2\otimes D_1 + \alpha_1\id_3\otimes\id_2\otimes A_1.
\end{equation}
In fact, a $k$-fold product, when evaluated this way, matches the right hand side of \cref{eq:hierarchysum}. Therefore, the $k$-level weighted hierarchy can also be written unambiguously as 
\begin{equation}
\label{eq:hpcomposition}
A^{\vhp k} = A_k \mathbin{\Pi_{\alpha_{k}}} A_{k-1} \mathbin{\Pi_{\alpha_{k-1}}} \cdots \mathbin{\Pi_{\alpha_{2}}} (\alpha_1 A_1).
\end{equation}
Henceforth, the weight $\alpha_1$, which scales the lowest-level adjacency matrix $A_1$, will be dropped due to our normalization choice of $\alpha_1 = 1$. 

An important class of hierarchy graphs is one where the level weights follow a geometric progression of weights, i.e., $\alpha_i = \alpha^{i-1}$. We will denote such hierarchies by $G^{\whp k}$, where the scalar subscript $\alpha$ will be understood to mean the mutual weighting between successive hierarchies. For $\alpha > 1$, this leads to a ``fat tree'' structure, while for $\alpha< 1$, we instead get a ``skinny tree'' for which the edge weights decrease between consecutive levels from the leaves to the root. These constructions are illustrated in Fig.~\ref{fig:trees}, and mentioned because fat trees are known to be a commonly used architecture in classical networks \cite{Leiserson1985}.

\begin{figure}[htpb]
	\centering
	\includegraphics[width=.45\textwidth]{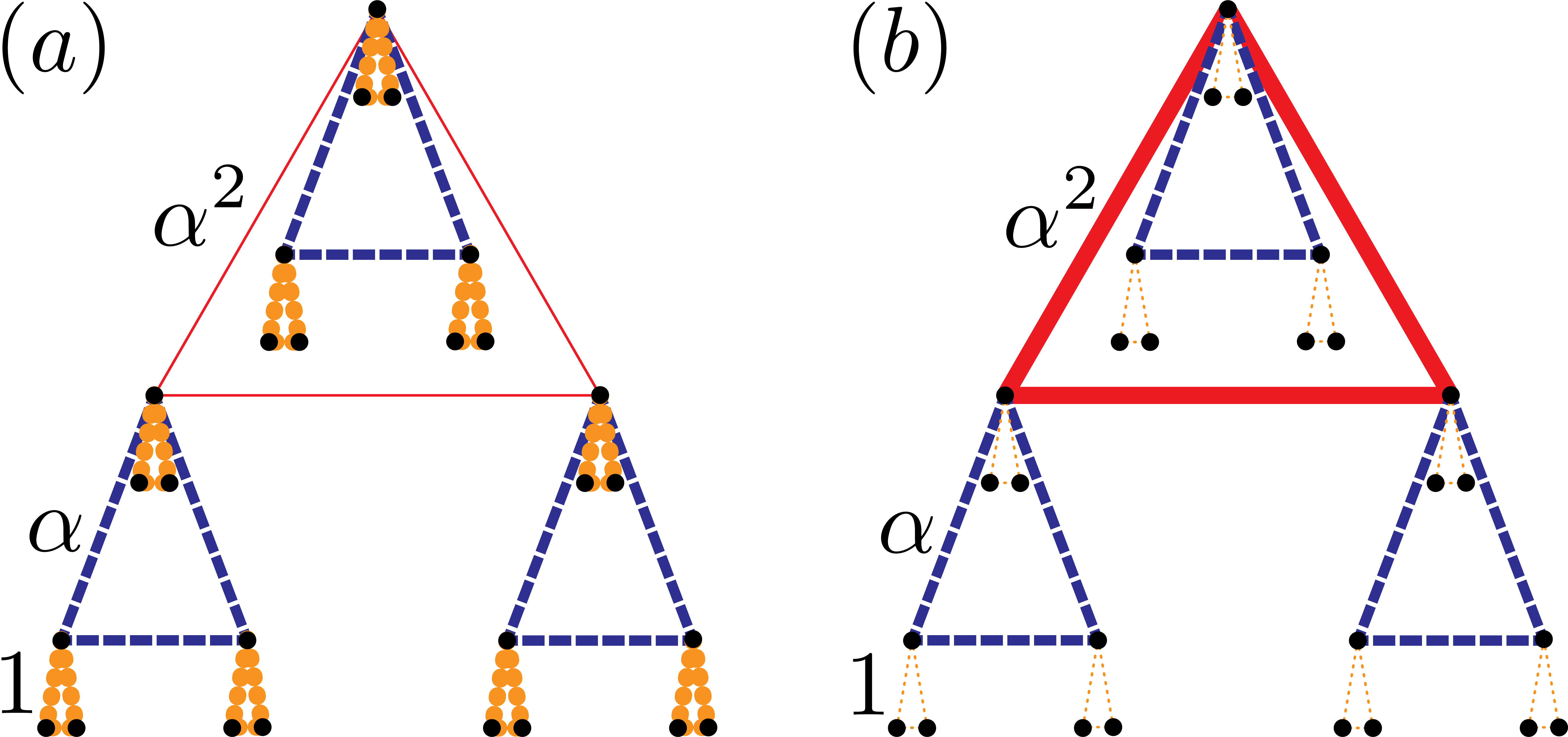}
	\caption{An illustration of the use of the hierarchical product to produce (a) ``skinny'' and (b) ``fat'' trees. In each case, the hierarchy $K_3^{\whp 3}$ is drawn, with the thickness of edges illustrating the weight of those edges. Depending on whether $\alpha < 1$ or $\alpha > 1$, this can lead to either lower-weighted high-level edges as in (a) or higher-weighted ones as in (b). Note that, for ease of visualization, here we break the usual convention of taking the lowest-level edges as unit weight.} 
	\label{fig:trees}
\end{figure}

Allowing a clear separation of the modular system into hierarchical levels, each of which can be assigned unique edge weight, enables straightforward discussion of computation that occurs both within and between modules in a unified framework. When two nodes interact, we can assign this a cost that depends on the edges between them.

\begin{figure}[htbp]
  \centering
  \includegraphics[width=.45\textwidth]{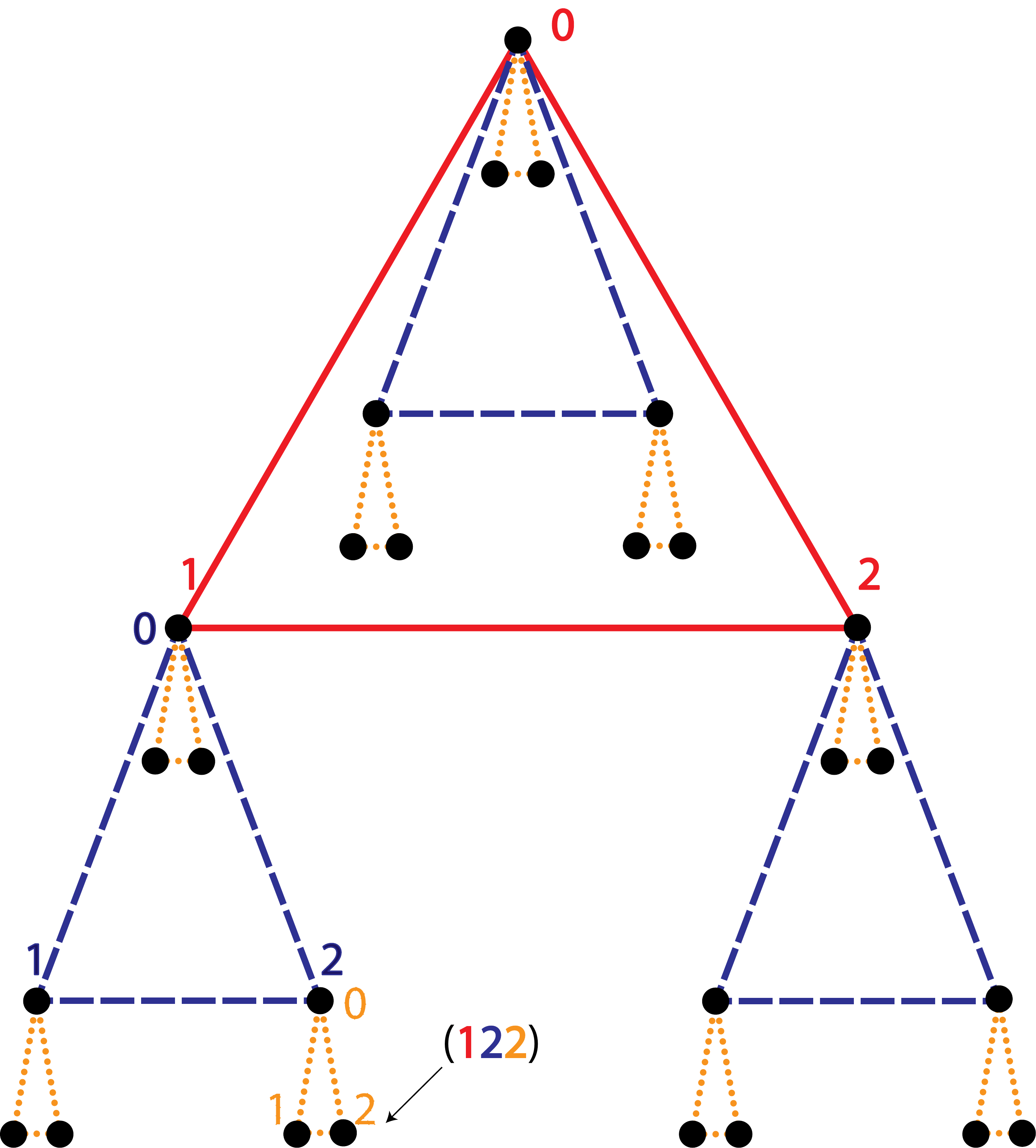}
  \caption{Addressing nodes in the hierarchy, layer by layer. Shown is a three-level hierarchy with the triangle graph $K_3$ as its base. Each vertex is represented as a 3-digit number in base 3. The first digit points to a node at the top level (red solid triangle), the second to a location in the second level (blue dashed triangle), and finally, the last digit (yellow dotted triangle) specifies the node location completely.}
  \label{fig:hplayers}
\end{figure}
\subsubsection{Node Addressal} 
\label{sec:addressal}
A hierarchy on $N$ nodes gives a natural labeling of the nodes. Suppose the hierarchy $H$ contains $k$ levels and each level is described by a graph $G$ with $|G|=n$ nodes, where $n^k=N$. Label the vertices of $G$ by indices $j=0,1,\ldots,n-1$. Then, the adjacency matrix $\id_G\otimes G$ (which corresponds to $n$ disjoint copies of $G$) has vertices which may be labeled as $(jk)$, where $j,k=0,1,\ldots n-1$. The first label identifies which copy of $G$ the node occurs in, while the second identifies where in $G$ it appears. The same vertex labeling can then be used for the 2-level hierarchy $G\uhp G$. In this manner, the $k$-level hierarchy has $n^{k}$ vertices with labels of the form $(b_1b_2\cdots b_k)$, where $b_i\in\curly{0,1,\ldots, n-1}$ for all $i$. This is essentially a $k$-digit, base-$n$ representation of numbers from $0$ to $N=n^k - 1$, as illustrated in Fig.~\ref{fig:hplayers}.

This node addressal scheme allows for each node to be uniquely identified in a way that simultaneously describes its connectivity to other nodes and allows for easy counting of how many nodes lie in either the entire graph or in particular subgraphs. This addressal scheme will be important for describing a variant of hierarchies in Sec.~\ref{sec:flexible} and for implementing the graphs in software, e.g.\ as used to generate the numerical results in Sec.~\ref{sec:probabilistic}.

\subsubsection{Spectral Properties} 
\label{sec:spectral}
One of the tools frequently used in analyzing large networks is the spectral decomposition of the Laplacian. The behavior of the largest eigenvalue, the first eigenvalue gap, and the distribution of eigenvalues as a function of the network parameters are some of the diagnostics that can provide key information about dynamical processes on the network, and can also be used as points of comparison between competing network topologies \cite{Newman2000}.

The smallest eigenvalue of a Laplacian is always $\lambda_1=0$, which corresponds to the uniform eigenvector $\myvec{e}_1 = \paren{1, 1, \ldots, 1}$. In ascending order, the eigenvalues of $L$ may be denoted by $0=\lambda_1 \le \lambda_2 \le \cdots \le \lambda_N$. We now state some graph properties that can be related to the spectrum of $L$ \cite{Newman2000,Mohar1991}.

The second eigenvalue $\lambda_2$ is known as the \emph{algebraic connectivity} of the graph and is closely related to the expansion and connectivity properties of the graph. Broadly, the larger the value of $\lambda_2$, the better the connectivity of the network. To illustrate this point, consider the graph diameter, $\delta(H)$, which can be bounded using $\lambda_2$ as follows:
\begin{equation}
\label{eq:diambound}
  \frac{4}{N\lambda_2} \le \text{diam}\paren{H} \le 2\ceil[\Big]{\frac{\Delta+\lambda_2}{4\lambda_2}\ln\paren{N-1}},
\end{equation}
where $\Delta$ is the maximum degree of $H$. It can be seen that a larger value for $\lambda_2$ will lead to a smaller graph diameter. We also have the following asymptotic bound on the mean distance between nodes, $\bar{\rho}(H)$:
\begin{equation}
\label{eq:mdbound}
  \frac{2}{(N-1)\lambda_2(H)} + \frac{1}{2} \lesssim \bar{\rho}(H) \lesssim \ceil[\Big]{\frac{\Delta+\lambda_2}{4\lambda_2}\ln\paren{N-1}}.
\end{equation}
Another important diagnostic of a network is given by the \emph{Cheeger constant} $h(H)$ \cite{Chung1997}, also called the isoperimetric number or the graph conductance. This graph invariant is a measure of how difficult the graph is to disconnect by cutting edges. For a connected graph, this number is always positive. As benchmark values, the complete graph $K_N$ has Cheeger constant $N/2$ while a cycle graph $C_N$ has Cheeger constant $4/N$. The relationship between $\lambda_2$ and $h(H)$ can be seen through the following bounds:
\begin{equation}
\label{eq:cheegerbound}
  \frac{\lambda_2}{2} \le h(H) \le \sqrt{\lambda_2\paren{2 \Delta - \lambda_2}}.
\end{equation}
Many other graph properties may be derived from the Laplacian spectrum as well (see, e.g., Refs.~\cite{Newman2000, Mohar1991}).

For a large network, finding the eigenvalues can be numerically expensive. However, hierarchies have a special structure which can be exploited for the evaluation of graph spectra. Here, we show (in Theorem \ref{th:specMain}) that if the spectra of the base graphs $L_i$ are known, then one can derive the spectrum of the $k$-level hierarchy efficiently using a recursive procedure. We first present two lemmas. The first lemma generalizes Theorem 3.10 from Ref.~\cite{Barriere2009}, which states that the characteristic polynomial $\phi_P(x)$ $\paren{= \det\brac{x\id - P}}$ of an unweighted hierarchical product of adjacency matrices $A$, $B$ is given by 
\begin{equation} \label{eq:thm310}
  \phi_P(x) = \phi_{B\pr}\paren{x}^{n_A}\phi_A\paren{\frac{\phi_{B}\paren{x}}{\phi_{B\pr}\paren{x}}},
\end{equation}
where $A\pr$ (resp.\ $B\pr$) is the matrix $A$ (resp.\ $B$) with the first row and first column removed, and $n_A = |G_A|$ is the order of the graph $A$. In fact, \cref{eq:thm310} applies to Laplacians as well as adjacency matrices. The lemma below further generalizes this statement to a weighted product of Laplacians.
\begin{lemma}
\label{lm:spec1}
  Let $K$ and $L$ be two graph Laplacians with characteristic polynomials given by $\phi_{K}(x)$ and $\phi_L(x)$, respectively. Then, the characteristic polynomial $\phi_{\uhp}(x)$ of the hierarchical product $K \uhp_\alpha L$  is given by
\begin{equation}
\label{eq:HPchar}
  \phi_{\uhp} \paren{x} = \brac{\alpha \phi_{L\pr}\paren{x}}^{n_K}\phi_K\paren{\frac{1}{\alpha}\frac{\phi_{L}\paren{x}}{\phi_{L\pr}\paren{x}}},
\end{equation}
where $n_k = \dim \curly{K}$, and $L\pr$ is defined similar to $A\pr$ and $B\pr$ above.
\end{lemma}

\begin{proof}

Denote the spectra of $K$ and $L$ by $\curly{\kappa_j}$ and $\curly{\lambda_j}$, respectively. Recall that the $\alpha$-weighted hierarchical product may be written as
\begin{equation}
  K\whp L = \alpha K\otimes D_L + \id_K \otimes L.
\end{equation}
If $U_K$ is a unitary that diagonalizes $K$, we conjugate the above equation with the unitary $U_K\otimes\id_L$, and look at the resulting block matrix. Each block corresponds to an eigenvalue of $K$, and thus the $j$-th block is given by  $\alpha\kappa_jD_L + L$. The full spectrum may then be expressed as a disjoint union of the block spectra,  
\begin{equation}
 \text{spec}\paren{K\whp L} = \bigsqcup_{j=1}^{|K|}\text{spec}\paren{\alpha\kappa_jD_L + L}.
\end{equation}
Now, we apply \cref{eq:thm310} to $K\whp L \equiv \paren{\alpha K}\uhp L$ and use the fact that $\phi_{\alpha K}(x) = \det\brac{x\id - \alpha K} = \alpha^{n_K}\det\brac{\frac{x}{\alpha}\id - K}\equiv \alpha^{n_K}\phi_{K}\paren{\frac{x}{\alpha}}$. This yields Eq.\ \eqref{eq:HPchar}, as desired.
\end{proof}

Now we show that if the eigenvalues of $K$ and the polynomials $\phi_L$ and $\phi_{L\pr}$ are known, then there is a straightforward procedure to compute the eigenvalues of $K\whp L$. 
\begin{lemma}
\label{lm:spec2}
  Let $K$ and $L$ be graph Laplacians, as before. Each eigenvalue of the product characteristic polynomial $\phi_{\uhp}$ can be found as a solution of the equation
\begin{equation}
\label{eq:recursive}
  \alpha\kappa_i = \frac{\phi_{L}\paren{x}}{\phi_{L\pr}\paren{x}}
\end{equation}
for some $K$-eigenvalue $\kappa_i$.
\end{lemma}

\begin{proof}
Any eigenvalue of the product graph must be a zero of the left-hand side of Eq.~(\ref{eq:HPchar}) and, by equality, a zero of the right-hand side. Now, the degree of polynomial $\phi_K$ is $n_K$, which implies that the term of degree $n_K$ must be nonzero. Thus, in the product $\phi_{L\pr}\paren{x}^{n_K}\phi_K\paren{\frac{1}{\alpha}\frac{\phi_{L}\paren{x}}{\phi_{L\pr}\paren{x}}}$, there must be a term which is \emph{indivisible} by the polynomial $\phi_{L\pr}\paren{x}$. Therefore, the zero of the right-hand side cannot be a root of the polynomial $\phi_{L\pr}$. 

We are seeking values of $x$ such that the polynomial $\phi_K\paren{\frac{1}{\alpha}\frac{\phi_{L}\paren{x}}{\phi_{L\pr}\paren{x}}}$ evaluates to zero. In other words, we are looking for $x$ such that the term $\frac{1}{\alpha}\frac{\phi_{L}\paren{x}}{\phi_{L\pr}\paren{x}}$ is a root of $\phi_K$. Therefore, we solve Eq.\ \eqref{eq:recursive} for $x$, for all roots $\kappa_i$ of $K$. 
\end{proof}

If the forms of $\phi_L$ and $\phi_{L\pr}$ are known (and if each have sufficiently low degree), then computing the roots of $\phi_{\uhp}$ becomes tractable, even if $K$ is a large matrix. This suggests a recursive procedure for computing the spectrum of a $k$-level hierarchy, by writing it as a product of the $\paren{k-1}$-level hierarchy with the $k$-th base graph. We now frame this as our main result of this section:
\begin{theorem}
  \label{th:specMain}
Suppose we have a $k$-level hierarchy $L^{\vhp k}$ described by base graph Laplacians $L_1,L_2,\ldots,L_k$ and weights $\vecalpha=\paren{1,\alpha_2,\ldots,\alpha_{k}}$ as follows,
\begin{equation}
  L^{\vhp k}=L_{k}\uhp_{\alpha_{k}}L_{k-1}\uhp_{\alpha_{k-1}}\cdots \uhp_{\alpha_3}L_{2}\uhp_{\alpha_{2}}L_1 .
\end{equation}
Define a new set of weights $\vecbeta=\paren{1,\beta_2,\ldots,\beta_{k}}$ with $\beta_i = \alpha_i/\alpha_{i-1}$, and a new set of Laplacians $M_{k},M_{k-1},\ldots, M_{1}$ recursively as 
\begin{align*}
M_k &= L_k , \\ 
M_i &= M_{i+1}\uhp_{\beta_{i+1}}L_i.
\end{align*}
Then, the following hold:
\begin{enumerate}
\item $M_1 = L^{\vhp k}$.
\item Any eigenvalue of $M_i$ (for $i<k$) may be found as a solution to the equation
  \begin{equation}
\label{eq:specRecursive}
    \beta_{i+1}\mu^{\paren{i+1}} = \frac{\phi_{L_i}(x)}{\phi_{L_i\pr}(x)}
  \end{equation}
for some $\mu^{\paren{i+1}}\in\text{spec}\curly{M_{i+1}}$.
\end{enumerate}

\end{theorem}
\begin{proof}
First, we prove statement 1. It can be seen that 
\begin{align}
M_{k-1} &= M_k\uhp_{\beta_{k}}L_{k-1} = L_k\uhp_{\beta_{k}}L_{k-1} \nonumber \\ 
&= \frac{1}{\alpha_{k-1}}\paren{\alpha_{k}L_k\otimes D_{k-1} + \alpha_{k-1}\id_{k}\otimes L_{k-1}},\\
M_{k-2} &= M_{k-1}\uhp_{\beta_{k-1}}L_{k-2} \nonumber \\
&= \frac{1}{\alpha_{k-2}}(\alpha_{k}L_k\otimes D_{k-1}\otimes D_{k-2}\ + \nonumber \\
\alpha_{k-1}\id_{k}&\otimes L_{k-1}\otimes D_{k-2} + \alpha_{k-2}\id_{k-1}\otimes \id_{k-2}\otimes L_{k-2}),
\end{align} 
and so on, until we have an $\vecalpha$-weighted sum over all $k$ of the base graphs (with an overall denominator of $\alpha_1=1$), which is precisely $L^{\vhp k}$.

The proof of statement 2 follows as a direct consequence of Lemma \ref{lm:spec2}, with $K=M_{i+1}, L=L_{i},$ and $\alpha=\beta_{i+1}$.
\end{proof}
Theorem \ref{th:specMain} provides an algorithm to compute the spectrum of $L^{\vhp k}$, namely:
\begin{enumerate}
\item Compute the relative weight vector $\vecbeta$ from $\vecalpha$. 
\item Start with $i=k$, where the spectrum of $M_k = L_k$ is known. Decrease $i$ by one.
\item Compute the spectrum of $M_i$ from the known spectrum of $M_{i+1}$ and Eq.\ \eqref{eq:specRecursive}. Decrease $i$ by one.
\item Perform step 3 repeatedly, halting at $i=0$. Return the spectrum of $M_1 = L^{\vhp k}$.
\end{enumerate}
Therefore, given a large hierarchy, one can efficiently compute the Laplacian eigenvalues and use them to find bounds on important graph properties. This is a scalable technique for obtaining figures of merit efficiently for hierarchies. Later, in Sec.~\ref{sec:comparison}, we will present analytic results for some of these figures of merit for simple hierarchies, but the results of the current section can be used even in more complicated cases, such as hierarchies that do not use the same $G$ at every layer or that have heterogeneous scaling parameters. 

\begin{figure}[htb]
  \centering
  \includegraphics[scale=0.25]{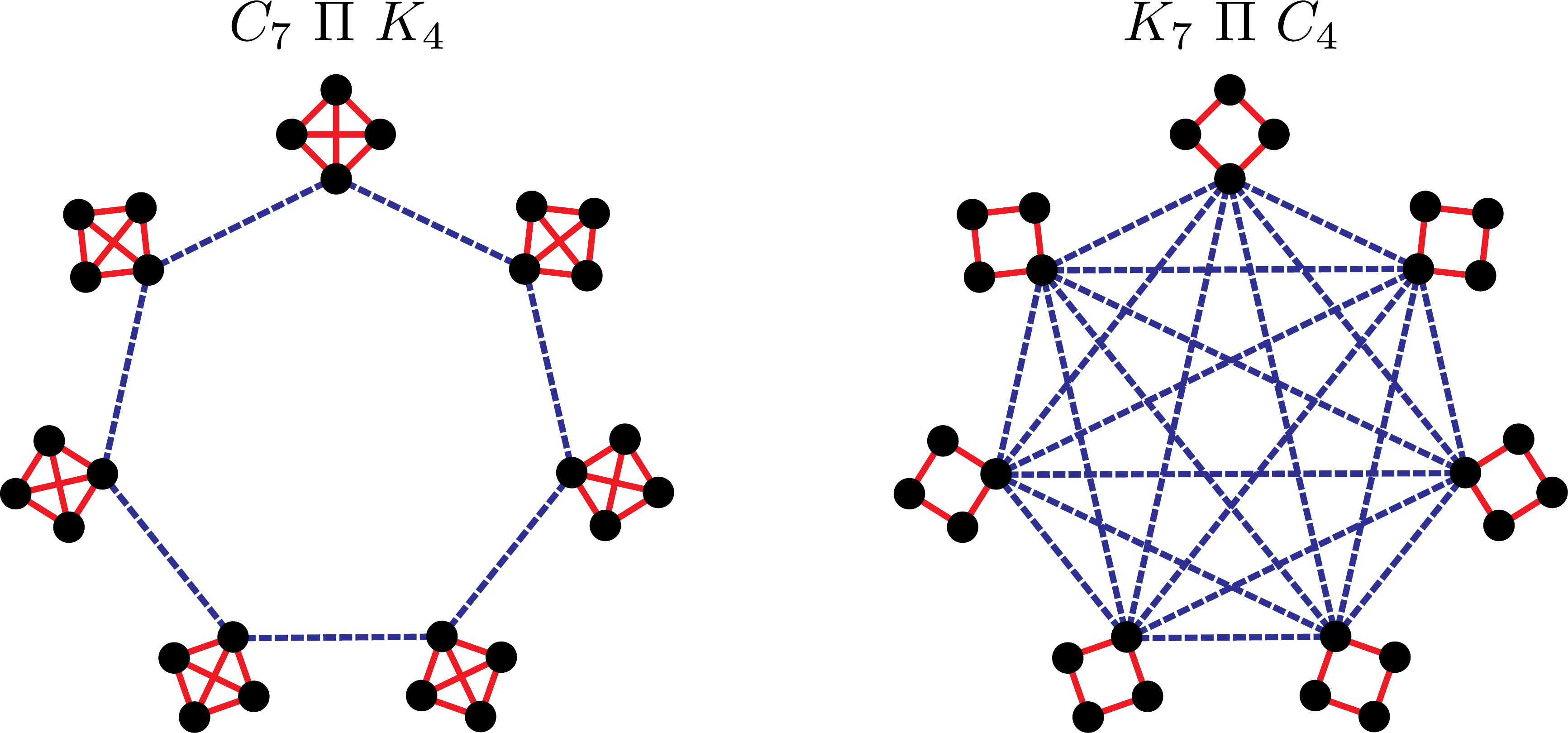}
  \caption{Two topologies with the same number of nodes (28) and edges (49). While the diameters for the two graphs are the same, are they equally well-connected? A comparison of the Cheeger constants (see \cref{tbl:connectivity}) suggests that the left graph is less interconnected. This is consistent with the spectral gap, which is smaller for the left graph, indicating poorer connectivity.}
  \label{fig:CK74}
\end{figure}

Due to the structural richness and heterogeneity of graphs, it is not always easy to decide whether one graph is, for instance, more connected than another graph. One aspect of connectivity is how close the nodes are to one another, which is captured by quantities like the diameter and mean distance. In Fig.~\ref{fig:CK74}, we compare two graphs, $C_7\uhp K_4$ and $K_7\uhp C_4$, which have an identical number of nodes (28) and edges (49). The two graphs also have identical diameters (5 each), but the mean distance for the left graph is smaller (see Table \ref{tbl:connectivity}). Under these measures, the left graph appears better connected. 

Better connectivity also corresponds to having fewer bottlenecks in the graph, which corresponds to a larger Cheeger constant. In Fig.~\ref{fig:CK74}, the graph on the right has a larger Cheeger constant, as one would expect given that it has complete connectivity between the seven modules. Note that this metric of connectivity need not agree with the mean distance, as seen in this example. 

Similarly, a parameter-by-parameter comparison of the two hierarchy graphs $C_{13}\uhp K_5$ and $K_{13}\uhp C_5$ (Table \ref{tbl:connectivity}) reveals that, while both graphs are two-level hierarchies with the same number of nodes and edges, $K_{13}\uhp C_5$ has the smaller diameter, smaller mean distance, larger cheeger constant, and a larger spectral gap, all of which indicate better connectivity. While structural comparisions for the above examples can be carried out simply by inspection or a quick calculation of graph quantities, general hierarchies may be far too complex to compare this way. In practice, when choosing a modular topology with the best connectivity, one might hope for a single, balanced measure of connectivity that relates to aspects such as node distance and bottleneckedness and is easy to compute. The spectral gap $\lambda_2$ meets these requirements. It is asymptotically related to the other invariants discussed here via upper and lower bounds in \cref{eq:diambound,eq:mdbound,eq:cheegerbound}. Furthermore, $\lambda_2$ can be efficiently computed using the recursive procedure described earlier in this section. 

\begin{table}
	\begin{tabular}{l | l l | l l}
		Graph Invariant & $C_7\uhp K_4$ \text{vs.}& $K_7\uhp C_4$ & $C_{13}\uhp K_5$ \text{vs.} & $K_{13}\uhp C_5$ \\ \hline
		Number of edges & 49 & 49 & 143 & 143 \\
		Number of nodes & 28 & 28 & 65 & 65 \\
		Diameter & \underline{5} & \underline{5} & 8 & \underline{5} \\
		Mean distance & \underline{2.68} & 2.71 &  4.77 & \underline{3.23} \\
		Cheeger constant & 0.17 & \underline{1.0} & 0.07 & \underline{1.4} \\
                Spectral gap $\lambda_2$ & 0.16 & \underline{0.46} & 0.04 & \underline{0.34} \\
	\end{tabular}
	\caption{Comparison of topologies by connectivity measure. In each case, the graphs being compared have an identical number of nodes and edges.   The better value for each comparison is underlined.}
	\label{tbl:connectivity}
\end{table}

\subsubsection{Truncated Hierarchical Product}
\label{sec:flexible}
\begin{figure}[htpb]
	\centering
	\includegraphics[width = .45\textwidth]{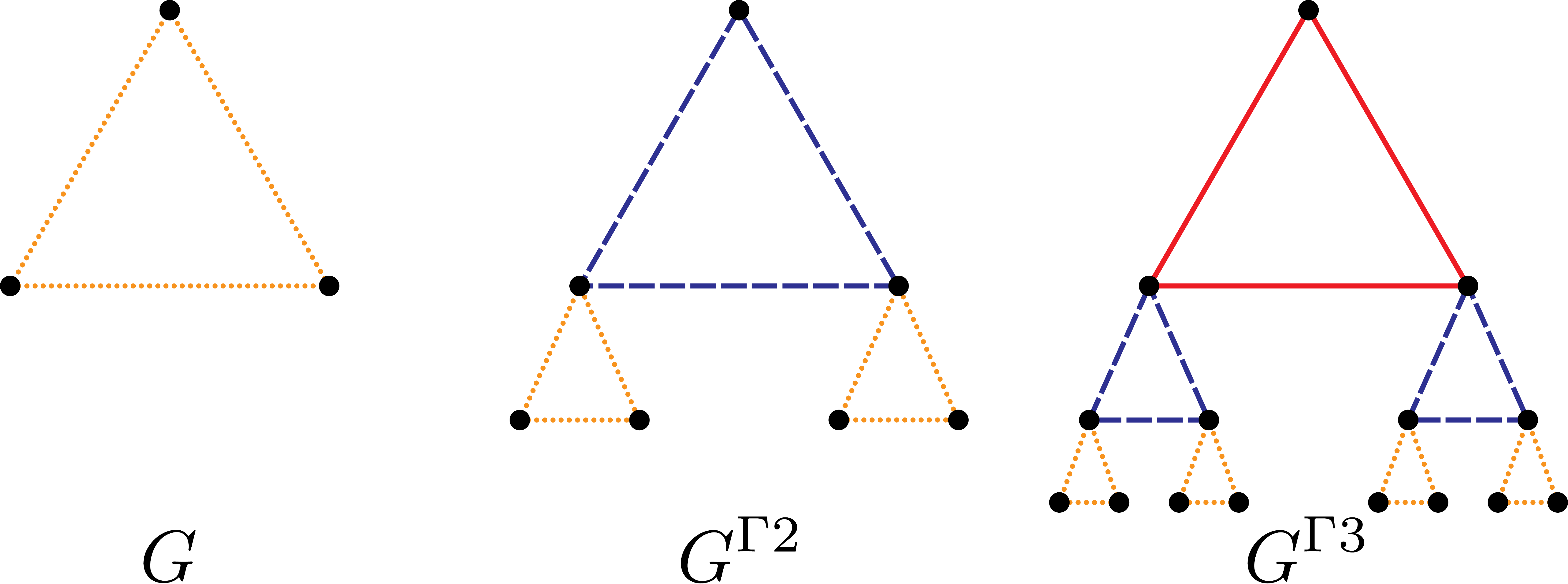}
	\caption{A demonstration of how our hierarchical product can be truncated to avoid requiring many interconnections at one node. As the hierarchy grows, the graph is duplicated and then attached to a subset of nodes in a larger version of the base graph, $G$.}
	\label{fig:comm_qubit}
\end{figure}
In some scenarios, there may be physical or technological limitations on the total number of interconnections allowed at a single node of a quantum computer. In our framework, this manifests as a restriction on the maximum degree of a node. We believe that hierarchical structures can still prove useful in this context, but (as we will see in Sec.~\ref{sec:comparison}) the hierarchy we have described thus far has a maximum degree which grows linearly with the number of levels of the hierarchy.

We now introduce an architecture which maintains the hierarchical properties but also has a bounded maximum node degree (i.e.\ maximum node degree that does not go to infinity as the number of levels goes to infinity). 
To model such an architecture, we modify the hierarchical product $G_1 \uhp G_2$. Whereas previously, $\abs{G_1}$ copies of $G_2$ were connected according to $G_1$, we now bring together $\abs{G_1} - 1$ copies, which we connect according to $G_1$, and add the root node of $G_1$ without an associated subhierarchy (see Fig.~\ref{fig:comm_qubit}). When extended to a many-level hierarchy, this means that every node will be connected to, at most, two levels, and so its degree will not grow as the hierarchy grows. We will denote this \textit{truncated} hierarchical product by $G_1 \tuhp G_2$, and its weighted version as $G_1 \twhp G_2$. It can be written algebraically in terms of adjacency matrices by adopting a more general definition of the hierarchical product.
\begin{definition}
\label{def:trunchp}
Given rooted graphs $G$ and $H$, the \emph{weighted truncated hierarchical product} $P = G\twhp H$ is a graph on vertices $V_P = V_G\times V_H$ and edges $E_P \subseteq V_P\times V_P$ specified by the adjacency matrix
\begin{equation}
  A_P = \alpha A_G \otimes D_H + P_G \otimes A_H,
\end{equation}
or, equivalently, the algebraic Laplacian
\begin{equation}
  L_P = \alpha L_G \otimes D_H + P_G \otimes L_H.
\end{equation}
Here, $P_G$ is a projector onto all nodes in $G$ except the root node. At the level of adjacency matrices, we may also write $A_P = A_G\twhp A_H$. An unweighted version, $G \tuhp H$, can be obtained by setting $\alpha = 1$. 
\end{definition}
An illustration of this architecture can be found in Fig.~\ref{fig:comm_qubit}. From this definition, we naturally derive both uneweighted and weighted truncated hierarchies, $G^{\tuhp k}$ and $G^{\tvhp k}$. We note that a generalization of this definition to allow an arbitrary projector (rather than one that only excludes the root node) is possible, but we do not consider such a case in this paper.

The addressing scheme outlined in Sec.~\ref{sec:addressal} can also be used for truncated hierarchies. However, since many nodes do not sit atop sub-hierarchies in this case, not all node addresses are valid. We will assume that the node in the $i$-th level which connects to the level above it has a zero in the $i$-th digit of its address. In a truncated hierarchy, each node whose address contains a zero (representing the ``root'' of a hierarchy) must have only zeros in all following positions, as it does not contain any further sub-hierarchies. The base-$n$ addressal scheme can thus be used to specify which nodes are present in a truncated hierarchy.

Note that the truncated hierarchical product adds nodes more slowly than (although with the same scaling as) the hierarchical product structure specified at the beginning of Sec.~\ref{sec:hpdef}. When we perform graph comparisons in Sec.~\ref{sec:comparison}, we will consider all cost functions and optimizations in terms of the total number of nodes so that the two architectures can be compared fairly.

\section{Graph Comparisons}
\label{sec:comparison}
\begin{figure*}[tb]
	\centering
	\includegraphics[width=.85\textwidth]{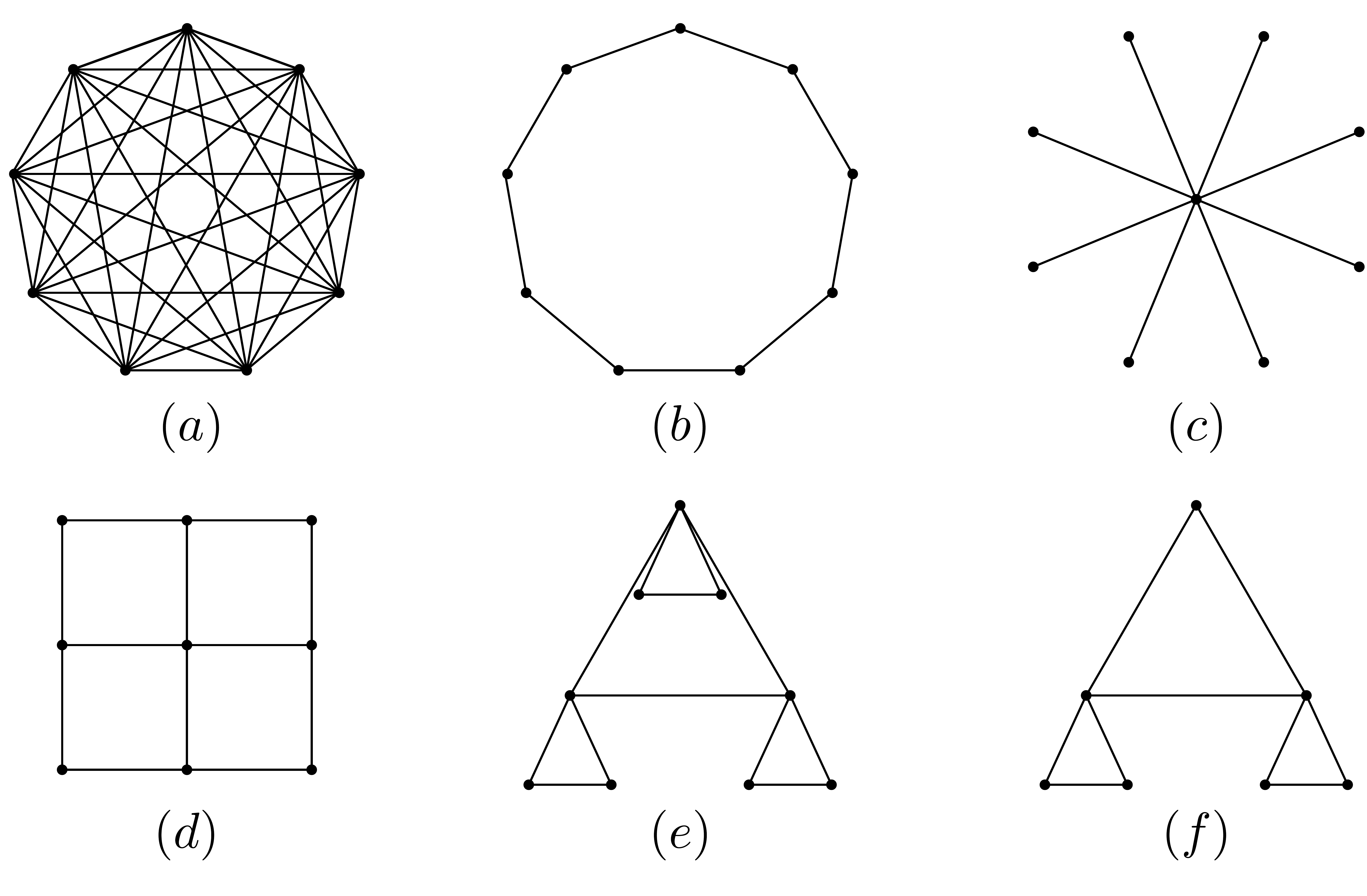}
	\caption{Illustration of the graph structures considered in this section, each with nine nodes except (f). (a) The complete graph $K_9$. (b) The cycle graph $C_9$. (c) The star graph $S_9$. (d) The nearest-neighbor grid in two dimensions. (e) The hierarchical product $K_3^{\uhp 2}$. (f) The truncated hierarchical product of Sec.~\ref{sec:flexible}, $K_3^{\tuhp 2}$.}
	\label{fig:graphs}
\end{figure*}
Having developed the machinery to construct hierarchies, we will now evaluate them against other potential architectures.
Any evaluation is impossible to do in an absolute sense, since what properties are desirable in a graph and how serious the cost of improving them is will depend on both the application as well as the physical system under consideration. In general, we assume that the most desirable quality of a graph is some measure of connectivity or the ease with which the graph can transport information between nodes.  
Note that it is always possible to translate between quantum circuit architectures with some overhead. A detailed atlas summarizing these overheads can be found in Ref.~\cite{Cheung2007}.

We will look at the scenario of state transfer, which is an important subroutine that may need to be carried out if an algorithm requires gates to be performed between two qubits that are not directly connected.
We consider the worst-case state transfer time on a given graph, which allows us to evaluate graphs without reference to any particular quantum algorithm. If we are interested in the time taken for state transfer in the graph, an appropriate metric can be the diameter of the graph, $\delta(H)$, under the assumption that information transfer takes unit time along any edge in the graph.
The diameter then captures the maximum distance, and hence the maximum time required for information to travel between any two nodes in the system.

For graphs produced by the weighted hierarchical product, we will also consider a diameter which takes into account edge weight. This ``weighted diameter,'' $\delta_w(H)$, can be found by considering all pairs of nodes $j,k$ and identifying the two whose least-weighted connecting path has the highest sum weight of edges. If we consider a path between two nodes $j$ and $k$ to be a set of nodes $P = \{j, v_1, v_2 \dots v_n, k\}$ with a weight $W(P)$ given by the sum $w_{j,v_1} + w_{v_1, v_2} + \dots +  w_{v_n, k}$, then the weighted diameter can be written as:
\begin{equation}
	\delta_w(H) = \max_{j,k} \min_P W(P).
\end{equation}
One way to grasp why the weighted diameter is a useful quantity is to consider the time weights of edges, where the weight signifies the time required to perform a gate between two connected qubits. In this case, the weighted diameter is the maximum time it will take us to perform a chain of two-qubit gates that connects two different qubits (for instance, using SWAP operations to bring the two qubits to adjacent positions and then performing the final desired operation). 

However, optimizing \emph{only} with respect to connectivity yields a trivial result, because a fully connected graph is obviously most capable of communicating information between any two points. Therefore, we will consider a number of different possible ``costs'' associated with physical implementations of graphs. One potential input to the cost function is the maximum degree of a graph, $\Delta(H)$. 
As discussed in the previous section, we want to avoid needing to connect too many different communication channels to a single node. 
Another is total edge weight $w(H)$ -- if it costs time, energy, money, coherence, or effort to produce communication between two nodes, we should try to use as few communication channels as possible.

We now walk through the calculations for several important graph quantities for several graphs: an all-to-all connected graph, a cycle graph, a star graph, a square grid, a hierarchy graph with scaling parameter $\alpha$, and a truncated version of that same hierarchy graph. We calculate how quantities scale with the total number of nodes $N$. For ease of calculation, we assume that $N$ nodes fit in the architecture of the current graph; for instance, we assume $N = \ell^d$ for some integer $\ell$ for a $d$-dimensional square graph. All results of this section are compiled in Table~\ref{tbl:graphs}, and examples of the graphs for small $N$ are illustrated in Fig.~\ref{fig:graphs}.

\setlength\tabcolsep{.5cm}
\begin{table*}[htpb]
	\begin{tabular}{l | l l l l}
		Graph $H$ & Diameter $\delta$  & Weighted Diameter $\delta_w$ & Maximum Degree $\Delta$ & Total Edge Weight $w(H)$ \\ \hline
		$K_N$ & const. & const. & $N$ & $N^2$ \\
		$S_N$ & const. & const. & $N$ & $N$ \\
		$C_N$ & $N$ & $N$ & const. & $N$ \\
		Square grid, $d$-dim & $d N^{1/d}$ & $d N^{1/d}$ & $d$ & $dN$ \\ 
		$K_n^{\whp k}, \alpha \neq n $ & $ \log_n N $ & $\max\left(\frac{2}{1-\alpha}, N^{\log_n \alpha} \right)$ &  $n \log_n N $ & $nN^{\max\left( 1, \log_n \alpha \right)} $\\ 
		$K_n^{\whp k}, \alpha  = n $& $\log_n N $ &  $\max\left(\frac{2}{1-\alpha}, N^{\log_n \alpha} \right)$   & $ n \log_n N $ & $nN \log_n N$   \\
		$K_{n+1}^{\twhp k}, \alpha \neq n $ & $ \log_{n} N $ & $\max\left(\frac{2}{1-\alpha}, N^{\log_{n} \alpha} \right)$ &  $ n$ & $nN^{\max\left( 1, \log_{n} \alpha \right)} $\\ 
		$K_{n+1}^{\twhp k}, \alpha  = n $& $\log_{n} N $ &  $\max\left(\frac{2}{1-\alpha}, N^{\log_{n} \alpha} \right)$   & $ n $ & $nN \log_{n} N$   \\
	\end{tabular}
	\caption{Summary of scalings of important graph properties with total node number, $N$. All entries describe only the scaling of the leading coefficient with $d$, $n$, and $N$.}
	\label{tbl:graphs}
\end{table*}
\subsection{Graph Calculations}
\subsubsection{Complete Graph, $K_N$}
Since all nodes in a complete graph [\cref{fig:graphs}(a)] have edges between them, the diameter is simply 1. This comes at the cost of very high maximum degree, $N-1$, as every node is connected to all $N-1$ other nodes. The total weight of every edge is the same, and there are $N(N-1)/2$ edges because every pair of nodes has a corresponding edge. Therefore, the total edge weight scales as $\Theta(N^2)$. 

\subsubsection{Cycle Graph, $C_N$}
In a cycle graph [\cref{fig:graphs}(b)], the diameter is $\lfloor N/2 \rfloor$, the distance to the opposite side of the circle. The maximum degree is only 2, and the total weight of the edges is likewise only $N$. This graph is thus able to reduce the cost factors associated with the complete graph, but at the cost of a much higher asymptotic diameter.

\subsubsection{Star Graph, $S_N$}
The star graph is the graph which has a single central node connected to all others [\cref{fig:graphs}(c)]. Like the complete graph, it also has a constant diameter, although this diameter is two rather than one. The maximum degree of the star graph is $N-1$, the same as the complete graph. However, the star graph improves over the complete graph, as it has a lower total edge weight of $N-1$ rather than ${N \choose 2}$. Thus, we have improved the cost asymptotically without affecting the overall scaling of the diameter of the graph.

The example of $S_N$ raises a complication which we do not attempt to quantify in this paper. In a realistic distributed quantum computer, we expect that a significant amount of operations need to be performed at the same time and need to be scheduled on the graph. But in the star graph, all operations between nodes must pass through the single central hub. This is likely to lead to a scheduling bottleneck when performing general quantum algorithms. While we do not attempt to treat scheduling of such algorithms on the network in this paper, in future work we hope to consider these complications, which will at times make the star graph unsuitable for real-world use. An experimental comparison of the star graph and the complete graph in existing five-qubit quantum computers can be found in Ref.~\cite{Linke2017}. In those experiments, the requirement that all information be shuttled through a central node for the $S_N$ connectivity made high-fidelity execution of quantum algorithms more difficult.

\subsubsection{Square Grid Graph}
We consider now a square grid (i.e., a hypercubic lattice) in $d$ dimensions [\cref{fig:graphs}(d)]. Here, the diameter is $d (N^{1/d}-1)$, since this is the distance from the point in one corner labeled $(1 , 1 ,1, \dots)$ to the opposite corner at $(N^{1/d}, N^{1/d}, \dots )$ (note that diagonal moves are not allowed). The maximum degree depends on the dimension, as each interior node is connected to $2d$ other nodes.
The total edge weight can be found by considering that each node on the interior of the graph corresponds with exactly $d$ edges, and it is these edges that dominate as $N \to \infty$. Therefore, the total edge weight scales as $\Theta(d N)$. 

\subsubsection{Hierarchy Graph, $G^{\vhp k}$}
\label{sec:hierarchygraph}
As the hierarchy graph [\cref{fig:graphs}(e)] is built recursively, it is easiest to calculate its properties using recursion relations. We consider a graph that has $k$ levels to it, so that given a base graph $G$ and $n = \abs{G}$, then the overall graph has $n^k$ nodes. 

First, we calculate the unweighted diameter of a $k$-level hierarchy, which we denote by $\delta\paren{G^{\vhp k}}$. Since all sub-hierarchies are rooted at their first vertex, we will need to keep track of the \emph{eccentricity} of the root node, which we denote by $\varepsilon(F)$ for any subhierarchy $F$. The eccentricity of any graph node is defined as the maximum distance from that node to any other node in the graph $F$. Here, we fix $\varepsilon(F)$ to be the root eccentricity for the graph in question.

Now, we write recursion relations for two quantities, the unweighted diameter $\delta(G^{\vhp i})$ of an $i$-level hierarchy for some intermediate $i$, and the eccentricity $\varepsilon(G^{\vhp i})$ of the top-level root node of the current $i$-level hierarchy.

Consider a diametric path in an $i$-level hierarchy. This path must ascend and descend the entire hierarchy. That is, using the notation of Sec.~\ref{sec:addressal}, two maximally separated qubits have addresses that are different in their first digit. Such a path can always be partitioned into 3 disjoint pieces, the terminal two of which each lie in some $(i-1)$-level subhierarchy, while the middle piece lies in the current top (i.e. $i$-th) level. These three pieces must be independently maximal, since the path is diametric. The middle piece maximizes to the diameter of the top-level graph, which is simply $\delta(G)$. The two sub-level pieces each maximize to the root eccentricity of the $(i-1)$-th level subhierarchy, which is precisely the quantity $\varepsilon(G^{\vhp (i-1)})$. Therefore, our first recursion reads
\begin{equation}
  \delta(G^{\vhp i}) = 2\varepsilon(G^{\vhp (i-1)}) + \delta(G).
\end{equation}
The $i$-th level root eccentricity may be found by a similar argument. Partition the most eccentric path (starting at the top level root node) into two pieces, one which lies at the top level, and the other which lies exclusively in the lower levels. Maximizing both pieces, one gets
\begin{equation}
  \varepsilon(G^{\vhp i}) = \varepsilon(G^{\vhp (i-1)}) + \varepsilon(G).
\end{equation}
Solving the second relation, we get $\varepsilon(G^{\vhp i}) = i\varepsilon(G)$. By substitution, the first recursion has the solution
\begin{equation}
    \delta(G^{\vhp k}) = 2(k-1)\varepsilon(G) + \delta(G).
\end{equation}
Since the total number of levels is given by $k = \log_n N$, and the graph diameter is no greater than twice the eccentricity of any node, we conclude that the diameter scales as $\Theta(\varepsilon (G)\log_n N)$ for a general graph $G$. If we specifically examine the case when $G$ is a complete graph of order $n$, $\delta(G) = 1$ and $\varepsilon(G) = 1$, and the exact expression is $\delta\paren{G^{\vhp k}} = 2 \log_n(N) - 1$.

Next we calculate the maximum degree. Again, we proceed by recursion. Iterating the hierarchical product to some level $i$ can be viewed as attaching a copy of the graph $G^{\vhp \paren{i-1}}$ to every point in the graph $G$. Therefore, the degree of every root node in the $\paren{i-1}$-level subhierarchies increases by the degree of the corresponding node in graph $G$. The maximal increase achievable thus is the maximum degree $\Delta(G)$ of graph $G$. Since the root node for an $i$-level subhierarchy has $i$ distinct copies of $G$ attached to it, its degree is given by $i\cdot\text{deg}\paren{g_1}$, where $g_1$ is the root node of $G$. Then, the $i$-level maximum degree can be expressed as
\begin{align}
  \Delta(G^{\vhp i}) &= \max\curly{\paren{i-1}\text{deg}\paren{g_1} + \Delta\paren{G},\Delta(G^{\vhp \paren{i-1}})}\\
\ldots &= \max_{0\leq j\leq i-1}\left\{j\ \text{deg}\paren{g_1} + \Delta(G)\right\}\\
  &= (i-1)\text{deg}\paren{g_1} + \Delta\paren{G},
\end{align}
where the second step was obtained by recursion. For a general $G$, this gives the maximum degree scaling as $\Delta(G^{\vhp k}) = \Theta(\log_{n} N)$. For $K_n^{\vhp k}$, the root degree and the maximum degree of the base graph $K_n$ are both $n-1$, so $\Delta(K_n^{\vhp k}) = (n-1) \log_n N$. 

Now we consider the total edge weight of the hierarchy. We compute this by a recursion relation, first by duplicating the existing edge weight at $i-1$ levels by $n$ (the number of smaller hierarchies we must bring together) and then adding new edges. If the edges at level $i$ have weight $\alpha_i$, we can write this as:
\begin{equation}
	w(G^{\vhp i}) = n w(G^{\vhp \paren{i - 1}}) + \alpha_i w(G).
	\label{eqn:sieweight}
\end{equation}
By counting the number of subhierarchies with different weights, we find the following form for the total edge weight of the weighted hierarchy:
\begin{equation}
	w\left( G^{\vhp k} \right) = w(G) \sum_{i=1}^k \alpha_i \abs{G}^{k - i}.
\end{equation}
This can be verified by checking that it satisfies the recursion relation \cref{eqn:sieweight}. If we now specialize to the case where $G = K_n$ and $\alpha_i = \alpha^{i-1}$, we find
\begin{equation}
	w \left( K_n^{ \whp k} \right) = \frac{n ( n - 1)}{2} \sum_{i=1}^k \alpha^{i-1} n^{k - i}.
\end{equation}
This behavior can be broken into three regimes. For $\alpha = n$, the sum is constant, and the overall scaling is $\Theta(n N \log_n N)$. Otherwise, we can perform the geometric sum to obtain
\begin{equation}
	w \left( K_n^{\whp k} \right) = \frac{n (n-1)}{2} \frac{n^k - \alpha^k}{n - \alpha}.
\end{equation}
Here, the scaling will depend on the relative size of $n$ and $\alpha$. For $n > \alpha$, the first term in the numerator dominates, and $w\left( K_n^{\whp k} \right) = \Theta(n N)$. Otherwise, we can write $\alpha^k = N^{\log_n \alpha}$ and find  $w\left( K_n^{\whp k} \right) = \Theta(n N^{\log_n \alpha})$.

Finally, we calculate the weighted diameter of a $k$-level hierarchy $\delta_w(G^{\vhp k})$, just as for the unweighted diameter, by solving recursion relations for the quantities $\delta_{w}(G^{\vhp i})$ and $\varepsilon_{w}(G^{\vhp i})$, which are, respectively, the weighted diameter and weighted root eccentricity for an $i$-level weighted hierarchy. Here, note that the top level (at any intermediate stage $i$) is weighted by $\alpha_i$. Therefore, the recursion for the weighted diameter is modified to
\begin{equation}
\label{eqn:wdrecur}
  \delta_{w}(G^{\vhp i}) = 2\varepsilon_{w}(G^{\vhp \paren{i-1}}) + \alpha_i\delta_w(G).
\end{equation}
Similarly, the recursion for the weighted eccentricity becomes
\begin{equation}
  \varepsilon_{w}(G^{\vhp i}) = \varepsilon_{w}(G^{\vhp \paren{i-1}}) + \alpha_i\varepsilon_{w}(G),
\end{equation}
which has the solution $\varepsilon_{w}(G^{\vhp i}) = \varepsilon_{w}(G)\suml{j=1}{i}{\alpha_j}$. Finally, we have
\begin{equation}
    \delta_{w}(G^{\vhp k}) = 2\varepsilon_{w}(G) \suml{j=1}{k-1}{\alpha_j} + \delta_w(G) \alpha_k.
\end{equation}
For $G = K_n$ and $\alpha_i = \alpha^{i-1}$, this becomes:
\begin{align}
	\delta_w(K_n^{\whp k}) &=  2\sum_{i=1}^{k-1} \alpha^{i-1} + \alpha^{k-1}\\
  &= \frac{\alpha^{k} + \alpha^{k-1} -2}{\alpha - 1}.
  \label{eqn:dwhier}
\end{align}
Therefore, the scaling of the weighted diameter with $N$ has two regimes, depending on $\alpha$. For $\alpha < 1$ the geometric sum converges as $i \to \infty$ to $\frac{2}{1 - \alpha}$. This means that for $\alpha < 1$, a constant time suffices to traverse the entire hierarchy no matter how large it is. For $\alpha = 1$ the weighted diameter is equal to the (unweighted) diameter, which we have already computed. For $\alpha > 1$, $\delta_w$ scales as $\alpha^{k-1} = N^{\log_n \alpha}/\alpha \sim N^{\log_n \alpha}$. Note that the last scaling only applies if $\alpha$ does not scale with $n$. Since $n > 1$ and $\alpha >1$, this exponent $\log_n\alpha$ is always positive. Therefore, the total edge weight is asymptotically always either constant (for $\alpha < 1$) or growing (for $\alpha \geq 1$), as expected.

\subsubsection{Truncated Hierarchy, $G^{\tvhp k}$}
Finally, we look at how the results above are modified if we use the truncated hierarchical product discussed in Sec.~\ref{sec:flexible} [\cref{fig:graphs}(f)]. Although many of the calculations in terms of the number of levels $k$ are similar to those for the non-truncated hierarchy, it is no longer the case that $k = \log_n N$ exactly. In order to compare graphs fairly, we will need to recalculate the order of $G^{\tvhp k}$ so that results in this section can be written in terms of the total number of nodes, $N$. 

Under the node addressal scheme of Sec.~\ref{sec:addressal}, the nodes of a truncated hierarchy are in one-to-one correspondence with base-$n$ strings of length $k$ that only have trailing zeros. As before, a $0$ label points to a root node, but since root nodes do not bear subhierarchies due to truncation, all subsequent labels are forced to be $0$. In other words, we only label nodes using strings of the form $(l_1l_2\ldots l_i 00\ldots 0)$ for some $i\le k$, and $l_j\ne 0$ for all $j\le i$. The number of such strings with $i$ nonzero labels followed by $(k-i)$ zero labels is $(n-1)^i$. Therefore, the total number of nodes is
\begin{equation}
  N =  \sum_{i=0}^k (n-1)^i.
\end{equation}
Since $N = \Theta \left( (n-1)^k \right)$, many quantities of a truncated hierarchy with a base graph of order $n+1$ have the same scaling with the number of nodes $N$ as those for a non-truncated hierarchy with a base graph of order $n$.

In terms of the number of levels $k$, the maximum diameter will be proportional to $k$, just as it was in Sec.~\ref{sec:hierarchygraph}. It follows that the diameter scales with the total number of nodes as $\delta = \Theta\left(\log_{n-1} N\right)$ for a truncated hierarchy.

On the other hand, truncation offers a large improvement in the maximum degree of the hierarchy. As discussed in Sec.~\ref{sec:flexible}, the maximum degree of the truncated hierarchy is $\Delta(G^{\tvhp k}) = 2 \Delta(G)$, which is constant in $N$.

The edge weight recursion relation is simply $n-1$ copies of the current graph and then new, additional edges:
\begin{equation}
	w(G^{\tvhp i}) = (n - 1)  w(G^{\tvhp \paren{i - 1}}) + \alpha_i w(G).
\end{equation}
This is identical to the recursion relation for the standard hierarchy, Eq.~\eqref{eqn:sieweight}, except that there are now only $n-1$ copies, and also, for a given number of qubits $N$, the number of levels $k$ may be different by constant factors and terms. Thus, the only modification to the recursion relation is to replace $n$ with $n-1$, and the solution of the relation is otherwise identical. This means that none of the asymptotic scaling with $k$ is affected, and the scaling with $N$ is only affected by changing the total number of levels required to construct a graph of $N$ nodes.

The recursion relation for weighted diameter is similar to Eq.~\eqref{eqn:wdrecur}. Due to truncation, one needs to make a careful comparison of paths that do or do not terminate at the root node of the top level, but in any case the weighted diameter's scaling with $k$ is the same as the non-truncated weighted diameter's scaling. The weighted diameter scaling with $N$ can thus be found from \cref{eqn:dwhier}, using the appropriate value of $k$ for truncated hierarchies with $N$ nodes.

\subsection{Choosing Among Graphs}
\label{sec:choosing}
\subsubsection{Graph Embeddings}
The long list of comparisons summarized in Table~\ref{tbl:graphs} can make it difficult to see exactly when different graphs are preferable. To make our calculations more concrete, we would like to compare concrete scenarios for the connection of qubits arranged on a grid in $d$ dimensions. Specifically, in each dimension ($d = 1$, $2$, and $3$), we examine a hierarchy that is embedded into the grid, comparing its properties to the same grid but with nearest-neighbor connections. We consider building modules where each small module is a complete graph of size $n$, laid out in cubes on the grid so that the side-length of the cube is $n^{1/d}$. The $d = 1$ and $d = 2$ cases with $n = 2^{d}$ are illustrated in Fig.~\ref{fig:embeddings}.

\begin{figure}[htpb]
	\centering
	\includegraphics[width=.45\textwidth]{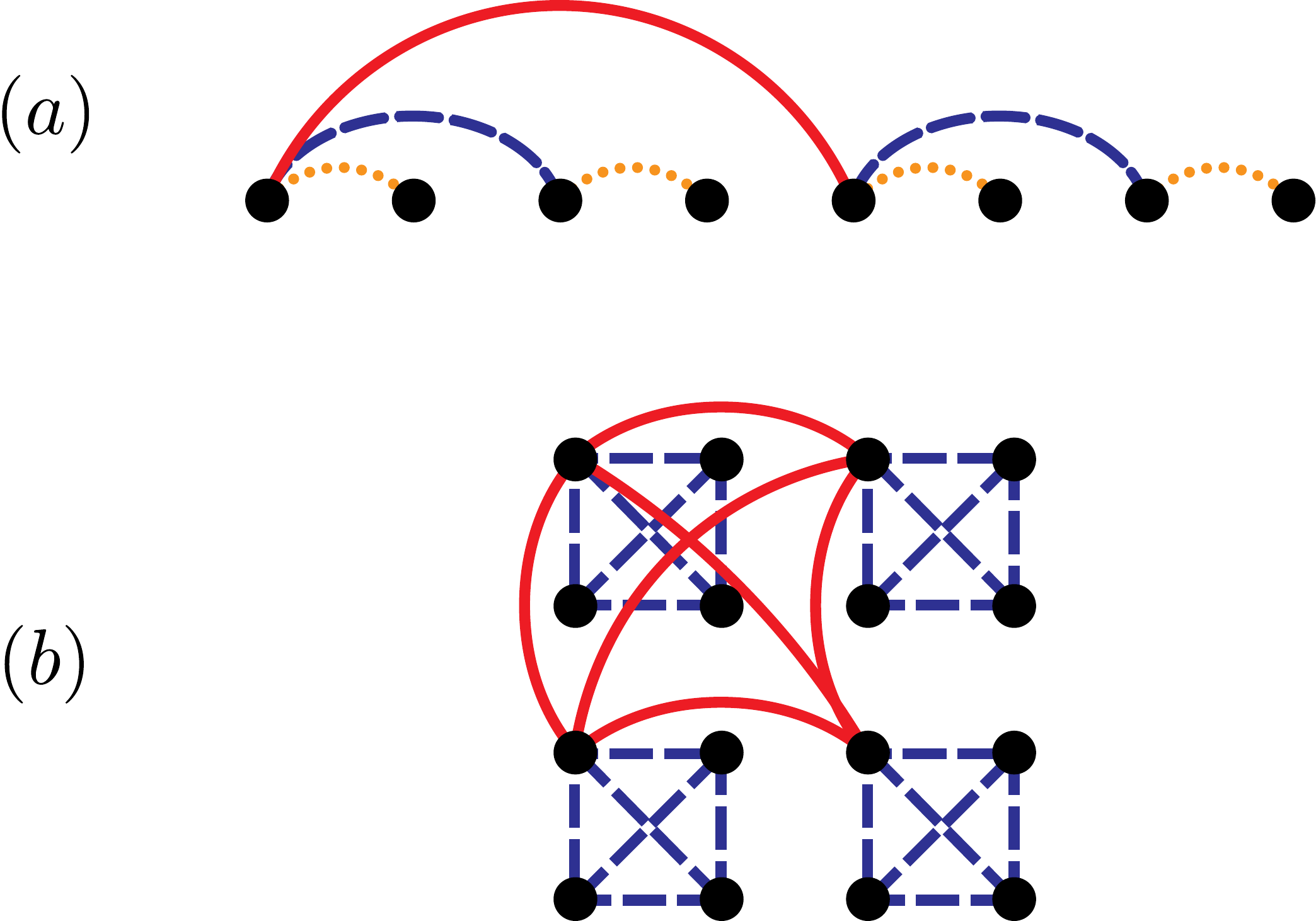}
	\caption{An illustration of the embedding of a hierarchy on a (a) one- or (b) two-dimensional lattice of qubits. In both cases, the length of an edge doubles at every level of the hierarchy, but the scaling in total edge length used changes from $\Theta(N \log_2 N)$ to $\Theta(N)$ when going from 1 to 2 dimensions. In $d = 3$, a similar hierarchy with doubling length scales connects modules of eight qubits.}
	\label{fig:embeddings}
\end{figure}

As shown, the length of an edge must increase by a factor of $n^{1/d}$ (2 in \cref{fig:embeddings}) at every level of the hierarchy in order to make these hierarchies possible. Therefore, to determine the total length of wire used, we can use a cost weight with $\alpha = n^{1/d}$. Keeping factors of $N$ only, Table~\ref{tbl:graphs} shows that for $d = 1$, we expect a total cost weight $\Theta \left(N \log_n N \right)$, while for the higher-dimensional cases we expect a total cost weight $\Theta \left(N\right)$ \footnote{If, for the application at hand, a planar graph is required, cycles such as $C_n$ can yield the same scaling.}. For the $d$-dimensional grid, this total cost weight is always $\Theta(N)$.

Now, to consider the performance of the two graphs, we must fix a separate scaling factor for the time weight, $\beta$. There are several options which might be reasonable for different physical applications. If $\beta = 1$, i.e., all links act identically in terms of time required to traverse them, then the weighted diameter of the hierarchy is simply $\Theta(\log_n N)$. Another option would be to take $\beta = \alpha$, i.e., to assume that links take as long to move through as they are long. In this case, we find that the hierarchy's weighted diameter scales as $\Theta \left(N^{1/d} \right)$, meaning that the hierarchy and nearest-neighbor graphs match in performance. 

We may also want to allow hierarchies to make use of the ``fat tree'' concept to produce a better-performing graph \cite{Leiserson1985}. Suppose that we allow ourselves to ``spend more'' on higher-level links, causing their cost weight to increase with a factor $\alpha$, but improving their performance so that the time weight scales with the factor $\beta = 1/\alpha$. In this case, the question is what range of $\alpha$ allows for the hierarchy to perform better than the nearest-neighbor grid (lower time-weighted diameter) for less cost (lower total edge cost weight)? (Note that this cost weight includes any contribution from ``lengthening'' wires at hire levels of the hierarchy.)

To answer the first, we compare the two asymptotic diameter scalings, $N^{\max ( 0, \log_n 1/\alpha)}$ and $N^{1/d}$. This suggests that if $\alpha \geq n^{-1/d}$, the hierarchy will allow for faster traversal than the nearest-neighbor grid. 
However, we wish to avoid causing the hierarchy to have a total cost weight that scales worse than $\Omega(N)$, which requires $\log_n \alpha < 1$. We find that a winning hierarchy can be constructed if $\alpha$ lies in the range $\alpha \in \left[n^{-1/d}, n\right)$. The optimal $\alpha$ is as large as possible but less than $n$; at that point an additional logarithmic factor is introduced to the total cost weight scaling.

In these cases, we have not allowed the nearest-neighbor grid to modify the weight (either kind) of its links. This is because any modification in its cost or time weight enters simply as a constant factor; if the individual links have weight $c$ instead of 1, the overall weighted diameter is just $c N^{1/d}$ while the total cost weight is just $c N$. Of course, one can apply different constants to each figure of merit, or apply $c$ to one and $1/c$ to the other. In order to make the nearest-neighbor grid match the performance of the hierarchy, the unit-length time weight would have to be $N^{\log_n (\alpha) - 1/d}$ while the unit-length cost weight must not scale with $N$. 

\subsubsection{Pareto Efficiency}

Our calculation of various graph parameters suggests that the hierarchy architecture offers significant advantages over others. One way to make this comparison more exact is to appeal to the economics concept of Pareto efficiency, which is used to designate an acceptable set of choices in multiparameter optimization \cite{PindyckRubinfeld}. A choice is Pareto efficient if switching to a different choice will cause at least one parameter to become worse. Suppose we eliminate all constants to focus only on the scaling with $N$ for three parameters: weighted diameter, maximum degree, and total edge weight. By removing these constants, we assume that the small multiplicative factors they provide will not influence decision making. For simplicity, we will assume that both cost and time weights scale with the same factor, $\alpha$.

\begin{figure}[h]
	\centering
	\includegraphics[width=.25\textwidth]{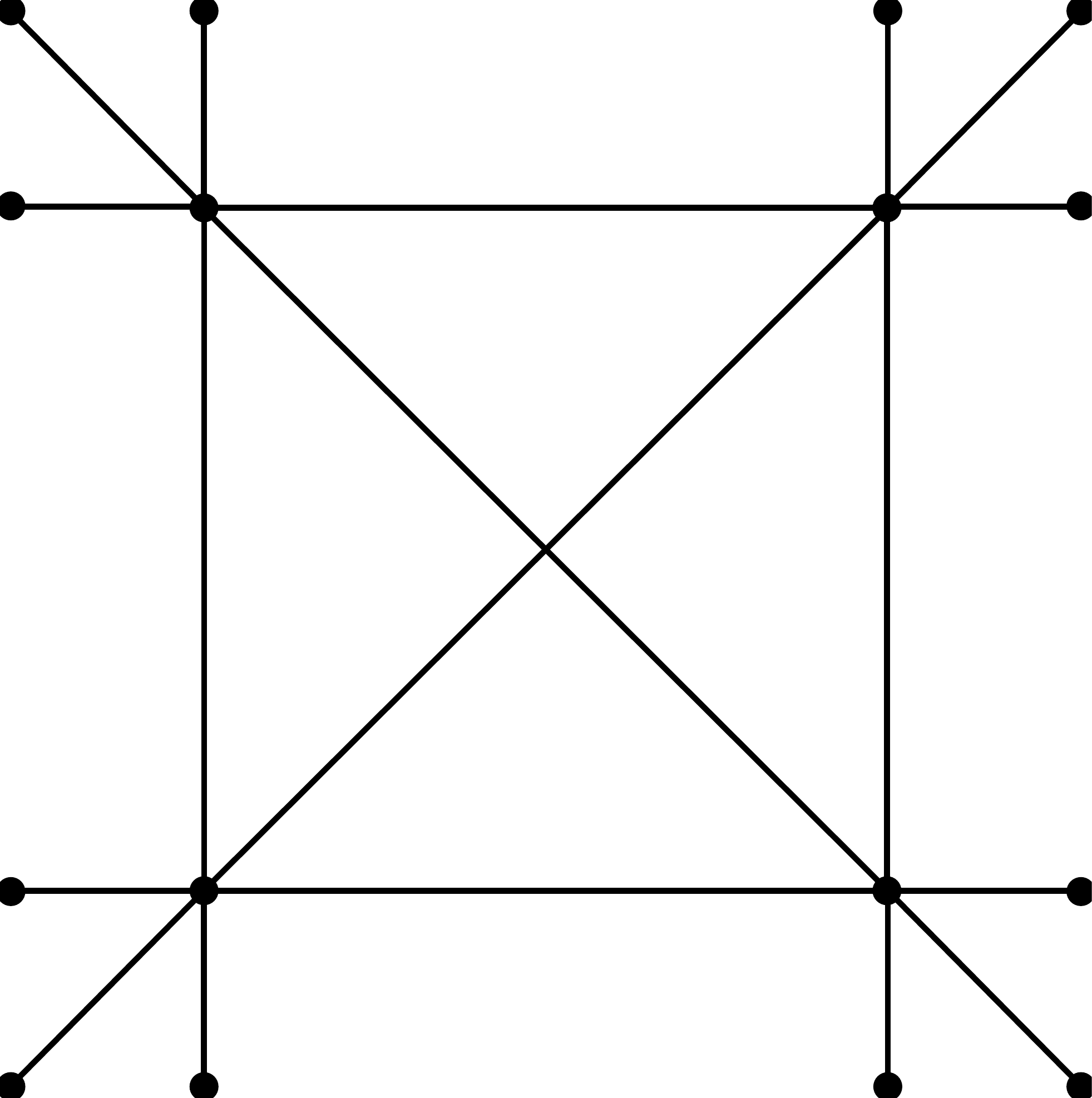}
	\caption{An example of a porcupine graph as defined in Ref.~\cite{Furedi1990}, in this case, $K_4 \uhp S_4$.}
	\label{fig:porcupine}
\end{figure}

For comparison, one could ask: what minimum number of edges is required for a graph on $N$ nodes to have maximum degree $\Delta$ and diameter $\delta$? Reference \cite{Furedi1990} answers this optimization question partially, and constructs what are known as \emph{porcupine graphs} which achieve the optimum, illustrated in Fig.~\ref{fig:porcupine}. We observe here that qualitatively, porcupines are modular, since they may be described by attaching trees to the nodes of a complete graph. In particular, the graph $K_{\sqrt{N}}\uhp S_{\sqrt{N}}$ is a porcupine graph that achieves a diameter $\delta=3$ and a maximum degree of $\Delta = 2(\sqrt{N}-1)$ with the minimal number of edges.

\begin{table}
	\begin{tabular}{@{} l @{} l | l l l}
		\multicolumn{2}{@{} l|}{$\phantom\star$ Graph} & $\delta_w$ & $\Delta$ & $w$ \\ \hline
		\multicolumn{2}{@{} l|}{$\phantom\star$ $K_N$} & const. & $N$ & $N^2$ \\
		\multicolumn{2}{@{} l|}{$\phantom\star$ $S_N$} & const. & $N$ & $N$ \\
		\multicolumn{2}{@{} l|}{$\phantom\star$ $C_N$} & $N$ & const. & $N$ \\
		\multicolumn{2}{@{} l|}{$\phantom\star$ Square grid} & $N^{1/d}$ & const. & $N$ \\
		\multicolumn{2}{@{} l|}{$\star$ $K_{\sqrt{N}} \uhp S_{\sqrt{N}}$} & const. & $\sqrt{N}$ & $N$ \\
		\multirow{2}{*}{$\star$ $K^{\twhp k}_{n+1}$ $\left\{\phantom{\begin{tabular}{@{\ }l@{}}\,\\\,\end{tabular}}\right.$} & $\alpha \neq 1$ & $N^{\log_n \alpha}$ & const. & $N$ \\
		& $\alpha = 1$ & $ \log_n N$ & const. & $ N$
	\end{tabular}
	\caption{An illustration of the scaling with $N$ of three key parameters to be used in Pareto optimization. Here $\delta_w$ is the weighted diameter, $\Delta$ is the maximum degree, and $w$ is the total edge weight of the graph. A star ($\star$) has been placed next to the two graphs we find to be Pareto efficient. We have also included the $\alpha = 1$ (unweighted) hierarchy in the final row, as it has a different scaling for the weighted diameter. Our Pareto efficiency judgment is made assuming $n^{1/d} \geq \alpha \geq 1$.}
	\label{tbl:pareto}
\end{table}

We summarize the scalings of these graphs in Table~\ref{tbl:pareto}. Assume that $n^{1/d} \geq \alpha \geq 1$. In this case, we can find the Pareto-efficient solutions by noting which options can be eliminated. We see that $K_N$ is strictly worse than $S_N$ and can be eliminated; $S_N$ is then dominated by the porcupine. $C_N$ is dominated by the square grid, which has identical scaling of total weight and degree but lower diameter. The square grid, in turn, is dominated by the hierarchy due to the assumptions we have made on $\alpha$. This means that the two Pareto-efficient choices in this case are the truncated hierarchy and the porcupine graph. If we chose any option besides these two, we could improve the scaling with respect to $N$ without any trade-off by switching to one of them. While this framework does not offer a decision rule to choose between the porcupine and $K_n^{\twhp k}$, the latter is clearly preferable if our aim is to create a modular quantum system that does not rely on a few centralized nodes. We stress that this optimization procedure is only intended to evaluate the quantities and graphs introduced, and the Pareto-efficient choices will change if other figures of merit or other graphs are included in the optimization.

\subsubsection{Optimality of diameter for hierarchical graphs}
The use of $K_n^{\twhp k}$ may be further motivated via the degree-diameter problem \cite{Hoffman1960} (for a survey, see Ref.\ \cite{Miller2005}). Given a graph with a maximum allowed degree $\Delta$ on each node and diameter no greater than $\delta$, the degree-diameter problem asks for the maximum number of nodes $N(\Delta,\delta)$ that such a network could hold. This problem is practically well-motivated in the design of networks, and may be answered for special classes of graphs. 
The Moore bound, which is a bound for general graphs, states that the number of nodes $N$ is at most $\frac{\Delta (\Delta -1)^\delta - 2}{\Delta - 2}$.
This means that for a constant maximum degree $\Delta \geq 3$, the diameter satisfies $\delta = \Omega(\log N)$, meaning that hierarchical graphs have optimal diameter up to a constant factor.
Tighter bounds on the number of nodes may be shown, for instance, when the \emph{tree-width} of the graph is bounded.
%The tree-width of a graph may be defined as follows (see Ref. \cite{Robertson1990}):
%\begin{definition}
%\label{def:treewidth}
%Let $G=(V,E)$ be a graph, and let $\mc{X} = \curly{X_i: X_i\subseteq V, i=1,\ldots, m}$ be a collection of vertex subsets of $G$. Let $T$ be a tree defined on nodes labeled by the subsets $X_i\in \mc{X}$, and satisfying the following properties:
%\begin{enumerate}
%\item $\bigcup_{i=1}^{m} X_i = V$.
%\item If $v\in X_i$, and $v\in X_j$, then $v\in X_k$ for all $X_k$ lying on the unique path between $X_i$ and $X_j$ in the tree $T$.
%\item For all $(v,w)\in E$, there exists $X_i\in \mc{X}$ such that $v\in X_i$ and $w\in X_i$.
%\end{enumerate}
%Then, we say that $T$ is a valid tree decomposition of width $w_T = \max_{i=1}^m|X_i|-1$. The \emph{tree-width} of graph $G$ is defined as the minimum width over all valid tree decompositions,
%\begin{equation}
%  \label{eq:deftreewidth}
 % t(G) = \min_{\text{all } T}w_T.
%\end{equation}
%\end{definition}
Ref. \cite{Pineda-Villavicencio2013} shows that graphs with small tree-widths $t$ and an odd diameter $\delta$ satisfy
\begin{equation}
  \label{eq:nodeCap}
  N\paren{\Delta,\delta;t} \sim t\paren{\Delta - 1}^{\frac{\delta-1}{2}}.
\end{equation}
As discussed towards the end of Sec.\ \ref{sec:struc}, hierarchies have low tree-widths. In particular, the tree-width  of the truncated hierarchy $K_n^{\twhp k}$ is at most $n-1$.
Next, the diameter of the truncated hierarchy $K_n^{\twhp k}$ is $\delta(k) = 2k-1$ (which is odd), and the maximum degree is $\Delta(k) = 2(n-1)$. Comparing the number of nodes in this hierarchy $N(k)$ to the node capacity $N\paren{\Delta(k),\delta(k);n-1}$ as in Eq.\ (\ref{eq:nodeCap}), we get
\begin{equation}
\frac{N(k)}{N\paren{\Delta(k),\delta(k);n-1}} \gtrsim \frac{n^k}{(n-1)\paren{2n - 3}^{k-1}} \label{eq:ratio}.
\end{equation}

Keeping the total number of nodes $N$ fixed, consider two limits: one, a shallow hierarchy in which the number of levels $k$ is $O(1)$, and two, a deep hierarchy, in which the size $n$ of the base graph is $O(1)$ [i.e., $k = O(\log N)$].
We see that when the hierarchy is shallow, the right side of Eq.\ (\ref{eq:ratio}) is $\Theta(1)$, which indicates optimality.
For a deep hierarchy, the above ratio scales as $2^{-\log_n N}= N^{\frac{-1}{\log(n)}}$, which is polynomially suboptimal. However, when $n=3$, the ratio in Eq.\ (\ref{eq:ratio}) is again $\Theta(1)$, and the truncated hierarchy $K_3^{\twhp k}$ is degree-diameter optimal in this case.

\section{Entangled State Construction} 
\label{sec:construction}
\subsection{Setup}
Although some of the graph properties calculated in the previous section give a heuristic sense for the capabilities of the hierarchical graph versus the nearest-neighbor or all-to-all graphs, we would like to examine their performance directly in terms of a quantum information processing task. The task we have chosen as a benchmark is the creation of a many-qubit GHZ state. Since this entangled state is difficult to create across long distances when using nearest-neighbor interactions, we hope that it can serve as a useful yet basic benchmark for processing quantum information with unitary evolution \cite{Eldredge2017}. As shown in Ref.~\cite{Eldredge2017}, preparation of a GHZ state also provides a means of transferring a state across the graph.  Thus, the results of this section also bound state transfer time.
However, in this work, unlike Ref.~\cite{Eldredge2017}, we focus on the use of discrete unitary operations (gates) rather than Hamiltonian interactions. This means that we cannot take advantage of the many-body interference which provided a speed-up in Ref.~\cite{Eldredge2017}.

We adopt a framework in which every vertex of the graph represents one logical qubit, while an edge of the graph represents the ability to perform a two-qubit gate between nodes. 
For the purposes of this work, we assume that we can ignore single-qubit operations, instead focusing on the cost imposed by the required two-qubit gates between nodes.

\subsection{Analytical Results for Deterministic Entanglement Generation}
\label{sec:deterministic}
In order to create the GHZ state, we assume that we begin with all qubits in the state $\ket{0}$ except for one qubit that we place in the initial state $\ket{+}$. By performing controlled-NOT operations between this qubit and its neighbors, a GHZ state of those qubits is created. The state can be expanded by continuing to use further CNOT operations to expand the ``bubble'' of nodes contained in the GHZ state until it eventually spans the entire graph. For state transfer, we instead assume the initial state $\ket{\psi}$ to be transferred sits on one qubit, which is then transferred through the graph using SWAP operations until it reaches its destination.

We first consider a graph which has been assigned time weights, so that a gate between two linked edges can be performed deterministically in a time given by the weight of the edge between them. We assume that one node can act as the control qubit for several CNOT operations at once. Therefore, according to our protocol above, the time $t_\mathrm{GHZ}$ required to construct the GHZ state is found by identifying the qubit that will take the longest to reach from the initial qubit by hopping on the graph. A similar argument holds for the state transfer time.

This implies that a GHZ state can be created, or a state transferred, in time that scales like the (time-)weighted eccentricity of the node we choose as the initial $\ket{0} + \ket{1}$ state. However, if we take the further step in our analysis of maximizing over weighted eccentricities (identifying the worst-case starting node), then the time will simply be the weighted diameter of the graph as calculated in the previous section. Note that the difference between the best-case weighted eccentricity (the weighted graph radius) and the worst-case weighted eccentricity (the weighted graph diameter) over all nodes is at most a factor of two -- if we look at the midpoint of the path that realizes the graph diameter, its distance to the endpoints of the path is bounded by the radius -- so from the perspective of how this time scales asymptotically with $N$, the two are interchangeable.

\subsection{Numerical Results for Probabilistic Entanglement Generation}
\label{sec:probabilistic}
As shown in the previous subsection, in a deterministic setting of entanglement generation where a gate between two nodes of our graph $H$ can be performed in fixed time, the time required to create a GHZ state is equal to the weighted diameter $\delta_T(H)$. However, in many situations in long-distance quantum information processing, probabilistic or heralded methods might be used instead. We might suppose that, in a small time step, the network succeeds  in performing a desired two-qubit gate with probability $p$ (and that we know whether the gate succeeded or not). Upon failure, one can try performing the gate again in the next time step without having to rebuild the state from the beginning. In this setting, we expect that the scaling will likely be similar to the deterministic case but more difficult to calculate exactly. Fortunately, it is easy to re-interpret the meanings of the edge weights to account for this. 

The main complication arising from the inclusion of unitaries that do not get completed in a fixed amount of time is that multiple paths between two nodes can all contribute to the total probability that entanglement has been produced, making it a harder problem to solve exactly. However, we can turn to numerical simulation to get an idea of the behavior. In the following, we define a new edge weight called the probability weight, $p_{ij}$, which is the probability of success of edge $(i,j)$ in one time step.

The algorithm for simulating the creation of a GHZ state is as follows:
\begin{itemize}
	\item At each time step $t$, identify the subgraph $F$ of nodes that have already joined the GHZ state.
	\item For each edge between a GHZ node $i \in F$ and a non-GHZ node $j \notin F$, identify the probability edge weight $p_{ij}$. With probability $p_{ij}$, allow node $j$ to join the GHZ state in the current time step, $t$.
	\item Once all edges have been tested, repeat the procedure for the next time step on the new, possibly larger, set of GHZ nodes.
\end{itemize}

A single number $p_0$ is chosen as the base probability, so that the probability weights on the lowest level are $p_0$, and edges on the $i$-th level of the hierarchy succeed with probability $p_0 \alpha^{i-1}$. Note that we must fix $\alpha < 1$.  
As a first step toward evaluating the performance of a graph, we estimate its time weights as $w_{ij} = 1/p_{ij}$, the time required to perform a two-qubit unitary on average. The overall estimate of the expected time taken is then $\delta_T/p_0$, where $\delta_T$ is the time taken for the deterministic case with time weights scaling by a factor $\beta = 1/\alpha$ at each level.
We find that this predicts very well the rate at which the GHZ state can be constructed over a wide range of $\alpha$ values (Fig.~\ref{fig:numdata}). The expected time remains $\Theta \left(N^{\log_n (1/\alpha)} \right)$. 

For graphs with multiple potential paths between two nodes, such as a two-dimensional grid, the expected time is not simply the deterministic time scaled by the extra time factor the probabilistic setup requires in each step. 
We can however still bound the expected time to build the GHZ state $\mathbb{E} [t_\mathrm{GHZ} ]$ above and below for a graph $H$. We will bound it above by considering a modified graph in which the only path between the initial qubit and the qubit farthest from the starting point has distance $d_w (H)$. Such a path completes in time $d_w(H)/p_0$ on average. Since $H$ has strictly more paths than this, the expected time will be lower. However, the shortest path between the initial and final qubits has total distance $d_w (H)$, which would take time $d_w(H)$ to complete even if $p_0 = 1$ and all gates were deterministic. Therefore, no path can finish faster than this, and the expected outcome over all possible paths cannot improve over $d_w(H)$. We can therefore write the following restriction on the expected time:
\begin{equation}
  d_w(H) \leq \mathbb{E}\left[t_\mathrm{GHZ} \right] \leq \frac{d_w(H) }{p_0},
\end{equation}
where $\mathbb{E}[\cdot]$ denotes the expected value. This implies $ \mathbb{E}\left[t_\mathrm{GHZ}\right] = \Theta \left( d_w(H) \right)$.
Therefore, although the prefactor is difficult to calculate, we can tell that the time required to complete the creation of a GHZ state on the nearest-neighbor graph with $d=2$ is $\Theta(\sqrt{N})$. This scaling implies that the condition for the hierarchy to outperform the nearest-neighbor grid in 2D is $\alpha \geq n^{-1/2}$, which is reflected in Fig.~\ref{fig:numdata}.

\begin{figure}[htbp]
	\centering
	\includegraphics[width = .45\textwidth]{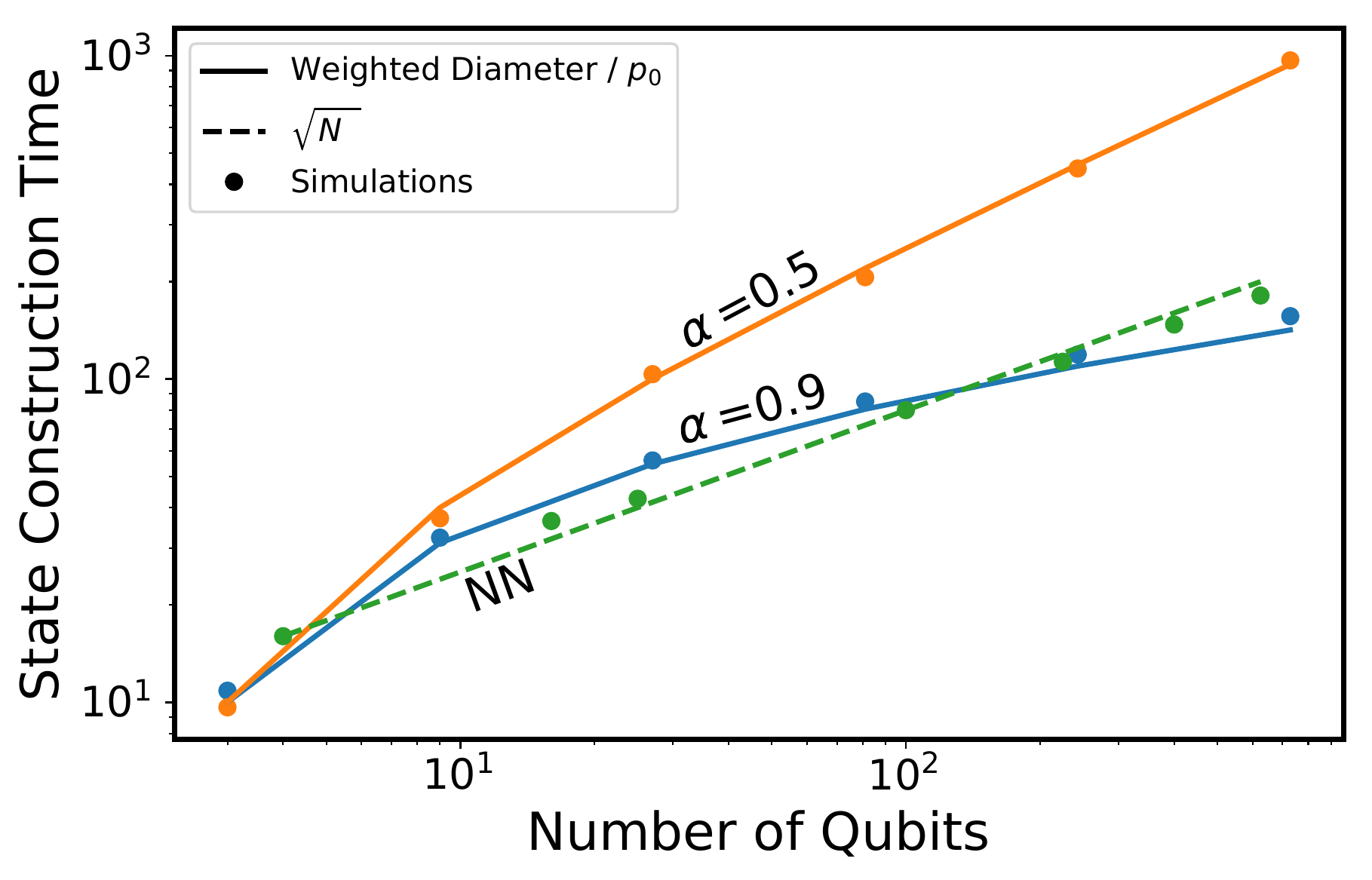}
	\caption{Graph-theoretic predictions and simulation of $t_\mathrm{GHZ}$ for the hierarchy $K_3^{\whp k}$ at various $\alpha$, and a two-dimensional nearest-neighbor (NN) grid; $p_0 = 0.1$. The $\sqrt{N}$ fit shows the scaling of $t_\mathrm{GHZ}$ for the nearest-neighbor case, with a prefactor in the range suggested by the text's argument. Note that since $n = 3$, the crossover for the hierarchy to asymptotically outperform the nearest-neighbor grid is at $\alpha \geq 1/\sqrt{3} \approx 0.58$, which is seen in the numerical results. Code for generating this figure can be found at \cite{ourgithub}.}
	\label{fig:numdata}
\end{figure}

Using the GHZ-creation time and state transfer as examples, we can see many of the advantages of hierarchical graphs as network topologies. Such architectures are able to rapidly incorporate a very large number of qubits (exponential in the number of hierarchy levels), while the time-weighted diameter (and thus communication time) grows linearly with the number of levels. Since the weighted diameter is not substantially changed even if we use the truncated hierarchical product of Sec.~\ref{sec:flexible}, these benefits can also be realized in that setup. 
\section{Circuit Placement on Hierarchies}
\label{sec:placement}
A final reason we believe hierarchies could be a useful way to organize modular quantum systems is that they may be able to take advantage of straightforward methods for circuit placement. Circuit placement is a problem that arises when a quantum circuit or algorithm must be translated onto a physical system \cite{Maslov2008}. Suppose we are given a specification for a quantum algorithm in the form of a circuit diagram, and we wish to run that algorithm on a given quantum computer (which presumably has enough quantum memory to perform that algorithm). In order to translate the circuit into instructions for our machine, we must identify each algorithm qubit with a machine qubit and then determine how the individual quantum gates can be realized in our machine \footnote{We studiously avoid referring to the machine qubits as ``physical'' in this paper, as we do not want to confuse this conceptual distinction with the physical/logical qubit divide in error correction. All of the qubits referenced in this section are logical qubits in the error-correcting sense.}.

Circuit placement is an important part of the quantum software stack, just as the compilation to machine code is in classical computers. By placing qubits which must operate on each other often close together in the real-world machine, we can minimize the amount of time spent performing long-range quantum gates. However, this problem is generally quite difficult for arbitrary instances and in fact has been shown to be NP-complete \cite{Maslov2008}. 

However, since we are interested in the sub-problem of circuit placement on hierarchies, it is possible that the hardness results of Ref.~\cite{Maslov2008} do not apply and the exact solution can be found in polynomial time, just as the problem can be solved tractably in linear qubit chains \cite{Pedram2016}. Whether or not an exact algorithm exists, we can appeal to heuristics to efficiently place circuits as well as possible. Such an approach is promising because hierarchies are extremely structured with clear prioritization of clustering between small groups of qubits, which can be recognized in the algorithm and matched to the physical architecture.

To explain further, we consider the following model. We suppose that we begin with a weighted circuit graph $C$ with a vertex set $V_C$ and an edge set $E_C$, in which an edge exists between two vertices if there is at least one two-qubit gate between them in the circuit, with the weight of the edge corresponding to the number of gates. We then specify a machine graph, $M$, with vertex set $V_M$ and edge set $E_M$, in which each edge $(u,v)$ indicates that the machine can perform two-qubit gates between $u$ and $v$.

We now seek a mapping $f: V_C \to V_M$ that assigns algorithm qubits to machine qubits. A mapping $f$ has a total cost found by considering, for every edge in $E_C$ between vertices $c_i$ and $c_j$, the shortest-path distance between $f(c_i)$ and $f(c_j)$ in $M$, multiplying that distance by the weight of the edge in $C$ and summing over all edges. Thus, it captures the total distance that must be traversed by all gates in order to execute the circuit when the current mapping is used. Reducing this is expected to reduce the amount of time spent performing SWAP gates in order to connect two distant qubits. Performing this mapping is an important subroutine in any quantum programming framework, and at least one existing quantum compiler has a ``mapper'' phase that takes into account the actual graph that a program must be compiled onto \cite{Haner2016, Steiger2018}. 

Our cost function is a choice made from convenience, and others are possible. Using this cost function ignores several important aspects of quantum circuits. First, our cost function does not account for the fact that a different mapping might allow for more parallelism, since it evaluates the cost of each gate individually. In addition, we take the circuit graph $C$ as a given, when in fact many different circuits exist for any given quantum operation. In fact, it is likely that optimization of $C$ could be performed, possibly by using the structure of $M$ itself. A more realistic model for circuit placement may require a back-and-forth in which a circuit is first placed, then optimized, then re-placed, and so forth.  A more advanced placement algorithm may even permit the swapping of qubits throughout the circuit, thus optimizing the placement of the quantum algorithm without constructing a circuit connectivity graph as an intermediate step.

For this paper, we will ignore these concerns and proceed with a heuristic approach to circuit placement for hierarchies. We describe our algorithm as ``partition and rotate,'' as it requires these two basic subroutines. First, qubits are partitioned into sub-hierarchies by examining whether they are connected by many gates in $C$. This process continues recursively, with each partition being subdivided and so on until every qubit is identified with its point in the hierarchy. This top-down process is then followed by a bottom-up process in which each small cluster is rotated so that its most-communicative qubit is at the root of the sub-hierarchy, and then the partitions themselves are rotated, and then clusters of clusters, etc. Ideally, this results in a mapping in which every qubit is (a) placed close to qubits it needs to communicate with and (b) placed in easy access to other modules if that qubit requires such access. We will now explore in detail these subroutines and the circuit speed-ups that result. We will place algorithms on a machine graph $M$ which we take to be defined by $K_n^{\uhp k}$ for some integer $k$. Note that we examine unweighted hierarchies, but these methods can be applied to weighted hierachies as well.

\begin{figure*}[tb]
	\centering
	\includegraphics[width=.8\textwidth]{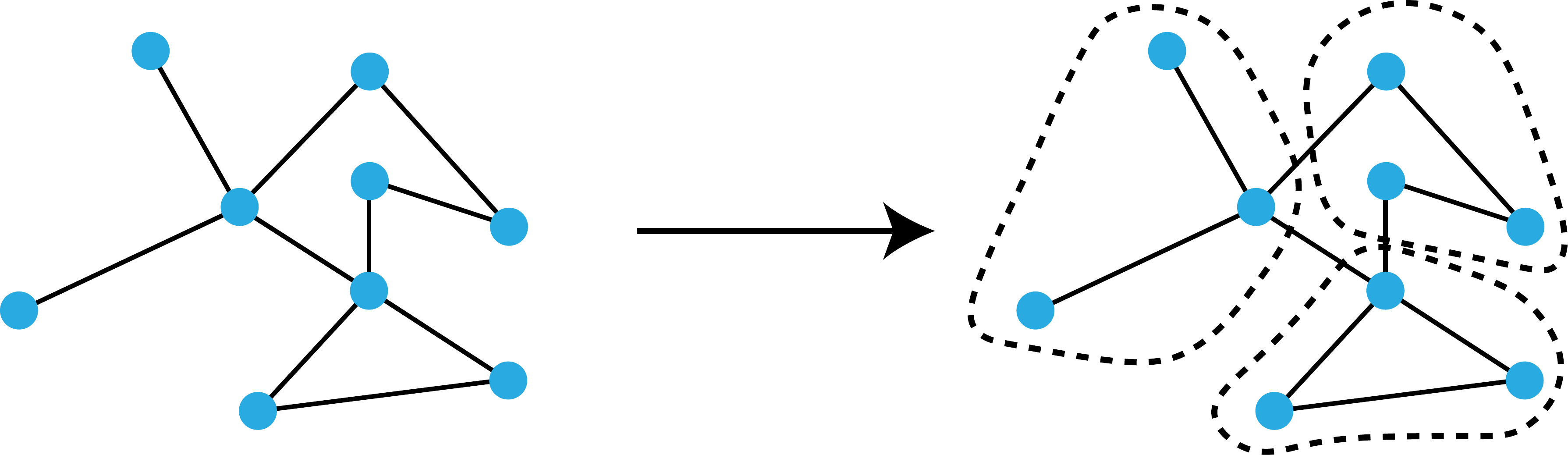}
	\caption{Illustration of how we might divide a hypothetical graph into smaller clusters. This process is repeated many times, recursively.}
	\label{fig:cluster}
\end{figure*}
\subsection{Partitioning}
For the first step of our algorithm, we wish to divide the computational graph $C$ into $n$ subgraphs which are as disconnected as possible. In addition, since we wish to assign each node in $C$ to physically separate and limited qubit registers, it is important that each of the subsets has precisely $\abs{C}/n$ nodes. This problem is known as \textit{balanced graph partitioning}, and the problem of finding the optimal solution is NP-complete for $n \geq 3$ \cite{Andreev2006}. However, heuristic methods exist which approximate the solution, and are widely used in the field of parallel computing and circuit design \cite{Karypis1998}. We have illustrated this process in \cref{fig:cluster}.

Our method for performing circuit placement on hierarchies relies on a subroutine that performs balanced graph partitioning.  There are many algorithms and software packages from which to choose.  Here, we have used a software package called Metis, which implements an algorithm called recursive bipartitioning \cite{Karypis1998}.

We begin by supposing that we have the circuit graph $C$ and we wish to identify groups of $\abs{C}/n$ nodes which have high connection to each other but low connection outside of the group. This is accomplished by finding a balanced graph partition in which the weight of the edges connecting each group is minimized. If we call the initial set of all nodes $S$, then we wish to identify subsets $S_0, S_1, \dots, S_n$. In terms of the addressal scheme of Sec.~\ref{sec:addressal}, all the nodes in set $S_i$ will have have digit $i$ in their base-$n$ representation. In the next section, we will discuss the choice of which digit to assign to each set.

Once the subsets $S_i$ are found, partitioning can be run again on that relevant subgraph, creating $n$ new subsets of this subset. Eventually, every node in the graph will be identified with a lowest-level module of size $n$, a  next-level module of size $n^2$, and so forth.

Here we have used a generalized, pre-existing algorithm for graph partitioning. It is possible that the specifics of this problem, and the possibility of co-designing the precise quantum circuit implementing the algorithm (and thus $C$) with the architecture, enable more specific, better-performing approaches. 
\begin{figure*}[tb]
	\centering
	\includegraphics[width=.8\textwidth]{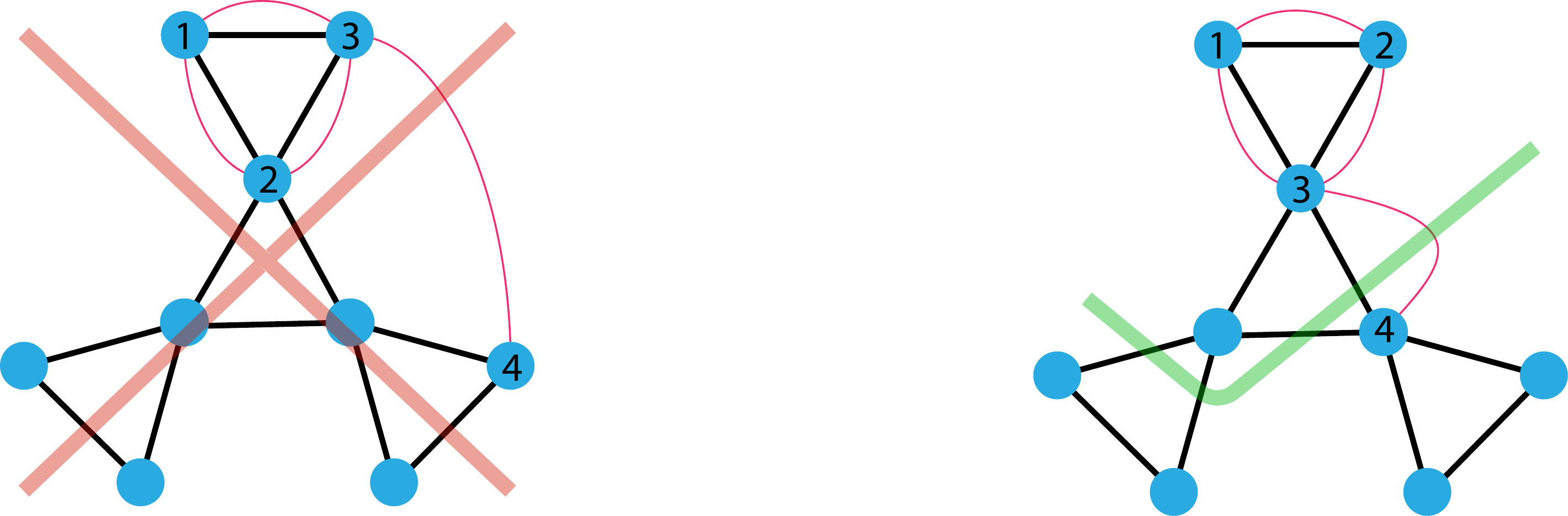}
	\caption{An illustration of how and why the process of rotation works in our circuit placement algorithm. In this diagram, red links represent gates to be performed (edges in $C$) and black ones are available communicative links (edges in $M$). In the graph $C$, the qubits 1, 2, and 3 are all connected, and 3 is connected with 4. These qubits have been correctly placed into clusters $(1, 2, 3)$ and $(4)$. However, if they are not rotated correctly (see left), the link between 3 and 4 can become quite long, necessitating a long-range quantum gate. By properly rotating (right), the gate between links 3 and 4 becomes much shorter, improving the placement.}
	\label{fig:rotation}
\end{figure*}
\subsection{Rotation}
Drawing partitions between qubits is not enough to fully specify their placement into a hierarchy. If we consider using the $i$-digit representation, we can imagine that partitioning essentially describes the process of deciding, from a set of qubits, which ones will share a digit in the next level. However, these digits are more than arbitrary markers, because there is one node in any sub-hierarchy which connects to the hierarchy above. This node (which we say has digit 0) has privileged access to communication with other sub-hierarchies. Therefore, in order for our circuit placement to succeed, we should ensure that the qubit on top of each sub-hierarchy is the one which requires the most access.

In order to do this, we implement a second subroutine, the ``rotate'' part of the algorithm. This is called rotation because, once we know which qubits will be together in a module, we must choose how to orient them relative to the larger modular structure. Whereas partitioning is top-down (the full graph is broken into small subgraphs which are then themselves partitioned), rotation is bottom-up. Suppose the modular structure is $K_n^{\uhp k}$. We begin with sets of $n$ qubits and must choose which will be the top of each smallest instance of $K_n$. We then take each partition of $n$ instances of $K_n$ and decide which instance of $K_n$ will connect to the next level up, and so on. This process is illustrated in Fig.~\ref{fig:rotation}. 

Note that the general structure of our algorithm is to first go down the hierarchy, partitioning nodes, and then to go up, re-arranging sub-hierarchies in the proper order. We perform this procedure only once to obtain our circuit mapping.

\subsection{Results}
Now that the placement algorithm is specified, we turn toward examining its performance on quantum circuits. We consider two separate questions. First, we investigate whether the algorithm is effective -- does it actually reduce, relative to a random assignment, the amount of distance that must be traversed in a circuit to execute all the requested gates? Second, we will examine whether the algorithm executes efficiently on a classical computer. This second point is important because in general the problem can be solved by brute-force search, but such a search requires a time $\mathcal{O}\left(N!\right)$ to perform (although, as we stated earlier, it is possible that an exact algorithm exists with a lower time cost for the special case of hierarchies).

\begin{figure}[htpb]
	\centering
	\includegraphics[width=.45\textwidth]{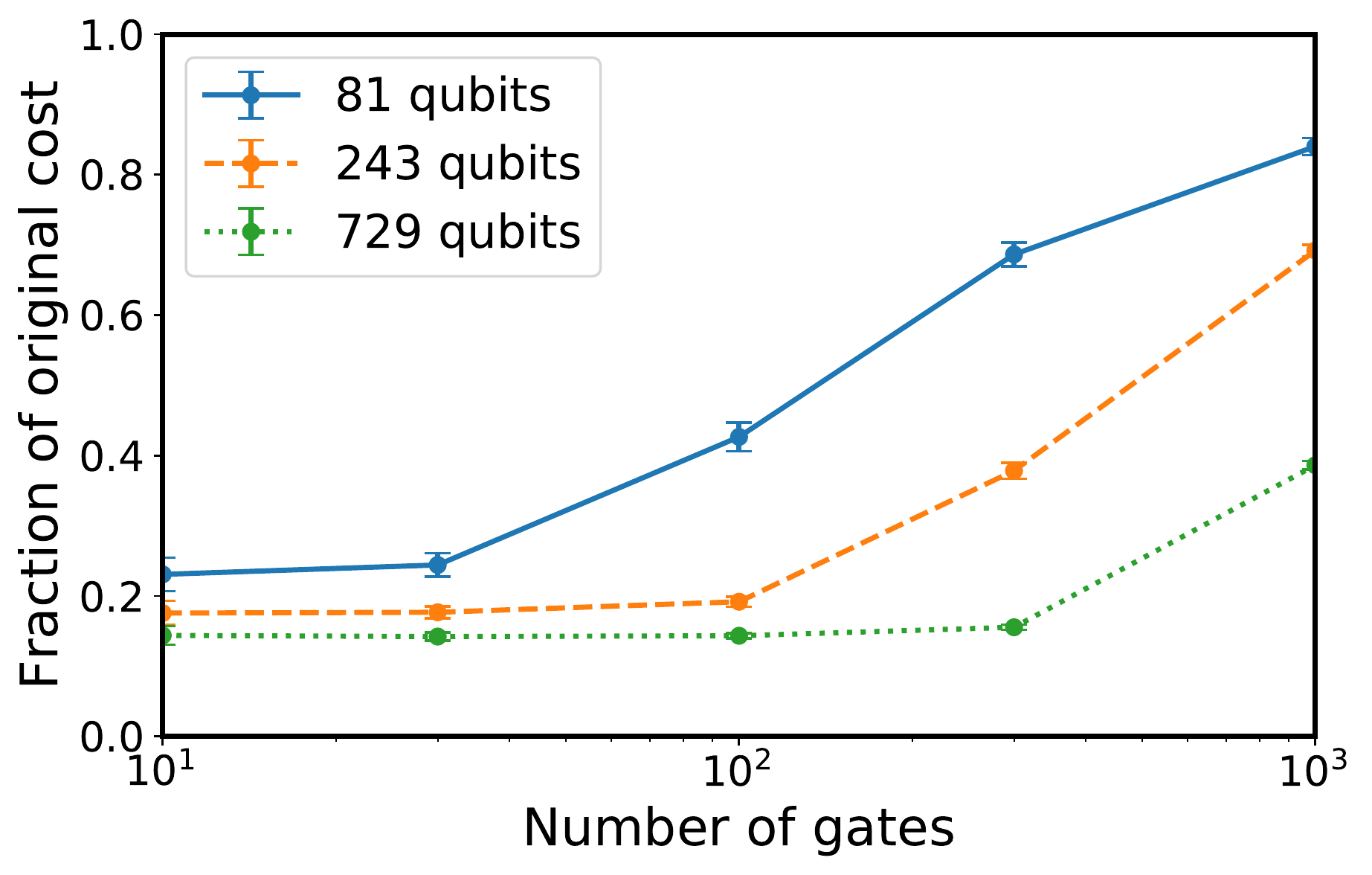}
	\caption{Plot of the average ratio (total gate distance after partition-and-rotate)/(total gate distance before) given 100 trials each for different numbers of random gates and random qubits. Error bars represent one standard deviation. As the number of gates begins to saturate the number of qubits, the possible improvement from optimization begins to decrease.}
	\label{fig:crperformance}
\end{figure}

To investigate the above concerns, we examine the algorithm's performance on random circuits. For each trial, we first generate a random circuit of $N_g$ two-qubit gates on $N$ total qubits.  The precise type of two-qubit gate is irrelevant in this framework. Likewise, single-qubit gates require no communication overhead, so we do not consider them.  The random circuit then implies a computational graph $C$, where, as described above, the vertices represent the algorithm qubits and the edge weights represent the number of gates that must be applied between each pair of qubits.  Once this computational graph has been generated, we first attempt to map it blindly to the hierarchy graph, using the addressing scheme of Sec.~\ref{sec:addressal} and an arbitrary order of the graph $C$. Then, we apply partition-and-rotate and calculate the new cost function. By comparing the cost function between these two, we develop an idea of how much long-range quantum information processing is eliminated by partition-and-rotate. We perform this several times to build up statistics on average time costs and average improvement. Code which performs circuit placement and generates the profiling figures included in this section can be found at \cite{ourgithub}.

In our simulations, we test hierarchies $K_3^{\uhp k}$ up to 729 qubits ($k=6$). We find that as gates are added, the improvement over the initial cost is decreased. This is sensible, because as more randomly placed gates are present, different node mappings become more similar. Such an effect will likely not be present for quantum algorithms which do not have their gates placed randomly. For cases in which the number of gates is significantly fewer than the number of qubits, partition-and-rotate is able to significantly reduce the cost function. We find that 100 gates can be placed on a 729 qubit hierarchy with a total cost less than 20\% of the original on average. When 1000 gates are placed on a 729 qubit hierarchy, the final cost is still only 40\% of the initial one. Results for $K_3^{\uhp 4}$, $K_3^{\uhp 5}$, and $K_3^{\uhp 6}$ can be seen in Fig.~\ref{fig:crperformance}.

\begin{figure}[b]
		\includegraphics[width=.45\textwidth]{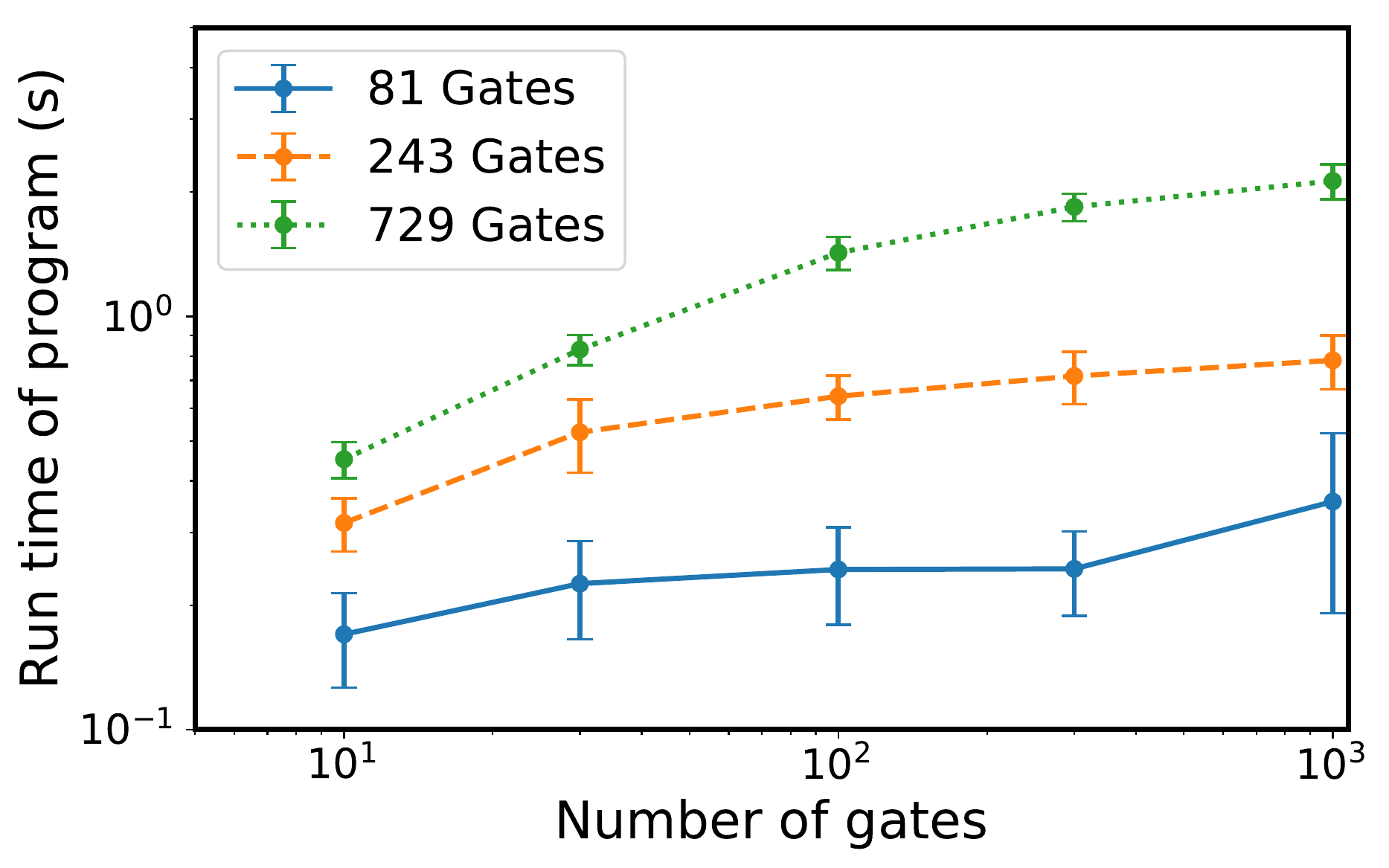}
		\caption{Average run times over 100 trials for partition-and-rotate on a 2015 MacBook Pro with a 2.6 GHz processor. Each line represents an increasing number of gates for a constant circuit size as measured by the number of qubits.}
		\label{fig:gate_time_scaling}
\end{figure}
\begin{figure}[htpb]
		\includegraphics[width=.45\textwidth]{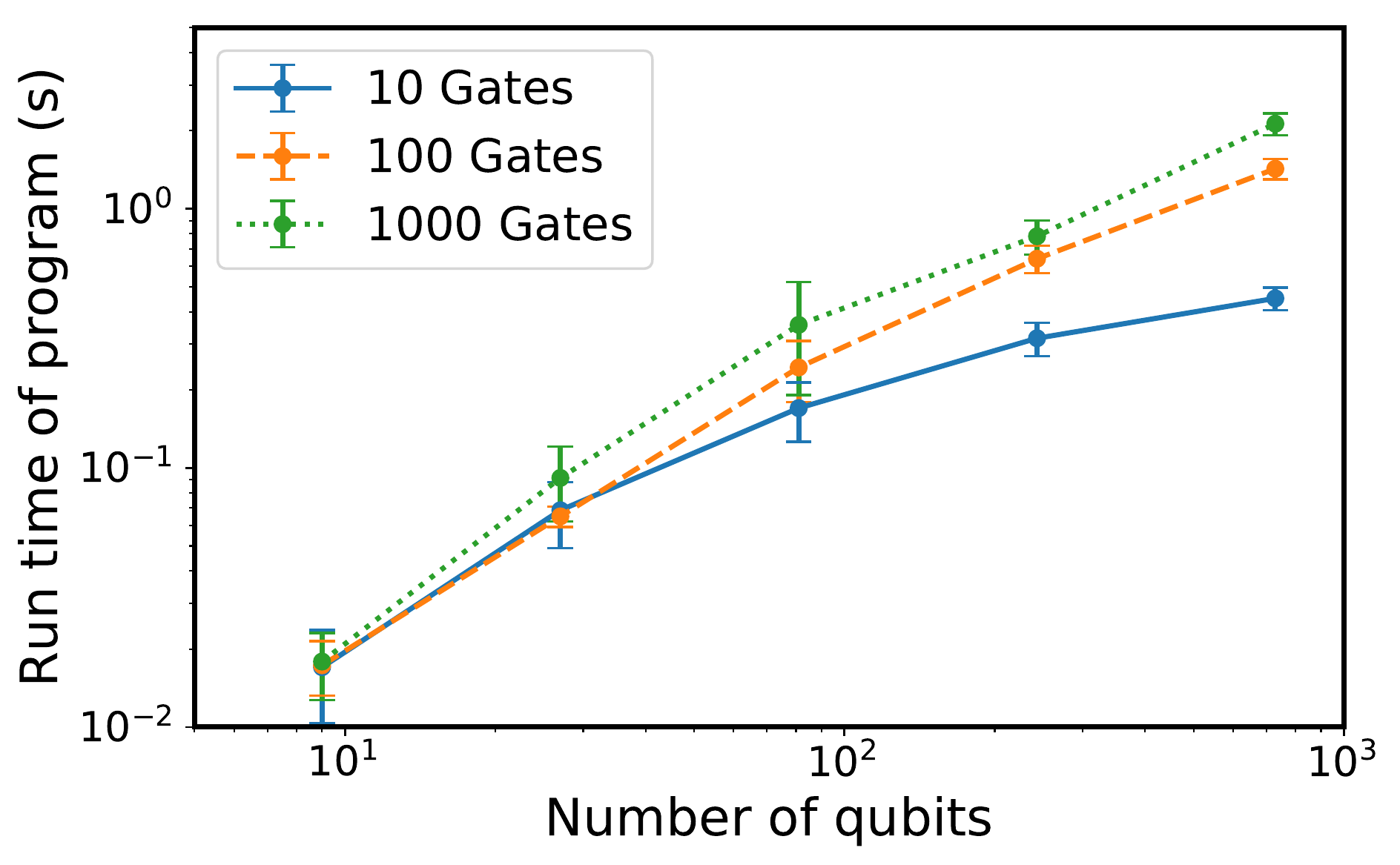}
		\caption{Average run times over 100 trials for partition-and-rotate on a 2015 MacBook Pro with a 2.6 GHz processor. Each line represents an increasing number of qubits for a constant number of gates.}
		\label{fig:qubit_time_scaling} 
\label{fig:times}
\end{figure}

Next, we examine the time required to place such a circuit. Our code, most of which is written in Python3 but which uses a C implementation of Metis for graph partitioning, can place 1000 gates on a 729-qubit hierarchy in roughly two seconds when running on a 2015 MacBook Pro. Although the algorithm seems naturally suited to parallelization, our implementation uses only a single core. Our current implementation appears to scale with the number of qubits as $\mathcal{O}(N)$ and not to depend on the number of gates included at all once there are a sizable number of gates. We illustrate these two relationships in Figs.~\ref{fig:gate_time_scaling} and \ref{fig:qubit_time_scaling}. These times compare favorably to the times reported in Ref.~\cite{Maslov2008}, with much optimization still possible in our implementation.

Note that using random graphs as described above means that our results may not be valid for more general quantum algorithms. It is possible that practical quantum algorithms have structure that makes them either particularly amenable or particularly difficult for partition-and-rotate algorithms to place, depending on the actual algorithm being examined.

\section{Conclusions and Outlook}
 
In this paper, we have developed the theory of hierarchies using the existing binary operation of the graph hierarchical product. We have shown that hierarchies may be a promising architecture for large quantum information processing systems. To demonstrate this, we analyzed both properties of the underlying graph (such as diameter, maximum degree, total edge weight) as well as the time it would require to perform a representative quantum information process (constructing the GHZ state/state transfer) in both deterministic and probabilistic settings. We have also computed and tabulated these properties for many other graphs which appear as potential architectures, for comparison. We have shown that, for much of parameter space, hierarchies have favorable scalings in cost and performance with the total number of qubits $N$ compared to these competitors. Also, since hierarchical graphs are hyperbolic, they share many of the advantages of hyperbolic graphs such as efficient routing schemes \cite{Kleinberg2007}, network security \cite{Jonckheere2004}, and node addressal \cite{Chepoi2012}. 

We have also presented a conceptually simple circuit placement algorithm which allows for simple optimization using existing graph-partitioning software packages. Our partition-and-rotate algorithm scales well with the number of qubits and gates in the circuit and reliably reduces the total distance that needs to be traversed by random quantum circuits, which we verified by simulation.

One significant limitation of our analysis in this paper has been that we remained confined to unitary operations. Non-unitary operations (for instance, measurements which are then fed forward to choose future unitary operations) are capable of establishing long-range correlations like those present in the GHZ state much more quickly than unitary ones if measurements and classical communication are fast. In the future, we hope to extend our results into non-unitary domains \cite{UsUpcoming}.
 
In addition, we have made the assumption that the primary way in which quantum architectures will differ is the speed with which two qubits can communicate (as represented by our time weights on edges). Another important case might be one in which the primary way edges are enhanced is by improving bandwidth or duplicating nodes to provide parallel routes rather than affecting gate speed directly. For some schemes, our abstract notion equating the time of a two-qubit gate with the edge weight may still be a useful tool of analysis, but in other cases bandwidth and speed may not be interchangeable. We intend to undertake the analysis appropriate for this case in a future manuscript \cite{UsUpcoming}.

In this paper, we limited ourselves to consideration of a few quantum processes (generation of a large entangled state, or transfer of a state across the graph), which might not be representative of other, more general distributed quantum information tasks.
Some algorithms, such as Shor's algorithm, are known to be able to run with little additional overhead even on one-dimensional, nearest-neighbor graphs \cite{Fowler2004}. Therefore, when selecting an architecture for a practical quantum computer, care will need to taken to select the proper benchmarking task.

In future work, we hope to look at a wider variety of quantum circuits and use those to better benchmark different modular architectures. In addition, we hope to gain a better understanding of the treatment of probabilistic links for general graphs. For instance, as we discussed briefly when assessing the star graph $S_N$, one real concern in a networked setting is whether some parts of the network will form bottlenecks. To analyze the impact of this in a general way will require a better understanding of realistic quantum algorithms and the demands they place on a network. Analyzing more complex quantum algorithms could also shed light on the performance of partition-and-rotate placement algorithms in realistic settings when sequencing and scheduling also enter into consideration.

Finally, in addition to asking ourselves how current circuits and algorithms can be executed on highly modular systems, we also hope to explore the possibility that highly modular architectures open up new possibilities for parallelized quantum algorithms. For instance, Ref.~\cite{Hoyer2005} shows that quantum fan-out gates can be used to parallelize gate sequences, decreasing the time to perform an algorithm at the cost of requiring additional memory qubits. Hierarchies could implement such schemes by using high-level connections to perform the initial fan-out gates and then performing the various parallelized operations in each individual module.

\acknowledgments
We thank A.\ Childs, L.\ Jiang, C.\ Monroe, A.\ Houck, K.\ Hyatt, I.\ Kim, A.\ Koll\'ar, M.\ Lichtman, Y.-K.\ Liu, D.\ Maslov, E.\ Schoute, A.\ Wallraff, and D.\ Wentzlaff for discussions. This work was supported by ARO MURI, ARL CDQI, NSF Ideas Lab, ARO, NSF PFC at JQI, and the Department of Energy ASCR Quantum Testbed Pathfinder program. A.~B.\ is supported by the QuICS Lanczos Fellowship. Z.~E.\ is supported in part by the ARCS Foundation.  J.~R.~G. is supported by the the NIST NRC Research Postdoctoral Associateship Award. F.~C. is supported by NSF Expeditions in Computing grant 1730449, Los Alamos National Laboratory and the U.S. Department of Defense under subcontract 431682, by NSF PHY grant 1660686, and by a research gift from Intel Corporation. This work was performed in part at the Aspen Center for Physics, which is supported by National Science Foundation grant PHY-1607611.

\bibliography{References}

\begin{thebibliography}{67}%
\makeatletter
\providecommand \@ifxundefined [1]{%
 \@ifx{#1\undefined}
}%
\providecommand \@ifnum [1]{%
 \ifnum #1\expandafter \@firstoftwo
 \else \expandafter \@secondoftwo
 \fi
}%
\providecommand \@ifx [1]{%
 \ifx #1\expandafter \@firstoftwo
 \else \expandafter \@secondoftwo
 \fi
}%
\providecommand \natexlab [1]{#1}%
\providecommand \enquote  [1]{``#1''}%
\providecommand \bibnamefont  [1]{#1}%
\providecommand \bibfnamefont [1]{#1}%
\providecommand \citenamefont [1]{#1}%
\providecommand \href@noop [0]{\@secondoftwo}%
\providecommand \href [0]{\begingroup \@sanitize@url \@href}%
\providecommand \@href[1]{\@@startlink{#1}\@@href}%
\providecommand \@@href[1]{\endgroup#1\@@endlink}%
\providecommand \@sanitize@url [0]{\catcode `\\12\catcode `\$12\catcode
  `\&12\catcode `\#12\catcode `\^12\catcode `\_12\catcode `\%12\relax}%
\providecommand \@@startlink[1]{}%
\providecommand \@@endlink[0]{}%
\providecommand \url  [0]{\begingroup\@sanitize@url \@url }%
\providecommand \@url [1]{\endgroup\@href {#1}{\urlprefix }}%
\providecommand \urlprefix  [0]{URL }%
\providecommand \Eprint [0]{\href }%
\providecommand \doibase [0]{http://dx.doi.org/}%
\providecommand \selectlanguage [0]{\@gobble}%
\providecommand \bibinfo  [0]{\@secondoftwo}%
\providecommand \bibfield  [0]{\@secondoftwo}%
\providecommand \translation [1]{[#1]}%
\providecommand \BibitemOpen [0]{}%
\providecommand \bibitemStop [0]{}%
\providecommand \bibitemNoStop [0]{.\EOS\space}%
\providecommand \EOS [0]{\spacefactor3000\relax}%
\providecommand \BibitemShut  [1]{\csname bibitem#1\endcsname}%
\let\auto@bib@innerbib\@empty
%</preamble>
\bibitem [{\citenamefont {Van~Meter}\ and\ \citenamefont
  {Itoh}(2005)}]{VanMeter2005}%
  \BibitemOpen
  \bibfield  {author} {\bibinfo {author} {\bibfnamefont {R.}~\bibnamefont
  {Van~Meter}}\ and\ \bibinfo {author} {\bibfnamefont {K.~M.}\ \bibnamefont
  {Itoh}},\ }\href {\doibase 10.1103/PhysRevA.71.052320} {\bibfield  {journal}
  {\bibinfo  {journal} {Phys. Rev. A}\ }\textbf {\bibinfo {volume} {71}},\
  \bibinfo {pages} {052320} (\bibinfo {year} {2005})}\BibitemShut {NoStop}%
\bibitem [{\citenamefont {Meter}\ \emph {et~al.}(2008)\citenamefont {Meter},
  \citenamefont {Munro}, \citenamefont {Nemoto},\ and\ \citenamefont
  {Itoh}}]{Meter2008}%
  \BibitemOpen
  \bibfield  {author} {\bibinfo {author} {\bibfnamefont {R.~V.}\ \bibnamefont
  {Meter}}, \bibinfo {author} {\bibfnamefont {W.~J.}\ \bibnamefont {Munro}},
  \bibinfo {author} {\bibfnamefont {K.}~\bibnamefont {Nemoto}}, \ and\ \bibinfo
  {author} {\bibfnamefont {K.~M.}\ \bibnamefont {Itoh}},\ }\href {\doibase
  10.1145/1324177.1324179} {\bibfield  {journal} {\bibinfo  {journal} {ACM J.
  Emerg. Technol. Comput. Syst.}\ }\textbf {\bibinfo {volume} {3}},\ \bibinfo
  {pages} {1} (\bibinfo {year} {2008})}\BibitemShut {NoStop}%
\bibitem [{\citenamefont {Ahsan}\ and\ \citenamefont {Kim}(2015)}]{Ahsan2015}%
  \BibitemOpen
  \bibfield  {author} {\bibinfo {author} {\bibfnamefont {M.}~\bibnamefont
  {Ahsan}}\ and\ \bibinfo {author} {\bibfnamefont {J.}~\bibnamefont {Kim}},\
  }in\ \href {\doibase 10.7873/DATE.2015.0318} {\emph {\bibinfo {booktitle}
  {2015 {{Design}}, {{Automation Test}} in {{Europe Conference Exhibition}}
  ({{DATE}})}}}\ (\bibinfo  {publisher} {{IEEE Conference Publications}},\
  \bibinfo {year} {2015})\ pp.\ \bibinfo {pages} {1108--1113}\BibitemShut
  {NoStop}%
\bibitem [{\citenamefont {Metodi}\ \emph {et~al.}(2005)\citenamefont {Metodi},
  \citenamefont {Thaker},\ and\ \citenamefont {Cross}}]{Metodi2005}%
  \BibitemOpen
  \bibfield  {author} {\bibinfo {author} {\bibfnamefont {T.}~\bibnamefont
  {Metodi}}, \bibinfo {author} {\bibfnamefont {D.}~\bibnamefont {Thaker}}, \
  and\ \bibinfo {author} {\bibfnamefont {A.}~\bibnamefont {Cross}},\ }in\ \href
  {\doibase 10.1109/MICRO.2005.9} {\emph {\bibinfo {booktitle} {38th {{Annual
  IEEE}}/{{ACM International Symposium}} on {{Microarchitecture}}
  ({{MICRO}}'05)}}}\ (\bibinfo  {publisher} {{IEEE}},\ \bibinfo {year} {2005})\
  pp.\ \bibinfo {pages} {305--318}\BibitemShut {NoStop}%
\bibitem [{\citenamefont {Duan}\ and\ \citenamefont {Monroe}(2010)}]{Duan2010}%
  \BibitemOpen
  \bibfield  {author} {\bibinfo {author} {\bibfnamefont {L.-M.}\ \bibnamefont
  {Duan}}\ and\ \bibinfo {author} {\bibfnamefont {C.}~\bibnamefont {Monroe}},\
  }\href {\doibase 10.1103/RevModPhys.82.1209} {\bibfield  {journal} {\bibinfo
  {journal} {Rev. Mod. Phys.}\ }\textbf {\bibinfo {volume} {82}},\ \bibinfo
  {pages} {1209} (\bibinfo {year} {2010})}\BibitemShut {NoStop}%
\bibitem [{\citenamefont {Monroe}\ and\ \citenamefont
  {Kim}(2013)}]{Monroe2013}%
  \BibitemOpen
  \bibfield  {author} {\bibinfo {author} {\bibfnamefont {C.}~\bibnamefont
  {Monroe}}\ and\ \bibinfo {author} {\bibfnamefont {J.}~\bibnamefont {Kim}},\
  }\href {\doibase 10.1126/science.1231298} {\bibfield  {journal} {\bibinfo
  {journal} {Science}\ }\textbf {\bibinfo {volume} {339}},\ \bibinfo {pages}
  {1164} (\bibinfo {year} {2013})}\BibitemShut {NoStop}%
\bibitem [{\citenamefont {Devoret}\ and\ \citenamefont
  {Schoelkopf}(2013)}]{Devoret2013}%
  \BibitemOpen
  \bibfield  {author} {\bibinfo {author} {\bibfnamefont {M.~H.}\ \bibnamefont
  {Devoret}}\ and\ \bibinfo {author} {\bibfnamefont {R.~J.}\ \bibnamefont
  {Schoelkopf}},\ }\href {\doibase 10.1126/science.1231930} {\bibfield
  {journal} {\bibinfo  {journal} {Science}\ }\textbf {\bibinfo {volume}
  {339}},\ \bibinfo {pages} {1169} (\bibinfo {year} {2013})}\BibitemShut
  {NoStop}%
\bibitem [{\citenamefont {Brecht}\ \emph {et~al.}(2016)\citenamefont {Brecht},
  \citenamefont {Pfaff}, \citenamefont {Wang}, \citenamefont {Chu},
  \citenamefont {Frunzio}, \citenamefont {Devoret},\ and\ \citenamefont
  {Schoelkopf}}]{Brecht2016}%
  \BibitemOpen
  \bibfield  {author} {\bibinfo {author} {\bibfnamefont {T.}~\bibnamefont
  {Brecht}}, \bibinfo {author} {\bibfnamefont {W.}~\bibnamefont {Pfaff}},
  \bibinfo {author} {\bibfnamefont {C.}~\bibnamefont {Wang}}, \bibinfo {author}
  {\bibfnamefont {Y.}~\bibnamefont {Chu}}, \bibinfo {author} {\bibfnamefont
  {L.}~\bibnamefont {Frunzio}}, \bibinfo {author} {\bibfnamefont {M.~H.}\
  \bibnamefont {Devoret}}, \ and\ \bibinfo {author} {\bibfnamefont {R.~J.}\
  \bibnamefont {Schoelkopf}},\ }\href {\doibase 10.1038/npjqi.2016.2}
  {\bibfield  {journal} {\bibinfo  {journal} {Npj Quantum Inf.}\ }\textbf
  {\bibinfo {volume} {2}},\ \bibinfo {pages} {16002} (\bibinfo {year}
  {2016})}\BibitemShut {NoStop}%
\bibitem [{\citenamefont {Kurpiers}\ \emph {et~al.}(2017)\citenamefont
  {Kurpiers}, \citenamefont {Magnard}, \citenamefont {Walter}, \citenamefont
  {Royer}, \citenamefont {Pechal}, \citenamefont {Heinsoo}, \citenamefont
  {Salathé}, \citenamefont {Akin}, \citenamefont {Storz}, \citenamefont
  {Besse}, \citenamefont {Gasparinetti}, \citenamefont {Blais},\ and\
  \citenamefont {Wallraff}}]{Kurpiers2017}%
  \BibitemOpen
  \bibfield  {author} {\bibinfo {author} {\bibfnamefont {P.}~\bibnamefont
  {Kurpiers}}, \bibinfo {author} {\bibfnamefont {P.}~\bibnamefont {Magnard}},
  \bibinfo {author} {\bibfnamefont {T.}~\bibnamefont {Walter}}, \bibinfo
  {author} {\bibfnamefont {B.}~\bibnamefont {Royer}}, \bibinfo {author}
  {\bibfnamefont {M.}~\bibnamefont {Pechal}}, \bibinfo {author} {\bibfnamefont
  {J.}~\bibnamefont {Heinsoo}}, \bibinfo {author} {\bibfnamefont
  {Y.}~\bibnamefont {Salathé}}, \bibinfo {author} {\bibfnamefont
  {A.}~\bibnamefont {Akin}}, \bibinfo {author} {\bibfnamefont {S.}~\bibnamefont
  {Storz}}, \bibinfo {author} {\bibfnamefont {J.-C.}\ \bibnamefont {Besse}},
  \bibinfo {author} {\bibfnamefont {S.}~\bibnamefont {Gasparinetti}}, \bibinfo
  {author} {\bibfnamefont {A.}~\bibnamefont {Blais}}, \ and\ \bibinfo {author}
  {\bibfnamefont {A.}~\bibnamefont {Wallraff}},\ }\href@noop {} {\  (\bibinfo
  {year} {2017})},\ \Eprint {http://arxiv.org/abs/1712.08593}
  {arXiv:1712.08593} \BibitemShut {NoStop}%
\bibitem [{\citenamefont {Hagberg}\ \emph {et~al.}(2008)\citenamefont
  {Hagberg}, \citenamefont {Schult},\ and\ \citenamefont {Swart}}]{NetworkX}%
  \BibitemOpen
  \bibfield  {author} {\bibinfo {author} {\bibfnamefont {A.~A.}\ \bibnamefont
  {Hagberg}}, \bibinfo {author} {\bibfnamefont {D.~A.}\ \bibnamefont {Schult}},
  \ and\ \bibinfo {author} {\bibfnamefont {P.~J.}\ \bibnamefont {Swart}},\ }in\
  \href@noop {} {\emph {\bibinfo {booktitle} {Proceedings of the 7th {{Python}}
  in {{Science Conference}}}}},\ \bibinfo {editor} {edited by\ \bibinfo
  {editor} {\bibfnamefont {G.}~\bibnamefont {Varoquaux}}, \bibinfo {editor}
  {\bibfnamefont {T.}~\bibnamefont {Vaught}}, \ and\ \bibinfo {editor}
  {\bibfnamefont {J.}~\bibnamefont {Millman}}}\ (\bibinfo {address} {Pasadena,
  CA USA},\ \bibinfo {year} {2008})\ pp.\ \bibinfo {pages} {11--15}\BibitemShut
  {NoStop}%
\bibitem [{\citenamefont {Bollinger}\ \emph {et~al.}(1996)\citenamefont
  {Bollinger}, \citenamefont {Itano}, \citenamefont {Wineland},\ and\
  \citenamefont {Heinzen}}]{Bollinger1996}%
  \BibitemOpen
  \bibfield  {author} {\bibinfo {author} {\bibfnamefont {J.~J.}\ \bibnamefont
  {Bollinger}}, \bibinfo {author} {\bibfnamefont {W.~M.}\ \bibnamefont
  {Itano}}, \bibinfo {author} {\bibfnamefont {D.~J.}\ \bibnamefont {Wineland}},
  \ and\ \bibinfo {author} {\bibfnamefont {D.~J.}\ \bibnamefont {Heinzen}},\
  }\href {\doibase 10.1103/PhysRevA.54.R4649} {\bibfield  {journal} {\bibinfo
  {journal} {Phys. Rev. A}\ }\textbf {\bibinfo {volume} {54}},\ \bibinfo
  {pages} {R4649} (\bibinfo {year} {1996})}\BibitemShut {NoStop}%
\bibitem [{\citenamefont {Eldredge}\ \emph {et~al.}(2016)\citenamefont
  {Eldredge}, \citenamefont {Foss-Feig}, \citenamefont {Gross}, \citenamefont
  {Rolston},\ and\ \citenamefont {Gorshkov}}]{Eldredge2016}%
  \BibitemOpen
  \bibfield  {author} {\bibinfo {author} {\bibfnamefont {Z.}~\bibnamefont
  {Eldredge}}, \bibinfo {author} {\bibfnamefont {M.}~\bibnamefont {Foss-Feig}},
  \bibinfo {author} {\bibfnamefont {J.~A.}\ \bibnamefont {Gross}}, \bibinfo
  {author} {\bibfnamefont {S.~L.}\ \bibnamefont {Rolston}}, \ and\ \bibinfo
  {author} {\bibfnamefont {A.~V.}\ \bibnamefont {Gorshkov}},\ }\href@noop {} {\
   (\bibinfo {year} {2016})},\ \Eprint {http://arxiv.org/abs/1607.04646}
  {arXiv:1607.04646} \BibitemShut {NoStop}%
\bibitem [{\citenamefont {Eldredge}\ \emph {et~al.}(2017)\citenamefont
  {Eldredge}, \citenamefont {Gong}, \citenamefont {Young}, \citenamefont
  {Moosavian}, \citenamefont {Foss-Feig},\ and\ \citenamefont
  {Gorshkov}}]{Eldredge2017}%
  \BibitemOpen
  \bibfield  {author} {\bibinfo {author} {\bibfnamefont {Z.}~\bibnamefont
  {Eldredge}}, \bibinfo {author} {\bibfnamefont {Z.-X.}\ \bibnamefont {Gong}},
  \bibinfo {author} {\bibfnamefont {J.~T.}\ \bibnamefont {Young}}, \bibinfo
  {author} {\bibfnamefont {A.~H.}\ \bibnamefont {Moosavian}}, \bibinfo {author}
  {\bibfnamefont {M.}~\bibnamefont {Foss-Feig}}, \ and\ \bibinfo {author}
  {\bibfnamefont {A.~V.}\ \bibnamefont {Gorshkov}},\ }\href {\doibase
  10.1103/PhysRevLett.119.170503} {\bibfield  {journal} {\bibinfo  {journal}
  {Phys. Rev. Lett.}\ }\textbf {\bibinfo {volume} {119}},\ \bibinfo {pages}
  {170503} (\bibinfo {year} {2017})}\BibitemShut {NoStop}%
\bibitem [{\citenamefont {Bravyi}\ \emph {et~al.}(2006)\citenamefont {Bravyi},
  \citenamefont {Hastings},\ and\ \citenamefont {Verstraete}}]{Bravyi2006a}%
  \BibitemOpen
  \bibfield  {author} {\bibinfo {author} {\bibfnamefont {S.}~\bibnamefont
  {Bravyi}}, \bibinfo {author} {\bibfnamefont {M.~B.}\ \bibnamefont
  {Hastings}}, \ and\ \bibinfo {author} {\bibfnamefont {F.}~\bibnamefont
  {Verstraete}},\ }\href {\doibase 10.1103/PhysRevLett.97.050401} {\bibfield
  {journal} {\bibinfo  {journal} {Phys. Rev. Lett.}\ }\textbf {\bibinfo
  {volume} {97}},\ \bibinfo {pages} {050401} (\bibinfo {year}
  {2006})}\BibitemShut {NoStop}%
\bibitem [{\citenamefont {Bentsen}\ \emph {et~al.}(2018)\citenamefont
  {Bentsen}, \citenamefont {Gu},\ and\ \citenamefont {Lucas}}]{Bentsen2018}%
  \BibitemOpen
  \bibfield  {author} {\bibinfo {author} {\bibfnamefont {G.}~\bibnamefont
  {Bentsen}}, \bibinfo {author} {\bibfnamefont {Y.}~\bibnamefont {Gu}}, \ and\
  \bibinfo {author} {\bibfnamefont {A.}~\bibnamefont {Lucas}},\ }\href@noop {}
  {\  (\bibinfo {year} {2018})},\ \Eprint {http://arxiv.org/abs/1805.08215}
  {arXiv:1805.08215} \BibitemShut {NoStop}%
\bibitem [{\citenamefont {Cuquet}\ and\ \citenamefont
  {Calsamiglia}(2012)}]{Cuquet2012}%
  \BibitemOpen
  \bibfield  {author} {\bibinfo {author} {\bibfnamefont {M.}~\bibnamefont
  {Cuquet}}\ and\ \bibinfo {author} {\bibfnamefont {J.}~\bibnamefont
  {Calsamiglia}},\ }\href {\doibase 10.1103/PhysRevA.86.042304} {\bibfield
  {journal} {\bibinfo  {journal} {Phys. Rev. A}\ }\textbf {\bibinfo {volume}
  {86}},\ \bibinfo {pages} {042304} (\bibinfo {year} {2012})}\BibitemShut
  {NoStop}%
\bibitem [{\citenamefont {A{\'c}in}\ \emph {et~al.}(2007)\citenamefont
  {A{\'c}in}, \citenamefont {Cirac},\ and\ \citenamefont
  {Lewenstein}}]{Acin2007}%
  \BibitemOpen
  \bibfield  {author} {\bibinfo {author} {\bibfnamefont {A.}~\bibnamefont
  {A{\'c}in}}, \bibinfo {author} {\bibfnamefont {J.~I.}\ \bibnamefont {Cirac}},
  \ and\ \bibinfo {author} {\bibfnamefont {M.}~\bibnamefont {Lewenstein}},\
  }\href {\doibase 10.1038/nphys549} {\bibfield  {journal} {\bibinfo  {journal}
  {Nat. Phys.}\ }\textbf {\bibinfo {volume} {3}},\ \bibinfo {pages} {256}
  (\bibinfo {year} {2007})}\BibitemShut {NoStop}%
\bibitem [{\citenamefont {Perseguers}\ \emph {et~al.}(2013)\citenamefont
  {Perseguers}, \citenamefont {Lapeyre}, \citenamefont {Cavalcanti},
  \citenamefont {Lewenstein},\ and\ \citenamefont
  {Ac{\'\i}n}}]{Perseguers2013}%
  \BibitemOpen
  \bibfield  {author} {\bibinfo {author} {\bibfnamefont {S.}~\bibnamefont
  {Perseguers}}, \bibinfo {author} {\bibfnamefont {G.~J.}\ \bibnamefont
  {Lapeyre}}, \bibinfo {author} {\bibfnamefont {D.}~\bibnamefont {Cavalcanti}},
  \bibinfo {author} {\bibfnamefont {M.}~\bibnamefont {Lewenstein}}, \ and\
  \bibinfo {author} {\bibfnamefont {A.}~\bibnamefont {Ac{\'\i}n}},\ }\href
  {\doibase 10.1088/0034-4885/76/9/096001} {\bibfield  {journal} {\bibinfo
  {journal} {Rep. Prog. Phys.}\ }\textbf {\bibinfo {volume} {76}},\ \bibinfo
  {pages} {096001} (\bibinfo {year} {2013})}\BibitemShut {NoStop}%
\bibitem [{\citenamefont {Kieling}\ \emph {et~al.}(2007)\citenamefont
  {Kieling}, \citenamefont {Rudolph},\ and\ \citenamefont
  {Eisert}}]{Kieling2007}%
  \BibitemOpen
  \bibfield  {author} {\bibinfo {author} {\bibfnamefont {K.}~\bibnamefont
  {Kieling}}, \bibinfo {author} {\bibfnamefont {T.}~\bibnamefont {Rudolph}}, \
  and\ \bibinfo {author} {\bibfnamefont {J.}~\bibnamefont {Eisert}},\ }\href
  {\doibase 10.1103/PhysRevLett.99.130501} {\bibfield  {journal} {\bibinfo
  {journal} {Phys. Rev. Lett.}\ }\textbf {\bibinfo {volume} {99}},\ \bibinfo
  {pages} {130501} (\bibinfo {year} {2007})}\BibitemShut {NoStop}%
\bibitem [{\citenamefont {Cuquet}\ and\ \citenamefont
  {Calsamiglia}(2009)}]{Cuquet2009}%
  \BibitemOpen
  \bibfield  {author} {\bibinfo {author} {\bibfnamefont {M.}~\bibnamefont
  {Cuquet}}\ and\ \bibinfo {author} {\bibfnamefont {J.}~\bibnamefont
  {Calsamiglia}},\ }\href {\doibase 10.1103/PhysRevLett.103.240503} {\bibfield
  {journal} {\bibinfo  {journal} {Phys. Rev. Lett.}\ }\textbf {\bibinfo
  {volume} {103}},\ \bibinfo {pages} {240503} (\bibinfo {year}
  {2009})}\BibitemShut {NoStop}%
\bibitem [{\citenamefont {Cuquet}\ and\ \citenamefont
  {Calsamiglia}(2011)}]{Cuquet2011}%
  \BibitemOpen
  \bibfield  {author} {\bibinfo {author} {\bibfnamefont {M.}~\bibnamefont
  {Cuquet}}\ and\ \bibinfo {author} {\bibfnamefont {J.}~\bibnamefont
  {Calsamiglia}},\ }\href {\doibase 10.1103/PhysRevA.83.032319} {\bibfield
  {journal} {\bibinfo  {journal} {Phys. Rev. A}\ }\textbf {\bibinfo {volume}
  {83}},\ \bibinfo {pages} {032319} (\bibinfo {year} {2011})}\BibitemShut
  {NoStop}%
\bibitem [{\citenamefont {{Yang Wang}}\ \emph {et~al.}(2003)\citenamefont
  {{Yang Wang}}, \citenamefont {Chakrabarti}, \citenamefont {{Chenxi Wang}},\
  and\ \citenamefont {Faloutsos}}]{YangWang2003}%
  \BibitemOpen
  \bibfield  {author} {\bibinfo {author} {\bibnamefont {{Yang Wang}}}, \bibinfo
  {author} {\bibfnamefont {D.}~\bibnamefont {Chakrabarti}}, \bibinfo {author}
  {\bibnamefont {{Chenxi Wang}}}, \ and\ \bibinfo {author} {\bibfnamefont
  {C.}~\bibnamefont {Faloutsos}},\ }in\ \href {\doibase
  10.1109/RELDIS.2003.1238052} {\emph {\bibinfo {booktitle} {22nd
  {{International Symposium}} on {{Reliable Distributed Systems}}, 2003.
  {{Proceedings}}.}}}\ (\bibinfo  {publisher} {{IEEE Comput. Soc}},\ \bibinfo
  {year} {2003})\ pp.\ \bibinfo {pages} {25--34}\BibitemShut {NoStop}%
\bibitem [{\citenamefont {Watts}\ and\ \citenamefont
  {Strogatz}(1998)}]{Watts1998}%
  \BibitemOpen
  \bibfield  {author} {\bibinfo {author} {\bibfnamefont {D.~J.}\ \bibnamefont
  {Watts}}\ and\ \bibinfo {author} {\bibfnamefont {S.~H.}\ \bibnamefont
  {Strogatz}},\ }\href {\doibase 10.1038/30918} {\bibfield  {journal} {\bibinfo
   {journal} {Nature}\ }\textbf {\bibinfo {volume} {393}},\ \bibinfo {pages}
  {440} (\bibinfo {year} {1998})}\BibitemShut {NoStop}%
\bibitem [{\citenamefont {Barab{\'a}si}\ and\ \citenamefont
  {Albert}(1999)}]{Barabasi1999}%
  \BibitemOpen
  \bibfield  {author} {\bibinfo {author} {\bibfnamefont {A.-L.}\ \bibnamefont
  {Barab{\'a}si}}\ and\ \bibinfo {author} {\bibfnamefont {R.}~\bibnamefont
  {Albert}},\ }\href {\doibase 10.1126/science.286.5439.509} {\bibfield
  {journal} {\bibinfo  {journal} {Science}\ }\textbf {\bibinfo {volume}
  {286}},\ \bibinfo {pages} {509} (\bibinfo {year} {1999})}\BibitemShut
  {NoStop}%
\bibitem [{\citenamefont {Albert}\ and\ \citenamefont
  {Barab{\'a}si}(2002)}]{Albert2002}%
  \BibitemOpen
  \bibfield  {author} {\bibinfo {author} {\bibfnamefont {R.}~\bibnamefont
  {Albert}}\ and\ \bibinfo {author} {\bibfnamefont {A.-L.}\ \bibnamefont
  {Barab{\'a}si}},\ }\href {\doibase 10.1103/RevModPhys.74.47} {\bibfield
  {journal} {\bibinfo  {journal} {Rev. Mod. Phys.}\ }\textbf {\bibinfo {volume}
  {74}},\ \bibinfo {pages} {47} (\bibinfo {year} {2002})}\BibitemShut {NoStop}%
\bibitem [{\citenamefont {Perseguers}\ \emph {et~al.}(2010)\citenamefont
  {Perseguers}, \citenamefont {Lewenstein}, \citenamefont {A{\'c}in},\ and\
  \citenamefont {Cirac}}]{Perseguers2010}%
  \BibitemOpen
  \bibfield  {author} {\bibinfo {author} {\bibfnamefont {S.}~\bibnamefont
  {Perseguers}}, \bibinfo {author} {\bibfnamefont {M.}~\bibnamefont
  {Lewenstein}}, \bibinfo {author} {\bibfnamefont {A.}~\bibnamefont
  {A{\'c}in}}, \ and\ \bibinfo {author} {\bibfnamefont {J.~I.}\ \bibnamefont
  {Cirac}},\ }\href {\doibase 10.1038/nphys1665} {\bibfield  {journal}
  {\bibinfo  {journal} {Nat. Phys.}\ }\textbf {\bibinfo {volume} {6}},\
  \bibinfo {pages} {539} (\bibinfo {year} {2010})}\BibitemShut {NoStop}%
\bibitem [{\citenamefont {Di~Franco}\ and\ \citenamefont
  {Ballester}(2012)}]{DiFranco2012}%
  \BibitemOpen
  \bibfield  {author} {\bibinfo {author} {\bibfnamefont {C.}~\bibnamefont
  {Di~Franco}}\ and\ \bibinfo {author} {\bibfnamefont {D.}~\bibnamefont
  {Ballester}},\ }\href {\doibase 10.1103/PhysRevA.85.010303} {\bibfield
  {journal} {\bibinfo  {journal} {Phys. Rev. A}\ }\textbf {\bibinfo {volume}
  {85}},\ \bibinfo {pages} {010303} (\bibinfo {year} {2012})}\BibitemShut
  {NoStop}%
\bibitem [{\citenamefont {Monroe}\ \emph {et~al.}(2014)\citenamefont {Monroe},
  \citenamefont {Raussendorf}, \citenamefont {Ruthven}, \citenamefont {Brown},
  \citenamefont {Maunz}, \citenamefont {Duan},\ and\ \citenamefont
  {Kim}}]{Monroe2014}%
  \BibitemOpen
  \bibfield  {author} {\bibinfo {author} {\bibfnamefont {C.}~\bibnamefont
  {Monroe}}, \bibinfo {author} {\bibfnamefont {R.}~\bibnamefont {Raussendorf}},
  \bibinfo {author} {\bibfnamefont {A.}~\bibnamefont {Ruthven}}, \bibinfo
  {author} {\bibfnamefont {K.~R.}\ \bibnamefont {Brown}}, \bibinfo {author}
  {\bibfnamefont {P.}~\bibnamefont {Maunz}}, \bibinfo {author} {\bibfnamefont
  {L.~M.}\ \bibnamefont {Duan}}, \ and\ \bibinfo {author} {\bibfnamefont
  {J.}~\bibnamefont {Kim}},\ }\href {\doibase 10.1103/PhysRevA.89.022317}
  {\bibfield  {journal} {\bibinfo  {journal} {Phys. Rev. A}\ }\textbf {\bibinfo
  {volume} {89}},\ \bibinfo {pages} {1} (\bibinfo {year} {2014})}\BibitemShut
  {NoStop}%
\bibitem [{\citenamefont {Briegel}\ \emph {et~al.}(1998)\citenamefont
  {Briegel}, \citenamefont {D{\"u}r}, \citenamefont {Cirac},\ and\
  \citenamefont {Zoller}}]{Briegel1998}%
  \BibitemOpen
  \bibfield  {author} {\bibinfo {author} {\bibfnamefont {H.-J.}\ \bibnamefont
  {Briegel}}, \bibinfo {author} {\bibfnamefont {W.}~\bibnamefont {D{\"u}r}},
  \bibinfo {author} {\bibfnamefont {J.~I.}\ \bibnamefont {Cirac}}, \ and\
  \bibinfo {author} {\bibfnamefont {P.}~\bibnamefont {Zoller}},\ }\href
  {\doibase 10.1103/PhysRevLett.81.5932} {\bibfield  {journal} {\bibinfo
  {journal} {Phys. Rev. Lett.}\ }\textbf {\bibinfo {volume} {81}},\ \bibinfo
  {pages} {5932} (\bibinfo {year} {1998})}\BibitemShut {NoStop}%
\bibitem [{\citenamefont {Barri{\`e}re}\ \emph {et~al.}(2009)\citenamefont
  {Barri{\`e}re}, \citenamefont {Comellas}, \citenamefont {Dalf{\'o}},\ and\
  \citenamefont {Fiol}}]{Barriere2009}%
  \BibitemOpen
  \bibfield  {author} {\bibinfo {author} {\bibfnamefont {L.}~\bibnamefont
  {Barri{\`e}re}}, \bibinfo {author} {\bibfnamefont {F.}~\bibnamefont
  {Comellas}}, \bibinfo {author} {\bibfnamefont {C.}~\bibnamefont {Dalf{\'o}}},
  \ and\ \bibinfo {author} {\bibfnamefont {M.}~\bibnamefont {Fiol}},\ }\href
  {\doibase 10.1016/j.dam.2008.04.018} {\bibfield  {journal} {\bibinfo
  {journal} {Discrete Appl. Math.}\ }\textbf {\bibinfo {volume} {157}},\
  \bibinfo {pages} {36} (\bibinfo {year} {2009})}\BibitemShut {NoStop}%
\bibitem [{\citenamefont {Godsil}\ and\ \citenamefont
  {McKay}(1978)}]{Godsil1978}%
  \BibitemOpen
  \bibfield  {author} {\bibinfo {author} {\bibfnamefont {C.}~\bibnamefont
  {Godsil}}\ and\ \bibinfo {author} {\bibfnamefont {B.}~\bibnamefont {McKay}},\
  }\href {\doibase 10.1017/S0004972700007760} {\bibfield  {journal} {\bibinfo
  {journal} {Bulletin of the Australian Mathematical Society}\ }\textbf
  {\bibinfo {volume} {18}},\ \bibinfo {pages} {21–28} (\bibinfo {year}
  {1978})}\BibitemShut {NoStop}%
\bibitem [{\citenamefont {Harary}(1955)}]{Harary1955}%
  \BibitemOpen
  \bibfield  {author} {\bibinfo {author} {\bibfnamefont {F.}~\bibnamefont
  {Harary}},\ }\href {\doibase 10.2307/1993073} {\bibfield  {journal} {\bibinfo
   {journal} {Trans. Amer. Math. Soc.}\ }\textbf {\bibinfo {volume} {78}},\
  \bibinfo {pages} {445} (\bibinfo {year} {1955})}\BibitemShut {NoStop}%
\bibitem [{\citenamefont {Gromov}(2007)}]{Gromov2007}%
  \BibitemOpen
  \bibfield  {author} {\bibinfo {author} {\bibfnamefont {M.}~\bibnamefont
  {Gromov}},\ }\href {http://www.springer.com/in/book/9780817645823} {\emph
  {\bibinfo {title} {Metric {{Structures}} for {{Riemannian}} and
  {{Non}}-{{Riemannian Spaces}}}}},\ edited by\ \bibinfo {editor}
  {\bibfnamefont {J.}~\bibnamefont {LaFontaine}}\ and\ \bibinfo {editor}
  {\bibfnamefont {P.}~\bibnamefont {Pansu}},\ Modern Birkh{\"a}user Classics\
  (\bibinfo  {publisher} {{Birkh{\"a}user Basel}},\ \bibinfo {year}
  {2007})\BibitemShut {NoStop}%
\bibitem [{\citenamefont {Chen}\ \emph {et~al.}(2013)\citenamefont {Chen},
  \citenamefont {Fang}, \citenamefont {Hu},\ and\ \citenamefont
  {Mahoney}}]{Chen2013}%
  \BibitemOpen
  \bibfield  {author} {\bibinfo {author} {\bibfnamefont {W.}~\bibnamefont
  {Chen}}, \bibinfo {author} {\bibfnamefont {W.}~\bibnamefont {Fang}}, \bibinfo
  {author} {\bibfnamefont {G.}~\bibnamefont {Hu}}, \ and\ \bibinfo {author}
  {\bibfnamefont {M.~W.}\ \bibnamefont {Mahoney}},\ }\href {\doibase
  10.1080/15427951.2013.828336} {\bibfield  {journal} {\bibinfo  {journal}
  {Internet Math.}\ }\textbf {\bibinfo {volume} {9}},\ \bibinfo {pages} {434}
  (\bibinfo {year} {2013})}\BibitemShut {NoStop}%
\bibitem [{\citenamefont {Lohsoonthorn}(2003)}]{Lohsoonthorn2003}%
  \BibitemOpen
  \bibfield  {author} {\bibinfo {author} {\bibfnamefont {P.}~\bibnamefont
  {Lohsoonthorn}},\ }\emph {\bibinfo {title} {Hyperbolic {{Geometry}} of
  {{Networks}}}},\ \href
  {http://digitallibrary.usc.edu/cdm/ref/collection/p15799coll16/id/642852}
  {\bibinfo {type} {{{PhD Thesis}}}},\ \bibinfo  {school} {University of
  Southern California}, \bibinfo {address} {Los Angeles, CA, USA} (\bibinfo
  {year} {2003})\BibitemShut {NoStop}%
\bibitem [{\citenamefont {Krioukov}\ \emph {et~al.}(2010)\citenamefont
  {Krioukov}, \citenamefont {Papadopoulos}, \citenamefont {Kitsak},
  \citenamefont {Vahdat},\ and\ \citenamefont {Bogu{\~n}{\'a}}}]{Krioukov2010}%
  \BibitemOpen
  \bibfield  {author} {\bibinfo {author} {\bibfnamefont {D.}~\bibnamefont
  {Krioukov}}, \bibinfo {author} {\bibfnamefont {F.}~\bibnamefont
  {Papadopoulos}}, \bibinfo {author} {\bibfnamefont {M.}~\bibnamefont
  {Kitsak}}, \bibinfo {author} {\bibfnamefont {A.}~\bibnamefont {Vahdat}}, \
  and\ \bibinfo {author} {\bibfnamefont {M.}~\bibnamefont {Bogu{\~n}{\'a}}},\
  }\href {\doibase 10.1103/PhysRevE.82.036106} {\bibfield  {journal} {\bibinfo
  {journal} {Phys. Rev. E}\ }\textbf {\bibinfo {volume} {82}},\ \bibinfo
  {pages} {036106} (\bibinfo {year} {2010})}\BibitemShut {NoStop}%
\bibitem [{\citenamefont {{de Montgolfier}}\ \emph {et~al.}(2011)\citenamefont
  {{de Montgolfier}}, \citenamefont {Soto},\ and\ \citenamefont
  {Viennot}}]{Montgolfier2011}%
  \BibitemOpen
  \bibfield  {author} {\bibinfo {author} {\bibfnamefont {F.}~\bibnamefont {{de
  Montgolfier}}}, \bibinfo {author} {\bibfnamefont {M.}~\bibnamefont {Soto}}, \
  and\ \bibinfo {author} {\bibfnamefont {L.}~\bibnamefont {Viennot}},\ }in\
  \href {\doibase 10.1109/NCA.2011.11} {\emph {\bibinfo {booktitle}
  {Proceedings of the 2011 {{IEEE}} 10th {{International Symposium}} on
  {{Network Computing}} and {{Applications}}}}},\ \bibinfo {series and number}
  {NCA '11}\ (\bibinfo  {publisher} {{IEEE Computer Society}},\ \bibinfo
  {address} {Washington, DC, USA},\ \bibinfo {year} {2011})\ pp.\ \bibinfo
  {pages} {25--32}\BibitemShut {NoStop}%
\bibitem [{\citenamefont {Bogu{\~n}{\'a}}\ \emph {et~al.}(2010)\citenamefont
  {Bogu{\~n}{\'a}}, \citenamefont {Papadopoulos},\ and\ \citenamefont
  {Krioukov}}]{Boguna2010}%
  \BibitemOpen
  \bibfield  {author} {\bibinfo {author} {\bibfnamefont {M.}~\bibnamefont
  {Bogu{\~n}{\'a}}}, \bibinfo {author} {\bibfnamefont {F.}~\bibnamefont
  {Papadopoulos}}, \ and\ \bibinfo {author} {\bibfnamefont {D.}~\bibnamefont
  {Krioukov}},\ }\href {\doibase 10.1038/ncomms1063} {\bibfield  {journal}
  {\bibinfo  {journal} {Nat. Commun.}\ }\textbf {\bibinfo {volume} {1}},\
  \bibinfo {pages} {62} (\bibinfo {year} {2010})}\BibitemShut {NoStop}%
\bibitem [{\citenamefont {Koll{\'a}r}\ \emph {et~al.}(2018)\citenamefont
  {Koll{\'a}r}, \citenamefont {Fitzpatrick},\ and\ \citenamefont
  {Houck}}]{Kollar2018}%
  \BibitemOpen
  \bibfield  {author} {\bibinfo {author} {\bibfnamefont {A.~J.}\ \bibnamefont
  {Koll{\'a}r}}, \bibinfo {author} {\bibfnamefont {M.}~\bibnamefont
  {Fitzpatrick}}, \ and\ \bibinfo {author} {\bibfnamefont {A.~A.}\ \bibnamefont
  {Houck}},\ }\href@noop {} {\  (\bibinfo {year} {2018})},\ \Eprint
  {http://arxiv.org/abs/1802.09549} {arXiv:1802.09549} \BibitemShut {NoStop}%
\bibitem [{\citenamefont {Oum}(2005)}]{OUM200579}%
  \BibitemOpen
  \bibfield  {author} {\bibinfo {author} {\bibfnamefont {S.-i.}\ \bibnamefont
  {Oum}},\ }\href {\doibase https://doi.org/10.1016/j.jctb.2005.03.003}
  {\bibfield  {journal} {\bibinfo  {journal} {J. Comb. Theory Ser. B}\ }\textbf
  {\bibinfo {volume} {95}},\ \bibinfo {pages} {79} (\bibinfo {year}
  {2005})}\BibitemShut {NoStop}%
\bibitem [{\citenamefont {Courcelle}\ \emph {et~al.}(1998)\citenamefont
  {Courcelle}, \citenamefont {Makowsky},\ and\ \citenamefont
  {Rotics}}]{Courcelle1998}%
  \BibitemOpen
  \bibfield  {author} {\bibinfo {author} {\bibfnamefont {B.}~\bibnamefont
  {Courcelle}}, \bibinfo {author} {\bibfnamefont {J.~A.}\ \bibnamefont
  {Makowsky}}, \ and\ \bibinfo {author} {\bibfnamefont {U.}~\bibnamefont
  {Rotics}}\ }(\bibinfo  {publisher} {{Springer, Berlin, Heidelberg}},\
  \bibinfo {year} {1998})\ pp.\ \bibinfo {pages} {1--16}\BibitemShut {NoStop}%
\bibitem [{\citenamefont {Leiserson}(1985)}]{Leiserson1985}%
  \BibitemOpen
  \bibfield  {author} {\bibinfo {author} {\bibfnamefont {C.~E.}\ \bibnamefont
  {Leiserson}},\ }\href {\doibase 10.1109/TC.1985.6312192} {\bibfield
  {journal} {\bibinfo  {journal} {IEEE Trans. Comput.}\ }\textbf {\bibinfo
  {volume} {C-34}},\ \bibinfo {pages} {892} (\bibinfo {year}
  {1985})}\BibitemShut {NoStop}%
\bibitem [{\citenamefont {Newman}(2000)}]{Newman2000}%
  \BibitemOpen
  \bibfield  {author} {\bibinfo {author} {\bibfnamefont {M.~W.}\ \bibnamefont
  {Newman}},\ }\emph {\bibinfo {title} {The {{Laplacian Spectrum}} of
  {{Graphs}}}},\ \href
  {https://www.seas.upenn.edu/~jadbabai/ESE680/Laplacian_Thesis.pdf} {\bibinfo
  {type} {Master's {{Thesis}}}},\ \bibinfo  {school} {University of Manitoba},
  \bibinfo {address} {Winnipeg, Canada} (\bibinfo {year} {2000})\BibitemShut
  {NoStop}%
\bibitem [{\citenamefont {Mohar}(1991)}]{Mohar1991}%
  \BibitemOpen
  \bibfield  {author} {\bibinfo {author} {\bibfnamefont {B.}~\bibnamefont
  {Mohar}},\ }in\ \href@noop {} {\emph {\bibinfo {booktitle} {Graph {{Theory}},
  {{Combinatorics}}, and {{Applications}}}}}\ (\bibinfo  {publisher}
  {{Wiley}},\ \bibinfo {year} {1991})\ pp.\ \bibinfo {pages}
  {871--898}\BibitemShut {NoStop}%
\bibitem [{\citenamefont {Chung}(1997)}]{Chung1997}%
  \BibitemOpen
  \bibfield  {author} {\bibinfo {author} {\bibfnamefont {F.~R.~K.}\
  \bibnamefont {Chung}},\ }\href
  {https://books.google.com/books/about/Spectral_Graph_Theory.html?id=zEhCAQAACAAJ}
  {\emph {\bibinfo {title} {Spectral {{Graph Theory}}}}}\ (\bibinfo
  {publisher} {{American Mathematical Society}},\ \bibinfo {year} {1997})\
  Chap.~\bibinfo {chapter} {2}\BibitemShut {NoStop}%
\bibitem [{\citenamefont {Cheung}\ \emph {et~al.}(2007)\citenamefont {Cheung},
  \citenamefont {Maslov},\ and\ \citenamefont {Severini}}]{Cheung2007}%
  \BibitemOpen
  \bibfield  {author} {\bibinfo {author} {\bibfnamefont {D.}~\bibnamefont
  {Cheung}}, \bibinfo {author} {\bibfnamefont {D.}~\bibnamefont {Maslov}}, \
  and\ \bibinfo {author} {\bibfnamefont {S.}~\bibnamefont {Severini}},\ }in\
  \href {http://citeseerx.ist.psu.edu/viewdoc/summary?doi=10.1.1.133.7479}
  {\emph {\bibinfo {booktitle} {Workshop on {{Quantum Information
  Processing}}}}}\ (\bibinfo {year} {2007})\BibitemShut {NoStop}%
\bibitem [{\citenamefont {Linke}\ \emph {et~al.}(2017)\citenamefont {Linke},
  \citenamefont {Maslov}, \citenamefont {Roetteler}, \citenamefont {Debnath},
  \citenamefont {Figgatt}, \citenamefont {Landsman}, \citenamefont {Wright},\
  and\ \citenamefont {Monroe}}]{Linke2017}%
  \BibitemOpen
  \bibfield  {author} {\bibinfo {author} {\bibfnamefont {N.~M.}\ \bibnamefont
  {Linke}}, \bibinfo {author} {\bibfnamefont {D.}~\bibnamefont {Maslov}},
  \bibinfo {author} {\bibfnamefont {M.}~\bibnamefont {Roetteler}}, \bibinfo
  {author} {\bibfnamefont {S.}~\bibnamefont {Debnath}}, \bibinfo {author}
  {\bibfnamefont {C.}~\bibnamefont {Figgatt}}, \bibinfo {author} {\bibfnamefont
  {K.~A.}\ \bibnamefont {Landsman}}, \bibinfo {author} {\bibfnamefont
  {K.}~\bibnamefont {Wright}}, \ and\ \bibinfo {author} {\bibfnamefont
  {C.}~\bibnamefont {Monroe}},\ }\href {\doibase 10.1073/pnas.1618020114}
  {\bibfield  {journal} {\bibinfo  {journal} {Proc. Natl. Acad. Sci.}\ }\textbf
  {\bibinfo {volume} {114}},\ \bibinfo {pages} {3305} (\bibinfo {year}
  {2017})}\BibitemShut {NoStop}%
\bibitem [{Note1()}]{Note1}%
  \BibitemOpen
  \bibinfo {note} {If, for the application at hand, a planar graph is required,
  cycles such as $C_n$ can yield the same scaling.}\BibitemShut {Stop}%
\bibitem [{\citenamefont {Pindyck}\ and\ \citenamefont
  {Rubinfeld}(2013)}]{PindyckRubinfeld}%
  \BibitemOpen
  \bibfield  {author} {\bibinfo {author} {\bibfnamefont {R.~S.}\ \bibnamefont
  {Pindyck}}\ and\ \bibinfo {author} {\bibfnamefont {D.~L.}\ \bibnamefont
  {Rubinfeld}},\ }\href@noop {} {\emph {\bibinfo {title} {Microeconomics}}}\
  (\bibinfo  {publisher} {{Pearson}},\ \bibinfo {year} {2013})\BibitemShut
  {NoStop}%
\bibitem [{\citenamefont {F{\"u}redi}(1990)}]{Furedi1990}%
  \BibitemOpen
  \bibfield  {author} {\bibinfo {author} {\bibfnamefont {Z.}~\bibnamefont
  {F{\"u}redi}},\ }\href {\doibase 10.1007/BF01787701} {\bibfield  {journal}
  {\bibinfo  {journal} {Graphs Comb.}\ }\textbf {\bibinfo {volume} {6}},\
  \bibinfo {pages} {333} (\bibinfo {year} {1990})}\BibitemShut {NoStop}%
\bibitem [{\citenamefont {Hoffman}\ and\ \citenamefont
  {Singleton}(1960)}]{Hoffman1960}%
  \BibitemOpen
  \bibfield  {author} {\bibinfo {author} {\bibfnamefont {A.~J.}\ \bibnamefont
  {Hoffman}}\ and\ \bibinfo {author} {\bibfnamefont {R.~R.}\ \bibnamefont
  {Singleton}},\ }\href {\doibase 10.1147/rd.45.0497} {\bibfield  {journal}
  {\bibinfo  {journal} {IBM J. Res. Dev.}\ }\textbf {\bibinfo {volume} {4}},\
  \bibinfo {pages} {497} (\bibinfo {year} {1960})}\BibitemShut {NoStop}%
\bibitem [{\citenamefont {Miller}\ and\ \citenamefont
  {Sir{\'a}n}(2005)}]{Miller2005}%
  \BibitemOpen
  \bibfield  {author} {\bibinfo {author} {\bibfnamefont {M.}~\bibnamefont
  {Miller}}\ and\ \bibinfo {author} {\bibfnamefont {J.}~\bibnamefont
  {Sir{\'a}n}},\ }\href
  {http://www.combinatorics.org/ojs/index.php/eljc/article/view/DS14/0}
  {\bibfield  {journal} {\bibinfo  {journal} {Electron. J. Comb.}\ }\textbf
  {\bibinfo {volume} {1000}},\ \bibinfo {pages} {DS14} (\bibinfo {year}
  {2005})}\BibitemShut {NoStop}%
\bibitem [{\citenamefont {Pineda-Villavicencio}\ and\ \citenamefont
  {Wood}(2015)}]{Pineda-Villavicencio2013}%
  \BibitemOpen
  \bibfield  {author} {\bibinfo {author} {\bibfnamefont {G.}~\bibnamefont
  {Pineda-Villavicencio}}\ and\ \bibinfo {author} {\bibfnamefont {D.~R.}\
  \bibnamefont {Wood}},\ }\href
  {http://www.combinatorics.org/ojs/index.php/eljc/article/view/v22i2p46}
  {\bibfield  {journal} {\bibinfo  {journal} {Electron. J. Comb.}\ }\textbf
  {\bibinfo {volume} {22}},\ \bibinfo {pages} {2.46} (\bibinfo {year}
  {2015})}\BibitemShut {NoStop}%
\bibitem [{\citenamefont {{Eldredge\ et al.}}()}]{ourgithub}%
  \BibitemOpen
  \bibfield  {author} {\bibinfo {author} {\bibfnamefont {Z.}~\bibnamefont
  {{Eldredge\ et al.}}},\ }\href@noop {} {\enquote {\bibinfo {title}
  {unitary-modular},}\ }\bibinfo {howpublished}
  {\url{https://github.com/zeldredge/unitary-modular}},\ \bibinfo {note}
  {{GitHub repository}}\BibitemShut {NoStop}%
\bibitem [{\citenamefont {Maslov}\ \emph {et~al.}(2008)\citenamefont {Maslov},
  \citenamefont {Falconer},\ and\ \citenamefont {Mosca}}]{Maslov2008}%
  \BibitemOpen
  \bibfield  {author} {\bibinfo {author} {\bibfnamefont {D.}~\bibnamefont
  {Maslov}}, \bibinfo {author} {\bibfnamefont {S.~M.}\ \bibnamefont
  {Falconer}}, \ and\ \bibinfo {author} {\bibfnamefont {M.}~\bibnamefont
  {Mosca}},\ }\href {\doibase 10.1109/TCAD.2008.917562} {\bibfield  {journal}
  {\bibinfo  {journal} {IEEE Trans. Comput.-Aided Des. Integr. Circuits Syst.}\
  }\textbf {\bibinfo {volume} {27}},\ \bibinfo {pages} {752} (\bibinfo {year}
  {2008})}\BibitemShut {NoStop}%
\bibitem [{Note2()}]{Note2}%
  \BibitemOpen
  \bibinfo {note} {We studiously avoid referring to the machine qubits as
  ``physical'' in this paper, as we do not want to confuse this conceptual
  distinction with the physical/logical qubit divide in error correction. All
  of the qubits referenced in this section are logical qubits in the
  error-correcting sense.}\BibitemShut {Stop}%
\bibitem [{\citenamefont {Pedram}\ and\ \citenamefont
  {Shafaei}(2016)}]{Pedram2016}%
  \BibitemOpen
  \bibfield  {author} {\bibinfo {author} {\bibfnamefont {M.}~\bibnamefont
  {Pedram}}\ and\ \bibinfo {author} {\bibfnamefont {A.}~\bibnamefont
  {Shafaei}},\ }\href {\doibase 10.1109/MCAS.2016.2549950} {\bibfield
  {journal} {\bibinfo  {journal} {IEEE Circuits Syst. Mag.}\ }\textbf {\bibinfo
  {volume} {16}},\ \bibinfo {pages} {62} (\bibinfo {year} {2016})}\BibitemShut
  {NoStop}%
\bibitem [{\citenamefont {H{\"a}ner}\ \emph {et~al.}(2016)\citenamefont
  {H{\"a}ner}, \citenamefont {Steiger}, \citenamefont {Svore},\ and\
  \citenamefont {Troyer}}]{Haner2016}%
  \BibitemOpen
  \bibfield  {author} {\bibinfo {author} {\bibfnamefont {T.}~\bibnamefont
  {H{\"a}ner}}, \bibinfo {author} {\bibfnamefont {D.~S.}\ \bibnamefont
  {Steiger}}, \bibinfo {author} {\bibfnamefont {K.}~\bibnamefont {Svore}}, \
  and\ \bibinfo {author} {\bibfnamefont {M.}~\bibnamefont {Troyer}},\
  }\href@noop {} {\  (\bibinfo {year} {2016})},\ \Eprint
  {http://arxiv.org/abs/1604.01401} {arXiv:1604.01401} \BibitemShut {NoStop}%
\bibitem [{\citenamefont {Steiger}\ \emph {et~al.}(2018)\citenamefont
  {Steiger}, \citenamefont {H{\"a}ner},\ and\ \citenamefont
  {Troyer}}]{Steiger2018}%
  \BibitemOpen
  \bibfield  {author} {\bibinfo {author} {\bibfnamefont {D.~S.}\ \bibnamefont
  {Steiger}}, \bibinfo {author} {\bibfnamefont {T.}~\bibnamefont {H{\"a}ner}},
  \ and\ \bibinfo {author} {\bibfnamefont {M.}~\bibnamefont {Troyer}},\ }\href
  {\doibase 10.22331/q-2018-01-31-49} {\bibfield  {journal} {\bibinfo
  {journal} {Quantum}\ }\textbf {\bibinfo {volume} {2}},\ \bibinfo {pages} {49}
  (\bibinfo {year} {2018})}\BibitemShut {NoStop}%
\bibitem [{\citenamefont {Andreev}\ and\ \citenamefont
  {Racke}(2006)}]{Andreev2006}%
  \BibitemOpen
  \bibfield  {author} {\bibinfo {author} {\bibfnamefont {K.}~\bibnamefont
  {Andreev}}\ and\ \bibinfo {author} {\bibfnamefont {H.}~\bibnamefont
  {Racke}},\ }\href {\doibase 10.1007/s00224-006-1350-7} {\bibfield  {journal}
  {\bibinfo  {journal} {Theory Comput. Syst.}\ }\textbf {\bibinfo {volume}
  {39}},\ \bibinfo {pages} {929} (\bibinfo {year} {2006})}\BibitemShut
  {NoStop}%
\bibitem [{\citenamefont {Karypis}\ and\ \citenamefont
  {Kumar}(1998)}]{Karypis1998}%
  \BibitemOpen
  \bibfield  {author} {\bibinfo {author} {\bibfnamefont {G.}~\bibnamefont
  {Karypis}}\ and\ \bibinfo {author} {\bibfnamefont {V.}~\bibnamefont
  {Kumar}},\ }\href {\doibase 10.1137/S1064827595287997} {\bibfield  {journal}
  {\bibinfo  {journal} {SIAM J. Sci. Comput.}\ }\textbf {\bibinfo {volume}
  {20}},\ \bibinfo {pages} {359} (\bibinfo {year} {1998})}\BibitemShut
  {NoStop}%
\bibitem [{\citenamefont {Kleinberg}(2007)}]{Kleinberg2007}%
  \BibitemOpen
  \bibfield  {author} {\bibinfo {author} {\bibfnamefont {R.}~\bibnamefont
  {Kleinberg}},\ }in\ \href {\doibase 10.1109/INFCOM.2007.221} {\emph {\bibinfo
  {booktitle} {{{IEEE INFOCOM}} 2007}}}\ (\bibinfo  {publisher} {{IEEE}},\
  \bibinfo {year} {2007})\ pp.\ \bibinfo {pages} {1902--1909}\BibitemShut
  {NoStop}%
\bibitem [{\citenamefont {Jonckheere}\ and\ \citenamefont
  {Lohsoonthorn}(2004)}]{Jonckheere2004}%
  \BibitemOpen
  \bibfield  {author} {\bibinfo {author} {\bibfnamefont {E.}~\bibnamefont
  {Jonckheere}}\ and\ \bibinfo {author} {\bibfnamefont {P.}~\bibnamefont
  {Lohsoonthorn}},\ }in\ \href@noop {} {\emph {\bibinfo {booktitle}
  {Proceedings of the 2004 {{American Control Conference}}}}},\ Vol.~\bibinfo
  {volume} {2}\ (\bibinfo {year} {2004})\ pp.\ \bibinfo {pages} {976--981
  vol.2}\BibitemShut {NoStop}%
\bibitem [{\citenamefont {Chepoi}\ \emph {et~al.}(2012)\citenamefont {Chepoi},
  \citenamefont {Dragan}, \citenamefont {Estellon}, \citenamefont {Habib},
  \citenamefont {Vax{\`e}s},\ and\ \citenamefont {Xiang}}]{Chepoi2012}%
  \BibitemOpen
  \bibfield  {author} {\bibinfo {author} {\bibfnamefont {V.}~\bibnamefont
  {Chepoi}}, \bibinfo {author} {\bibfnamefont {F.~F.}\ \bibnamefont {Dragan}},
  \bibinfo {author} {\bibfnamefont {B.}~\bibnamefont {Estellon}}, \bibinfo
  {author} {\bibfnamefont {M.}~\bibnamefont {Habib}}, \bibinfo {author}
  {\bibfnamefont {Y.}~\bibnamefont {Vax{\`e}s}}, \ and\ \bibinfo {author}
  {\bibfnamefont {Y.}~\bibnamefont {Xiang}},\ }\href {\doibase
  10.1007/s00453-010-9478-x} {\bibfield  {journal} {\bibinfo  {journal}
  {Algorithmica}\ }\textbf {\bibinfo {volume} {62}},\ \bibinfo {pages} {713}
  (\bibinfo {year} {2012})}\BibitemShut {NoStop}%
\bibitem [{\citenamefont {{Eldredge et al.}}()}]{UsUpcoming}%
  \BibitemOpen
  \bibfield  {author} {\bibinfo {author} {\bibfnamefont {Z.}~\bibnamefont
  {{Eldredge et al.}}},\ }\href@noop {} {\ }\bibinfo {note} {{in
  preparation.}}\BibitemShut {Stop}%
\bibitem [{\citenamefont {Fowler}\ \emph {et~al.}(2004)\citenamefont {Fowler},
  \citenamefont {Devitt},\ and\ \citenamefont {Hollenberg}}]{Fowler2004}%
  \BibitemOpen
  \bibfield  {author} {\bibinfo {author} {\bibfnamefont {A.~G.}\ \bibnamefont
  {Fowler}}, \bibinfo {author} {\bibfnamefont {S.~J.}\ \bibnamefont {Devitt}},
  \ and\ \bibinfo {author} {\bibfnamefont {L.~C.~L.}\ \bibnamefont
  {Hollenberg}},\ }\href {http://arxiv.org/abs/quant-ph/0402196} {\bibfield
  {journal} {\bibinfo  {journal} {Quantum Inf. Comput.}\ }\textbf {\bibinfo
  {volume} {4}},\ \bibinfo {pages} {237} (\bibinfo {year} {2004})}\BibitemShut
  {NoStop}%
\bibitem [{\citenamefont {Hoyer}\ and\ \citenamefont
  {Spalek}(2005)}]{Hoyer2005}%
  \BibitemOpen
  \bibfield  {author} {\bibinfo {author} {\bibfnamefont {P.}~\bibnamefont
  {Hoyer}}\ and\ \bibinfo {author} {\bibfnamefont {R.}~\bibnamefont {Spalek}},\
  }\href {\doibase 10.4086/toc.2005.v001a005} {\bibfield  {journal} {\bibinfo
  {journal} {Theory Comput.}\ }\textbf {\bibinfo {volume} {1}},\ \bibinfo
  {pages} {81} (\bibinfo {year} {2005})}\BibitemShut {NoStop}%
\end{thebibliography}%
\end{document}